\documentclass[format=acmsmall, review=false, screen=true]{acmart}

\usepackage{graphicx}
\usepackage{amsmath}
\usepackage{amssymb}
\usepackage{amsthm}
\usepackage[noend]{algpseudocode}
\usepackage{epstopdf}
\usepackage{algorithm}
\usepackage{color}
\usepackage{hyperref}
\usepackage{subcaption}
\usepackage{multirow}
\newtheorem{prop}{Proposition}
\newtheorem{remark}{Remark}

\newtheorem{lemma}{Lemma}
\newtheorem{theorem}{Theorem}
\newtheorem{corollary}{Corollary}
\newtheorem*{theorem*}{\bf Theorem}

\begin{document}

\title{Network Cache Design under Stationary Requests: \\Exact Analysis and Poisson Approximation}

\author{Nitish K. Panigrahy}
\affiliation{
  \institution{University of Massachusetts Amherst }
  \state{Amherst, MA 01003, USA} 
\thanks{This research was sponsored by the U.S. ARL and the U.K. MoD under Agreement Number W911NF-16-3-0001 and by the NSF under Grant CNS-1617437. The views and conclusions contained in this document are those of the authors and should not be interpreted as representing the official policies, either expressed or implied, of the National Science Foundation, U.S. Army Research Laboratory, the U.S. Government, the U.K. Ministry of Defence or the U.K. Government. The U.S. and U.K. Governments are authorized to reproduce and distribute reprints for Government purposes notwithstanding any copyright notation hereon.  The authors also thank Dr. Bo Jiang for useful discussions on MMPP}
}
\email{nitish@cs.umass.edu}

\author{Jian Li}
\affiliation{
  \institution{University of Massachusetts Amherst }
  \state{Amherst, MA 01003, USA} 
}
\email{jianli@cs.umass.edu}

\author{Don Towsley}
\affiliation{
  \institution{University of Massachusetts Amherst }
  \state{Amherst, MA 01003, USA} 
}
\email{towsley@cs.umass.edu}

\author{Christopher V. Hollot}
\affiliation{
  \institution{University of Massachusetts Amherst }
  \state{Amherst, MA 01003, USA} 
}
\email{hollot@ecs.umass.edu}

\begin{abstract}
The design of caching algorithms to maximize hit probability has been extensively studied. In this paper, we associate each content with a utility, which is a function of either the corresponding content hit rate or hit probability. We formulate a cache optimization problem to maximize the sum of utilities over all contents under stationary and ergodic request processes. This problem is non-convex in general but we reformulate it as a convex optimization problem when the inter-request time (irt) distribution has a non-increasing hazard rate function.  We provide explicit optimal solutions for some irt distributions, and compare the solutions of the hit-rate based (HRB) and hit-probability based (HPB) problems. We formulate a reverse engineering based dual implementation of LRU under stationary arrivals. We also propose decentralized algorithms that can be implemented using limited information and use a  discrete time Lyapunov technique (DTLT) to correctly characterize their stability. We find that decentralized algorithms that solve HRB are more robust than decentralized HPB algorithms.  Informed by these results, we further propose lightweight Poisson approximate decentralized and online algorithms  that are accurate and efficient in achieving optimal hit rates and hit probabilities.
\end{abstract}

\maketitle

\section{Introduction}\label{sec:intro}

Caching plays a prominent role in networks and distributed systems for improving system performance.  Since the number of contents in a system is typically significantly larger than cache capacity, the design of caching algorithms typically focuses on maximizing the number of requests that can be served from the cache. Considerable research has focused on the analysis of caching algorithms using the metric of hit probability under the Independent Reference Model (IRM) \cite{aven87,dehghan16,jung03,che02,nitishjian17,jslv17,nitishjianfaheem17}.  However, \emph{hit rate} \cite{fofack12} is a more relevant performance metric in real systems. For example, pricing based on hit rate is preferable to that based on cache occupancy from the perspective of a service provider \cite{ma15}.   Furthermore, one goal of a service provider in designing hierarchical caches would be to minimize the internal bandwidth cost, which can be characterized with a utility function $U_i=-C_i(m_i),$ where $C_i(m_i)$ is the cost associated with miss rate $m_i$ for content $i.$  Therefore, we focus on the hit rate.

Recently there has been a tremendous increase in the demand for different types of content with different quality of service requirements;  consequently, user needs have become more heterogeneous.  In order to meet such challenges, content delivery networks need to incorporate service differentiation among different classes of contents and applications.  Though considerable literature has focused on the design of fair and efficient caching algorithms for content distribution, little work has focused on the provision of multiple levels of service in network and web caches.

Moreover, cache behaviors of different contents are strongly coupled by conventional caching algorithms such as LRU \cite{aven87,jslv17,garetto16}, which make it difficult for cache service providers to provide differential services.  In this paper, we focus on Time-to-Live (TTL) caches.  When a content is inserted into the cache due to a cache miss, a timer is set.  Timer value can differ for different contents.  All requests for a content before the expiration of its timer results in a cache hit, and the first request after the expiration of its timer yields a cache miss.   This ability to decouple the behaviors of different contents make the TTL policy an interesting alternative to more popular algorithms like LRU.  Moreover, the TTL policy has the capacity of mimicking the behaviors of many caching algorithms \cite{baccelli13}. %this will be described in detail in Section~\ref{sec:model}. %with a controllable parameter; this will be described in detail in Section~\ref{sec:model}. 
 
In this paper, we consider a utility-driven caching framework, where each content is associated with a utility.  Content is stored and managed in the cache so as to maximize the aggregate utility for all content.   A related problem has been considered in \cite{dehghan16}, where the authors formulated a \emph{Hit-probability Based Cache Utility Maximization} (HPB-CUM) framework under IRM. The objective is to maximize the sum of utilities under a cache capacity constraint when utilities are increasing, continuously differentiable, and strictly concave function of hit probability.  % A more general formulation for cache network under IRM is considered in \cite{nitishjianfaheem17}. 
\cite{dehghan16}, \cite{fofack14} and \cite{nitishjianfaheem17} characterized optimal TTL cache policies, and also proposed distributed cache management algorithms. Here,  we focus on utilities as functions of hit rates. 

%Again, \emph{hit rate} is a more natural performance metric from the perspective of content providers.  Since the utility function should not only capture the impact of hit probability, but also the amount of data arriving at the system, which is characterized by the hit rate.  Therefore, we focus on utility functions as functions of hit rates.  

While characterization of hit rate under IRM is valuable, real-world request processes exhibit changes in popularity and temporal correlations in requests \cite{zink08, cha07tube}. To account for them, in this paper, \emph{we consider a very general traffic model where requests for distinct contents are described by mutually independent stationary and ergodic point processes} \cite{baccelli13}.

\subsection{Contributions}
Our main contributions in this paper can be summarized as follows.
%The main challenges in network cache design under stationary requests and our contributions in addressing those challenges are listed below.

1) We formulate a \emph{Hit-rate Based Cache Utility Maximization} (HRB-CUM) framework for maximizing aggregate content utilities subject to an expected cache size constraint at the service provider. In general, \emph{HRB-CUM with TTL caches under general stationary request process is a non-convex optimization problem.} We develop a convex optimization problem for the case that the {\it inter-request time} (irt) distributions have non-increasing hazard rate. This is an important case since inter-request times are often highly variable. We also  formulate a reverse engineering based dual implementation of LRU in HRB-CUM framework under stationary arrivals.

2) We compare hit rate based approaches to hit probability based approaches when utilities come from a family of $\beta$-fair utility functions. We find that HRB-CUM and HPB-CUM are identical under the log utility function with $\beta=1.$ However, for $\beta<1,$ there exists a threshold such that HRB-CUM favors more popular contents over HPB-CUM, i.e., popular contents will be cached under HRB-CUM, where as the reverse behavior holds for $\beta>1.$

3) We propose decentralized algorithms that adapt to different stationary requests using limited information. We find that the corresponding decentralized algorithms for HRB-CUM are more robust and stable than those for HPB-CUM {with respect to (w.r.t.)} convergence rate.

4) We apply the discrete time Lyapunov technique (DTLT) can be used to correctly characterize the stability of decentralized algorithms across different scaling parameters.

5) 
{Inspired by the analysis of decentralized algorithms,  % we make the following contributions 
%\begin{itemize}
%\item \emph{
we further propose a lightweight Poisson approximate online algorithm where we apply the dual designed for the case of requests described by a Poisson process to a workload where requests are described by stationary request processes.  Such a solution does not involve solving any non-linear equations and hence is computationally efficient.  %We also use estimation techniques to approximate the request rate, which make the algorithm be online.} %} This will be described in detail in Section~\ref{sec:estimator}.
%\item \emph{We numerically show that this approximation is accurate in achieving exact hit rates and hit probabilities. }
%\item \emph{

In particular, we consider an $m$-state MMPP. We characterize its limiting behavior in terms of state transition rates. We find that when the transition rates both go to infinity, $m$-state MMPP is equivalent to a Poisson process, i.e., our Poisson approximation is exact. We numerically show that our approximation is accurate in achieving near optimal hit rates and hit probabilities by considering a $2$-state MMPP request arrival process.

%We find that when the transition rates both go to infinity, $2$-MMPP is equivalent to a Poisson process, i.e., and our Poisson approximation is exact.  
%%Similarly, when transition rates go to zero, $2$-MMPP can be treated as a weighted sum of two Poisson processes.   
%We numerically show that our approximation is accurate in achieving near optimal hit rates and hit probabilities.   Furthermore, we also theoretically characterize the limiting behavior of the general $m$-state MMPP.

This analysis provides significant insights in modeling real traffic with Poisson process and also verify the robustness and wide applicability of Poisson process. Finally, we perform a trace-driven simulation to compare the performance of proposed Poisson approximate online algorithm to that of conventional caching policies, including LRU, FIFO and RANDOM.

\subsection{Related Work and Organization}
\noindent{\textit{Network Utility Maximization:}} Utility functions have been widely used in the performance analysis of computer networks. Since Kelly's seminal work \cite{kelly97, kelly98}, a rich literature uses network utility maximization problem in the analysis of throughput maximization, dynamic allocation, network routing etc and we do not attempt to provide a detailed overview here.

\noindent{\textit{Time-To-Live Caches:}}  {TTL caches have been employed in the Domain Name System (DNS) since the early days of Internet \cite{jung03}.  More recently, it has gained attention due to the ease by which it can be analyzed and can be used to model %fact that a simple and tractable analysis can be modeled to mimic
  the behaviors of caching algorithms such as LRU.}  The TTL cache has been shown to provide accurate estimates of the performance of large caches, as first introduced for LRU under IRM \cite{fagin77, che02} through the notion of \emph{cache characteristic time}. It has been further generalized to other settings \cite{berger14,fofack12, gast16, garetto16}. The accuracy of the TTL cache approximation of LRU is theoretically justified under IRM \cite{berger14} and stationary processes \cite{bo17}. %, and numerically verified under renewal processes \cite{garetto16}.  
A recent paper \cite{ferragut16} has tackled a similar problem close to ours, which focuses on maximizing hit probabilities under DHR demands.  Instead, we focus on optimizing the total utilities of cache contents.

{The paper is organized as follows. The next section contains some technical preliminaries. We formulate the HRB-CUM and HPB-CUM under general stationary requests in Section~\ref{sec:utility}, and present some specific inter-request processes under which HRB-CUM and HPB-CUM become convex optimization problems in Section~\ref{sec:distr}.  We compare their performance both theoretically and numerically in Section~\ref{sec:dist}.  We develop decentralized algorithms and give its performance evaluations in Section~\ref{sec:online} and characterize its stability performance in Section~\ref{sec:stabdec}. We present Poisson approximate online algorithms in Section~\ref{sec:expoapprox}. We perform a trace-driven simulation in Section~\ref{sec:tracesim}. We conclude the paper in Section~\ref{sec:concl}. }

\section{Technical Preliminaries}\label{sec:model}
We consider a cache of size $B$ serving $n$ distinct contents each with unit size.
\vspace{-0.1in}
\subsection{Content Request Process}\label{sec:trafficmodel}

%To characterize the performance of a caching algorithm, it is necessary to specify the content request processes for individual content.
 In this paper, the request processes for distinct contents are described by mutually independent stationary and ergodic simple point process as \cite{baccelli13,bo17}.   Our model generalizes the widely used Independence Reference Model (IRM) \cite{aven87}, where requests are described by Poisson processes.

Let $\{a_{ik}, k\in\mathbb{Z}\}$ represent successive request times to content $i=1, \cdots, n.$  Let $X_{ik}=a_{ik}- a_{i(k-1)}$ denote the inter-request times for a particular content $i$.  We consider $\{X_{ik}\}_{k \ge 1}$ to be a stationary point process with cumulative irt distribution functions (c.d.f.) satisfying \cite{baccelli13}
\begin{align}\label{eq:cdf-renewal}
F_i(t) = \mathbb{P}(X_{ik} \leq t), \;i=1, \cdots, n. 
\end{align}
For example, $F$ is a mixture of $l$ exponential distributions for an $l$-state MMPP.

The mean request rate $\mu_i$ for content $i$ is then given by
\begin{equation}
\mu_i = \frac{1}{\mathbb{E}[X_{ik}]} = \frac{1}{\int_0^\infty (1 - F_i(t)) dt}.
\end{equation}

Denote by $\hat{F}_i(t)$ the c.d.f. of the age associated with the irt distribution for content $i,$ satisfying (\cite{baccelli13})
\begin{align}\label{eq:age-distribution}
\hat{F}_i(t)=\mu_i\int_{0}^t (1-F(x))dx.
\end{align}

It is known (\cite{baccelli13}) that the popularity (requested probability) of content $i$ satisfies 
\begin{align}\label{eq:popularity}
p_i=\mu_i/\mu,
\end{align}
with $\mu=\sum_{i=1}^n\mu_i$.

In our work, we consider various irt distributions, including exponential, Pareto, hyperexponential and MMPP.

\subsection{Content Popularity}
Whereas our analytical results hold for any popularity law,  in our numerical studies we will use the Zipf distribution as this distribution has been frequently observed in real traffic measurements \cite{cha09}. Under the Zipf distribution, the probability of requesting the $i$-th most popular content is $A/i^{\alpha}$, where $\alpha$ is the Zipf
 parameter depending on the application \cite{fricker12}, and $A$ is the normalization factor satisfying $\sum_{i=1}^n p_i=1.$

\subsection{TTL Caches}
In a TTL cache, each content $i$ is associated with a timer $t_i$.  When content $i$ is requested, there are two cases: (i) if the content is not in the cache (miss), then content $i$ is inserted into the cache and its timer is set to $t_i;$ (ii) if the content is in the cache (hit), then the timer associated with content $i$ is reset. The timer decreases at a constant rate and the content is evicted once its timer expires. This is referred to as a \emph{Reset TTL Cache.}  We can control the hit probability of each content by adjusting its timer value.

Denote the \emph{hit rate} and \emph{hit probability} of content $i$ as $\lambda_i$ and $h_i,$ respectively, then from the analysis of previous work \cite{fofack14performance}, the hit probability and hit rate for a {reset TTL cache} 
  can be computed as 
\begin{align}\label{eq:reset-hit}
&h_i = F_i(t_i),\quad \lambda_i = \mu_iF_i(t_i), 
\end{align}
respectively, where requests for content $i$ follow a request process as described in Section~\ref{sec:trafficmodel}.

Let $h_i^{\text{in}}$ be the \emph{time-average probability} that content $i$ is in the cache ({i.e., \emph{occupancy probability}}), then we have \cite{garetto16, ferragut16}
\begin{align}\label{eq:occup}
h_i^{\text{in}}= \hat{F}_i(t_i).
\end{align}
In particular, our model reduces to classical IRM when the inter-request time are exponentially distributed, i.e., Poisson arrival process \cite{dehghan16}, with $F_i(t_i) = 1 - e^{-\mu_it_i}$ and {$h_i=h_i^{\text{in}}$, based on the PASTA property \cite{MeyTwe09}.}

\subsection{Utility Function and Fairness} 
Utility functions capture the satisfaction perceived by a content provider. %by the user after being served a content. Different utility functions define different fairness properties.  
Here, we focus on the widely used $\beta$-fair utility functions \cite{srikant13} given by 
\begin{equation}\label{eq:utility}
    U_i(x)=
\begin{cases}
     w_i\frac{x^{1-\beta}}{1-\beta},& \beta\geq0, \beta\neq 1;\\
     w_i\log x,& \beta=1,
\end{cases}
\end{equation}
where $w_i>0$ denotes a weight associated with content $i$.

\section{Cache Utility Maximization}\label{sec:utility}

In this section, we formulate a utility maximization problem for cache management (CUM). In particular, we consider a formulation based on hit rate (HRB-CUM)\footnote{{From this section onwards, we will use superscript $r$ and $p$ to distinguish corresponding hit rates, hit probabilities and occupancy probabilities under HRB-CUM and HPB-CUM, respectively.}}.  As mentioned in the introduction, one can also formulate a problem based on hit probability. The formulation for HPB-CUM can be found in Appendix~\ref{sec:hpb-cum}. %\cite{nitishjian-mobihoc18-tech}.

%\subsection{HRB-CUM}\label{sec:hrb-cum}
We are interested in optimizing the sum of utilities over all contents, 
\begin{subequations}\label{eq:hrbcum}
\vspace{-0.1in}
\begin{align}
 \max_{\{t_1,\cdots,t_n\}} \quad&\sum_{i=1}^n U_i(\lambda_i^r(t_i))\displaybreak[0]\\ 
 \text{s.t.} \quad&\sum_{i=1}^n h_i^{r, \text{in}}(t_i) \leq B,\displaybreak[1] \label{eq:hrbcum-constraint1}\\
 & 0\leq h_i^{r, \text{in}}(t_i)\leq 1,\displaybreak[2] \label{eq:hrbcum-constraint2}\\
 &0\leq h_i^r(t_i)=\lambda_i^r(t_i)/\mu_i\leq 1. \label{eq:hrbcum-constraint3}
\end{align}
\end{subequations}
{Constraint~(\ref{eq:hrbcum-constraint1}) ensures that the \emph{expected} number of contents does not exceed the cache size.  ~(\ref{eq:hrbcum-constraint2}) and~(\ref{eq:hrbcum-constraint3}) are inherent constraints on occupancy probability $h_i^{r, \text{in}}(t_i)=\hat{F}_i(t_i)$ and hit probability $h_i^r(t_i)=\lambda_i^r(t_i)/\mu_i=F_i(t_i),$ respectively.}  Although the objective function is concave, ~(\ref{eq:hrbcum}) is not a convex optimization problem w.r.t. timer $t_i$, since the feasible set is not convex.   See  Appendix~\ref{appa} for details.  %\cite{nitishjian-mobihoc18-tech} for details. %We omit the argument due to space limitations, which is available in \cite{nitishjian-mobihoc18-tech}.  
 Hence,~(\ref{eq:hrbcum}) is hard to solve in general.

{In the following, we will show that~(\ref{eq:hrbcum}) can be reformulated as a convex problem.} From~(\ref{eq:reset-hit}), we have $t_i=F_i^{-1}(\lambda_i^r/\mu_i),$ with $F_i^{-1}(\cdot)$ being the inverse function of $F_i(\cdot)$.  Then by~(\ref{eq:occup}),
\begin{align}\label{eq:newfung}
h_i^{r, \text{in}}=\hat{F}_i(F_i^{-1}(\lambda_i^r/\mu_i)) \triangleq g_i(\lambda_i^r/\mu_i).
\end{align}
From~(\ref{eq:age-distribution}), we know there exists a one-to-one correspondence between $\hat{F}_i$ and $F_i$, hence $g_i(\cdot)$ exists. 
Therefore, ~(\ref{eq:hrbcum}) can be reformulated as follows
\begin{subequations}\label{eq:hrbcum2}
\begin{align}
\text{\bf{HRB-CUM}:}\quad \max_{\{\lambda_1,\cdots,\lambda_n\}}\quad & \sum_{i=1}^n U_i(\lambda_i^r) \displaybreak[0]\label{eq:hrbcum2-constraint0}\\
\text{s.t.}\quad&\sum_{i=1}^n g_i(\lambda_i^r/\mu_i)\leq B,\displaybreak[1] \label{eq:hrbcum2-constraint1}\\
&0\leq \lambda_i^r/\mu_i\leq 1. \label{eq:hrbcum2-constraint2}
\end{align}
\end{subequations}
Again \eqref{eq:hrbcum2-constraint1} is a constraint on average cache occupancy.  Note that we can obtain HPB-CUM from~(\ref{eq:hrbcum2}) by replacing $\lambda_i^r$ by $h_i^p$ in~(\ref{eq:hrbcum2-constraint0}), and $\lambda_i^r/\mu_i$ by $h_i^p$ in~(\ref{eq:hrbcum2-constraint1}) and~(\ref{eq:hrbcum2-constraint2}), respectively. 

\begin{remark}
Let the buffer size $B(n)$ be a function of $n$ and let $\epsilon$ be a constant greater than zero, $\epsilon > 0$.  If $\epsilon^2B(n) = \omega (1)$, then the probability  that the number of cached contents exceeds $B(n) (1 + \epsilon)$ decreases exponentially as a function of $\epsilon ^2B(n)$, \cite{dehghan16}.  Thus, we can let $\epsilon$ go to zero while allowing $B$ to grow with $n$.   The practical import is that the buffer  can be sized as $B(1+\epsilon)$ while the optimizer works with $B$.  Hence the fraction of buffer used, $\epsilon/(1+\epsilon)$, to protect against violations goes to zero as $n$ gets large.
\end{remark}

%\begin{remark}
%It can be shown that if $n\to\infty$ and $B$ grows in a sub-linear manner, the probability of violating the target cache size $B$ becomes negligible. A cache of size $B(1+\epsilon)$ with $\epsilon>0$ can be provided where the portion $B$ can be used to solve the optimization problem and $\epsilon$ to handle capacity constraint violations \cite{dehghan16}.
%\end{remark}

A related problem has been formulated in \cite{ferragut16}, where the authors formulated the optimization problem as a function of  $h_i^{r, \text{in}}$.  However, such a formulation may not be suitable for designing decentralized algorithms since we need a closed form expression for $\hat{F}_i^{-1}$. More details on the advantages of our formulation over \cite{ferragut16} in decentralized algorithm design are given in Section~\ref{sec:online}.  Furthermore, \cite{ferragut16} only considers linear utilities while we aim to characterize the impact of different utility functions on optimal TTL policies. %\red{[More Details on relation with \cite{ferragut16}]}}

Now we consider the convexity of (\ref{eq:hrbcum2}).
\begin{lemma}\label{lm:hazard}
Let $F_i(t)$ and $\hat{F}_i(t)$ be the c.d.f. and age distribution for the request process of content $i,$ given in~(\ref{eq:cdf-renewal}) and~(\ref{eq:age-distribution}), respectively.  Denote the corresponding density function as $f_i(t)$.  Let $\zeta_i(t)$ be the hazard rate function associated with $F_i(t)$, given as 
\begin{align}
\zeta_i(t) =\frac{f_i(t)}{1-F_i(t)},\quad t\in[0, F_i^{-1}(1)].
\end{align}
Then
\begin{align}\label{eq:dxr1}
\frac{\partial g_i(\lambda_i^r/\mu_i)}{\partial \lambda_i^r}=\frac{1-F_i(F_i^{-1}(\lambda_i^r/\mu_i))}{f_i(F_i^{-1}(\lambda_i^r/\mu_i))}=\frac{1}{\zeta_i( F_i^{-1}(\lambda_i^r/\mu_i))}.
\end{align}
\end{lemma}
The proof can be found in Appendix~\ref{appa}.% \cite{nitishjian-mobihoc18-tech}.

From~(\ref{eq:dxr1}), it is clear that the behavior of the hazard rate function plays a prominent role in solving~(\ref{eq:hrbcum2}).  In particular, if $\zeta_i(t)$ is non-increasing (DHR), then by~(\ref{eq:dxr1}), $g^\prime(\lambda_i^r/\mu_i)$ is non-decreasing in $\lambda_i^r.$ Therefore, the feasible set in~(\ref{eq:hrbcum2}) is convex.  {Since the objective function is strictly concave and continuous, ~(\ref{eq:hrbcum2}) is a convex optimization problem, and an optimal solution exists.}  In this paper, we mainly focus on the case that $\zeta_i(t)$ is DHR ,  and refer the interested reader to \cite{ferragut16} for discussions of other cases. We will discuss several widely used distributions satisfying DHR in Section~\ref{sec:distr}.  %In particular, we will also consider a uniform distribution, which has an increasing hazard rate, under which~(\ref{eq:hrbcum2}) is not a convex optimization problem, but we will show that an efficient approximate solution exists under linear or quadratic utilities.

In the following, we focus on the case that $\zeta_i(t)$ is DHR, i.e.,~(\ref{eq:hrbcum2}) is a convex optimization problem.  We write the Lagrangian function as
\begin{align}\label{eq:lagrangianhr}
\mathcal{L}^r(\boldsymbol \lambda^r, \eta^r)=\sum_{i=1}^n U_i(\lambda_i^r)-\eta^r\left[\sum_{i=1}^n g_i(\lambda_i^r/\mu_i)-B\right],
\end{align}
where $\eta^r$ is the Lagrangian multiplier {and $\boldsymbol \lambda^r=(\lambda_1^r,\cdots,\lambda_n^r).$}  We first consider complementary slackness conditions \cite{srikant13}, i.e., $\eta^r[\sum_{i=1}^n g_i(\lambda_i^r/\mu_i)-B] = 0.$ It is clear that $\eta^r \ne 0$, otherwise $\mathcal{L}^r(\boldsymbol \lambda^r, \eta^r)$ is maximized at $\lambda_i^r = \mu_i,$ $\forall i$. Therefore, $\sum_{i=1}^n g_i(\lambda_i^r/\mu_i) = n \nleq B$, which does not satisfy the constraint.

To achieve the maximum of $\mathcal{L}^r(\boldsymbol \lambda^r, \eta^r),$ its derivative w.r.t. $\lambda_i^r$ for $i=1, \cdots, n,$ should satisfy
\begin{align}
\frac{\partial \mathcal{L}^r(\boldsymbol \lambda^r, \eta^r)}{\partial \lambda_i^r}
&=U_i^\prime(\lambda_i^r)-\frac{\eta^r}{\mu_i} g^\prime_i(\lambda_i^r/\mu_i) = 0,
\end{align} 
i.e.,  
\begin{align}\label{eq:ydef}
\eta^r &=\frac{\mu_i U_i^\prime(\lambda_i^r)}{g_i^\prime(\lambda_i^r/\mu_i)}\triangleq y_i(\lambda_i^r/\mu_i),
\end{align}
where $y_i(\cdot)$ is a continuous and differentiable function on $[0, 1].$ Hence we have 
\begin{align}
\lambda_i^r=\begin{cases}
                      \mu_i y_i^{-1}(\eta^r), & 0\leq y_i^{-1}(\eta^r)\leq 1,\\
                      \mu_i, & y_i^{-1}(\eta^r) > 1,\\
                      0, &y_i^{-1}(\eta^r) <0.
                      \end{cases}
\end{align}
Again, by the cache capacity constraint, $\eta^r$ is the solution of the following fixed-point equation
\begin{align}\label{eq:capacity-constraint-hrb}
\sum_{i=1}^n g_i(\lambda_i^r/\mu_i) &=\sum_{i=1}^ng_i(y_i^{-1}(\eta^r))=B.
\end{align}
As discussed earlier, our optimization framework holds for TTL caches. Once we determine $\eta^r$ from~(\ref{eq:capacity-constraint-hrb}), the timer can be computed as
\begin{align}
 t_i = F_i^{-1}(y_i^{-1}(\eta^r)), \quad i=1, \cdots, n,
 \end{align}
 then by~(\ref{eq:reset-hit}), the hit probability and hit rate for reset TTL cache under HRB-CUM is
\begin{align}\label{eq:hpbopt}
h_i^r=y_i^{-1}(\eta^r), \quad &\lambda_i^r=\mu_iy_i^{-1}(\eta^r). 
\end{align}

\begin{remark}
Note that the above solution only requires the knowledge of the irt distribution. Dependencies among inter-request times do not affect the solution.
\end{remark}

\subsection{Reverse Engineering}
Many conventional caching policies such as LRU and FIFO can be duplicated by appropriately choosing utility functions.  Dehghan et.al. \cite{dehghan16} first reverse engineered classical replacement policies in a HPB-CUM framework under IRM.  Similar results hold for HRB-CUM framework.  Below we present one such formulation of utility functions to mimic the behavior of LRU in a HRB-CUM framework.

When the request arrivals for each content follow a stationary process, the hit rates for LRU caches can be expressed as $\lambda_i^r = \mu_iF_i(T)$, where $T$ is the characteristic time obtained by solving the fixed point equation $\sum_{i=1}^n \hat{F}_i(T) = B$ \cite{garetto16}. Applying similar reverse engineering techniques as adopted in \cite{dehghan16}, we can express $T$ as a decreasing function of the dual variable $\eta^r.$ More precisely, taking $T = 1/\eta^r$ and combining with \eqref{eq:ydef}, we get
\begin{align}\label{eq:lrurev}
\eta^r &=\frac{\mu_i U_i^\prime(\mu_iF_i(1/\eta^r))}{g_i^\prime(F_i(1/\eta^r))}.
\end{align}
\noindent Substituting $x_i^r = \mu_iF_i(1/\eta^r)$ and integrating both sides, we obtain
\begin{align}\label{eq:lrurev2}
U_i(x^r) = \int \frac{g_i^\prime(x^r/\mu_i)}{\mu_iF_i^{-1}(x^r/\mu_i)}dx^r.
\end{align}

\noindent Note that when the request arrival process is Poisson, $g_i^\prime(x) = 1$ and $F_i^{-1}(x) = (-1/\mu_i)\log(1-x)$. Substituting in \eqref{eq:lrurev2}, we get $U_i(x^r) = \mu_i \textbf{li}[\mu_i(1-x^r)]$, where $\textbf{li}(x) = \int_0^x dt/\log t.$

%\subsection{HPB-CUM}\label{sec:hpb-cum}
%Following a similar argument in Section~\ref{sec:hrb-cum},  we can formulate the following hit probability based optimization problem
%\begin{align}\label{eq:hpbcum2}
%\text{\bf{HPB-CUM:}}\quad \max_{0\leq h_i^p\leq 1} \sum_{i=1}^n U_i(h_i^p),\quad\text{s.t.} \sum_{i=1}^n g_i(h_i^{p}) \leq B,
%\end{align}
%
%The Lagrangian function can be written as 
%\begin{align}
%\mathcal{L}^p(\boldsymbol h^p, \eta^p)=\sum_{i=1}^n U_i(h_i^p)-\eta^p\left[\sum_{i=1}^n g_i(h_i^{p})-B\right],
%\end{align}
%where $\eta^p$ is the Lagrangian multiplier {and $\boldsymbol h^p=(h_1^p,\cdots,h_n^p)$}.  Similarly, the derivative of $\mathcal{L}^p(\boldsymbol h^p, \eta^p)$ w.r.t. $h_i^p$ for $i=1, \cdots, n,$ should satisfy the following condition so as to achieve its maximum
%\begin{align}\label{eq:funv}
%\eta^p=U_i^\prime(h_i^p)/g^\prime_i(h_i^p) \triangleq v_i(h_i^p),
%\end{align}
%where $v_i(\cdot)$ is a continuous and differentiable function on $[0, 1],$ i.e., there exists a one-to-one mapping between $\eta^p$ and $h_i^p$ if $0\le v_i^{-1}(\eta^p)\le1.$  Again, by the cache capacity constraint, we can compute $\eta^p$ through the following fixed-point equation
%\begin{align}\label{eq:capacity-constraint-hpb}
%\sum_{i=1}^n g_i(h_i^{p}) &= \sum_{i=1}^ng_i(v_i^{-1}(\eta^p)) = B.
%\end{align}
%Finally, given $\eta^p$, the timer, hit probability, and hit rate are
%\begin{align}
% t_i = F_i^{-1}( v_i^{-1}(\eta^p)), \quad h_i^p= v_i^{-1}(\eta^p), \quad &\lambda_i^p=\mu_i v_i^{-1}(\eta^p), \nonumber\\
%& i=1,\cdots,n.
%\end{align}
%

\section{Specific Inter-request Time Distributions}\label{sec:distr}
\begin{table*}[]
\resizebox{\columnwidth}{!}{
\begin{tabular}{|c| c| c| c| c| c|}
\hline
 {\bf Processes}&{\bf Parameters}&${\bf F_i(t)}$&$\bf{\hat{F}_i(t)}$&${\bf g_i(x)}$&{\bf Optimal Solution}\\
\hline
Process with&$\mu_i$: rate&$1-e^{-\mu_it}$&$1-e^{-\mu_it}$&x&centralized: convex solver\\
exponential irt&&&&&decentralized: Dual\\
\hline
Process with a&$k_i:$ shape, $\sigma_i$: scale&$1 -(1+\frac{k_it}{\sigma_i})^{-\frac{1}{k_i}}$&$1 - (1 + {\frac{k_it}{\sigma_i}})^{\frac{k_i-1}{k_i}}$&$1 - (1 -x)^{1-k_i}$&centralized: convex solver\\
Generalized Pareto irt &$\theta_i(=0):$ location&&&&decentralized: Dual + fixed point\\
\hline
Process with &$l$: order&$1-\sum\limits_{j=1}^{l}p_{ji}e^{-\theta_{ji}t}$&$\mu_i\sum\limits_{j=1}^{l}\frac{p_{ji}}{\theta_{ji}}(1-e^{-\theta_{ji}t})$&No closed form&centralized: No exact solution\\
hyper-exponential&$p_{ji}:$ phase probability &&&&decentralized: Dual + fixed point\\
irt&$\theta_{ji}:$ phase rate&&&&\\
\hline
%Weibull&$k_i:$ shape, $\theta_i$: scale&$1-e^{-(t/\theta_i)^{k_i}}$&$\mu_i\theta_i\int_{0}^{t/\theta_i}e^{-x^{k_i}}\;dx$&No closed form&centralized: No exact solution\\
%&&&&&distributed: Dual + fixed point\\
%\hline
Process with&$\theta_{1i}, \theta_{2i}:$ arrival rate&$1-\sum\limits_{j=1}^{2}q_{ji}e^{-u_{ji}t}$&$\mu_i\sum\limits_{j=1}^{2}\frac{q_{ji}}{u_{ji}}(1-e^{-u_{ji}t})$&No closed form&centralized: No exact solution\\
$2$-MMPP Process&$r_{12i}, r_{21i}:$ Tran. rate&&&&decentralized: Dual + fixed point\\
&$q_{1i}, q_{2i}, u_{1i}, u_{2i}:$ \eqref{eq:mmppparams}&&&&\\
\hline
%Uniform&$b_i:$ maximum value&$t/b_i,\quad 0 \leq t \leq b_i$&$2\left(t/b_i\right) - \left(t/b_i\right)^2$&$2x - x^2$&centralized: Nonconvex\\
%&&&&&SDR: for QCQP\\
%\hline
\end{tabular}
}
%\vspace{0.25cm}
\caption{{Properties of specific traffic distributions.  The final column entitled ``Optimal Solution" will be discussed in Section~\ref{sec:online},  where ``centralized" is obtained by solving~(\ref{eq:hrbcum2}) and ``decentralized" is obtained through designing decentralized algorithms.}}
\label{tbltraff}
\vspace{-0.35in}
\end{table*}

In this section, we investigate irt distributions that are DHR such that~(\ref{eq:hrbcum2}) is a convex optimization problem.  For ease of exposition, we relegate detailed explanations of different parameters and derivations to Appendix~\ref{appb}. %Appendix~\ref{appb}. \cite{nitishjian-mobihoc18-tech}. 
{The properties of these distributions are presented in Table \ref{tbltraff}.}

First, for both exponential and generalized Pareto distributions, we have explicit forms for $g_i(\cdot).$  Thus the optimization problem in~(\ref{eq:hrbcum2}) can be solved in both centralized and distributed manner.   {However, we will see that the distributed dual algorithm for generalized Pareto distribution involves solving a fixed point equation, which exhibits high computational complexity.  This will be further discussed in Section~\ref{sec:online}.}

Second, for the hyperexponential distribution, {we were not able to obtain an explicit form for $F_i^{-1}(\cdot)$, and hence not for $g_i(\cdot)$ from~(\ref{eq:newfung}).}  Therefore, it is difficult to obtain an exact solution of~(\ref{eq:hrbcum2}) through a centralized solver.  %{Similarly, there is no explicit form of $\hat{F}_i(t)$ for Weibull distribution.  Therefore, it is difficult to obtain an exact solution of~(\ref{eq:hrbcum2}) through centralized solver for these two distributions besides a few special cases.}  
 {However,  we will see that the corresponding problems of~(\ref{eq:hrbcum2})  can be solved in a distributed fashion by solving fixed point equations without the need for an explicit form of $g_i(\cdot)$. Again, this is further discussed in Section~\ref{sec:online}.}

An important class of processes that give rise to hyperexponential irt distributions are Markov modulated Poisson processes (MMPP). MMPP is a doubly stochastic Poisson process with request rate varying according to a Markov process.  MMPPs have been widely used to model request processes with bursty arrivals, which occur in various application domains such as web caching \cite{rodriguez01} and Internet traffic modeling \cite{paxson95}.  % and queuing theory \cite{salvador03}. } 
%\red {[Motivation for MMPP can be moved to introduction: Markov Modulated Poisson Process (MMPP) is a doubly stochastic Poisson process which is used to model request processes with irregular burst of arrivals. MMPP is one of the most widely used model to characterize the request arrival process in various domains such as web caching, internet traffic modeling and queuing theory \cite{paxson95},\cite{rodriguez01}, \cite{salvador03}. Hence, we wish to analyze the performance of HRBCUM and HPBCUM when request process is described by an MMPP distribution.]} 
 {We consider request processes following two state MMPPs.  Without loss of generality (W.l.o.g.), denote the states as $1$ and $2.$ The transition rate for content $i$ from state $1$ to $2$ is $r_{12i}$, and $r_{21i}$ vice versa.  Arrivals for content $i$ at states $1$ and $2$ are described by Poisson processes with rates $\theta_{1i}$ and $\theta_{2i}$, respectively. Then the steady state distribution satisfies $\boldsymbol \pi_i = [\pi_{1i},\pi_{2i}] = [r_{21i}/(r_{12i}+r_{21i}),r_{12i}/(r_{12i}+r_{21i})]$.  Denote $\boldsymbol p_i = [p_{1i},p_{2i}] = [\frac{\theta_{1i}r_{21i}}{\theta_{1i}r_{21i}+\theta_{2i}r_{12i}}, \frac{\theta_{2i}r_{12i}}{\theta_{1i}r_{21i}+\theta_{2i}r_{12i}}].$  We assume that the initial probability vector for this $2$-MMPP is chosen according to $\boldsymbol p$.  Under this assumption, the inter-request times of this $2$-MMPP are described by a second order hyperexponential distribution with parameters satisfying (\cite{kang95})
%\begin{small}
\begin{align}\label{eq:mmppparams}
    &u_{1i} = (\theta_{1i}+\theta_{2i}+r_{12i}+r_{21i}-\delta_i)/2, \quad u_{2i} = (\theta_{1i}+\theta_{2i}+r_{12i}+r_{21i}+\delta_i)/2, \nonumber\displaybreak[1]\\
    &q_{1i} = \frac{\theta_{2i}^2r_{12i}+\theta_{1i}^2r_{21i}}{(\theta_{1i}r_{21i}+\theta_{2i}r_{12i})(u_{1i}-u_{2i})} - \frac{u_{2i}}{u_{1i}-u_{2i}},   \quad q_{2i} = 1-q_{1i},\nonumber\displaybreak[3]\\
    &\delta_i = \sqrt{(\theta_{1i}-\theta_{2i}+r_{12i}-r_{21i})^2+4r_{12i}r_{21i}}. 
\end{align} 
%\end{small}
\noindent{Again,  it is difficult to obtain exact solution of~(\ref{eq:hrbcum2}) through centralized solver for a two state MMPP.} 

%\blue{Finally, for uniform distribution, {it is clear that $g_i$ is a concave function and then HRB-CUM~(\ref{eq:hrbcum2}) is non-convex.}  However, for linear and quadratic utilities, HRB-CUM~(\ref{eq:hrbcum2}) is a quadratic constraint quadratic optimization problem (QCQP) \cite{luo10}.   A semidefinite relaxation of this QCQP can yield computationally efficient approximate solution \cite{luo10}.  We do not pursue this further in this paper.}

\begin{remark}
Note that obtaining optimal solution for inter-request times characterized by a hyperexponential distribution has a significant advantage since many {heavy-tailed distributions} can be well approximated by a hyperexponential distribution \cite{feldmann97}.  Similarly for an $m$-state MMPP, $F_i(\cdot)$ is a mixture of $m$ exponential distributions.  {In particular, we will} discuss the algorithm for obtaining optimal solution for a two-state MMPP in Section \ref{sec:expoapprox}.  Furthermore, we also consider Weibull distribution.  Due to space limits, we relegate its properties to Appendix~\ref{appb}. %Appendix~\ref{appb}.% \cite{nitishjian-mobihoc18-tech}.
\end{remark}

\section{Performance Comparison}\label{sec:dist}

Different utility functions define different fairness properties.  In this section, we analytically compare the performance of HRB-CUM and HPB-CUM under different utility functions and request arrival processes considered in Section~\ref{sec:distr}.  We omit proofs in this section and relegate them to Appendix~\ref{appc}.%Appendix~\ref{appc}.%\cite{nitishjian-mobihoc18-tech}. 

\subsection{Identical Distributions}
Assume that all contents have the same request arrival process, i.e., $F_i(\cdot)=F(\cdot)$ for all $i,$ then we have $\hat{F}_i(\cdot)=\hat{F}(\cdot),$ $g_i(\cdot)=g(\cdot)$ and $\mu_i=\mu$ for all $i.$
\begin{theorem}\label{thm:identical}
Under identical stationary request processes, the solutions of HRB-CUM and HPB-CUM are equivalent.
\end{theorem}

%Further assume that all contents have the same utility function, i.e., $U_i(\cdot)= U(\cdot)$, for all $i.$ From~(\ref{eq:funv}), we know $v_i^{-1}=v^{-1}$ for all $i.$ Hence $h_i^p = h^p$ for all $i.$
Further assume that all contents have the same utility function, i.e., $U_i(\cdot)= U(\cdot)$, for all $i.$ From~(\ref{eq:ydef}), we know $y_i^{-1}=y^{-1}$ for all $i.$ Hence $\lambda_i^r = \lambda^r$ for all $i.$
 Therefore, from~(\ref{eq:capacity-constraint-hrb}), 
\begin{align}
&\sum_{i=1}^{n}g_i(\lambda_i^r/\mu_i) = ng(\lambda^r/\mu) = B, \quad\text{i.e.},\quad \lambda^r=\mu g^{-1}\left(B/n\right).
\end{align}

\subsection{$\beta$-fair Utility Functions}
We divide the set of $\beta$-fair utility functions into two subsets according to whether $\beta=1$ or $\beta\ne1.$ Consider the case that $\beta=1$ in~(\ref{eq:utility}),  i.e., $U_i(x) = w_i\log x$, then $U_i^\prime(x) = w_i/x.$  
\begin{theorem}\label{thm:log}
The solutions of HRB-CUM and HPB-CUM are identical under $\beta$-utility function with $\beta=1.$
\end{theorem}

{In the remainder of this section, we consider $\beta$-fair utility functions with $\beta>0$ and $\beta\neq 1.$} We compare the optimal hit probabilities  $h^r_i,$ $h^p_i$ and hit rates $\lambda^r_i$, $\lambda_i^p,$ under HRB-CUM and HPB-CUM for different weights $w_i$. W.l.o.g., we assume arrival rates satisfy $\mu_1\geq\cdots\geq\mu_n,$ such that content popularities satisfy $p_1\geq\cdots\geq p_n,$ {where $p_i=\mu_i/\mu$ and $\mu=\sum_i \mu_i.$}

\subsubsection{Poisson Request Processes}
With the Lagrangian method, we easily obtain the optimal hit rate $\lambda^r_i$ and hit probability $h^r_i$ under HRB-CUM for $\beta > 0$ and $\beta\neq 1,$ 
\begin{align}\label{eq:hit-rate-prob}
\lambda^r_i = \frac{w_i^{1/\beta} \mu_i^{1/\beta}}{\sum_{j} w_j^{1/\beta} \mu_j^{1/\beta - 1}}B,
\quad h^r_i = \frac{w_i^{1/\beta} \mu_i^{1/\beta - 1}}{\sum_{j} w_j^{1/\beta} \mu_j^{1/\beta - 1}}B.
\end{align} 
From \cite{dehghan16}, the corresponding optimal hit rate and hit probability under HPB-CUM are $\lambda_i^p = \frac{w_i^{1/\beta}\mu_i}{\sum_{j} w_j^{1/\beta}}B$ and  $h^p_i = \frac{w_i^{1/\beta}}{\sum_{j} w_j^{1/\beta}}B,$ respectively.   %Consider the Zipf popularity distribution with parameter $\alpha=0.8$, $n=10^3$ and $B=100$ in our numerical studies. 

\noindent{\bf Monotone non-increasing weights:}  
We consider monotone non-increasing {weights}, i.e., $w_1\geq\cdots\geq w_n,$ given $\mu_1\geq\cdots\geq\mu_n.$

\begin{theorem}\label{thm:decreasing-weight-hit-rate}
When $\{w_i, i=1,\cdots, n\}$ are monotone decreasing, (i) for $\beta<1,$ $\exists \tilde{j}\in(1,n)$ s.t. $\lambda^r_i>\lambda^p_i,$ $\forall i<\tilde{j}$; and (ii) for $\beta>1,$ $\exists \tilde{l}\in(1, n)$ s.t. $\lambda^r_i>\lambda^p_i,$ $\forall i>\tilde{l}.$ In particular, if $\tilde{j}, \tilde{l}\in\mathbb{Z}^+,$ then $\lambda^r_{\tilde{j}}=\lambda^p_{\tilde{j}},$ and $\lambda^r_{\tilde{l}}=\lambda^p_{\tilde{l}}.$
\end{theorem}

Theorem~\ref{thm:decreasing-weight-hit-rate} states that compared to HPB-CUM, HRB-CUM favors more popular contents  for $\beta<1,$ and less popular contents  for $\beta>1$.

The following corollary applies to the Zipf popularity distribution.
\begin{corollary}
If the popularity distribution is Zipfian: 
(a) When $\beta<1,$ $\lambda^r_i>\lambda^p_i$ for $ i=1,\cdots, i_0,$ and $\lambda^r_i<\lambda^p_i$ for $i=i_0+1,\cdots, n;$ 
(b) When $\beta>1,$ $\lambda^r_i<\lambda^p_i$ for $i=1,\cdots, i_0,$ and $\lambda^r_i>\lambda^p_i$ for $i=i_0+1,\cdots, n,$ where $i_0=\Bigg\lfloor\left(\frac{\sum_{j}w_j^{\frac{1}{\beta}}j^{\alpha(1-\frac{1}{\beta})}}{\sum_{j} w_j^{\frac{1}{\beta}}}\right)^{\frac{1}{\alpha(1-\frac{1}{\beta})}}\Bigg\rfloor.$
\end{corollary}

\begin{figure}%[htbp]
\centering
%\hspace{-0.5cm}
\begin{minipage}{0.45\textwidth}
\includegraphics[width=1\textwidth]{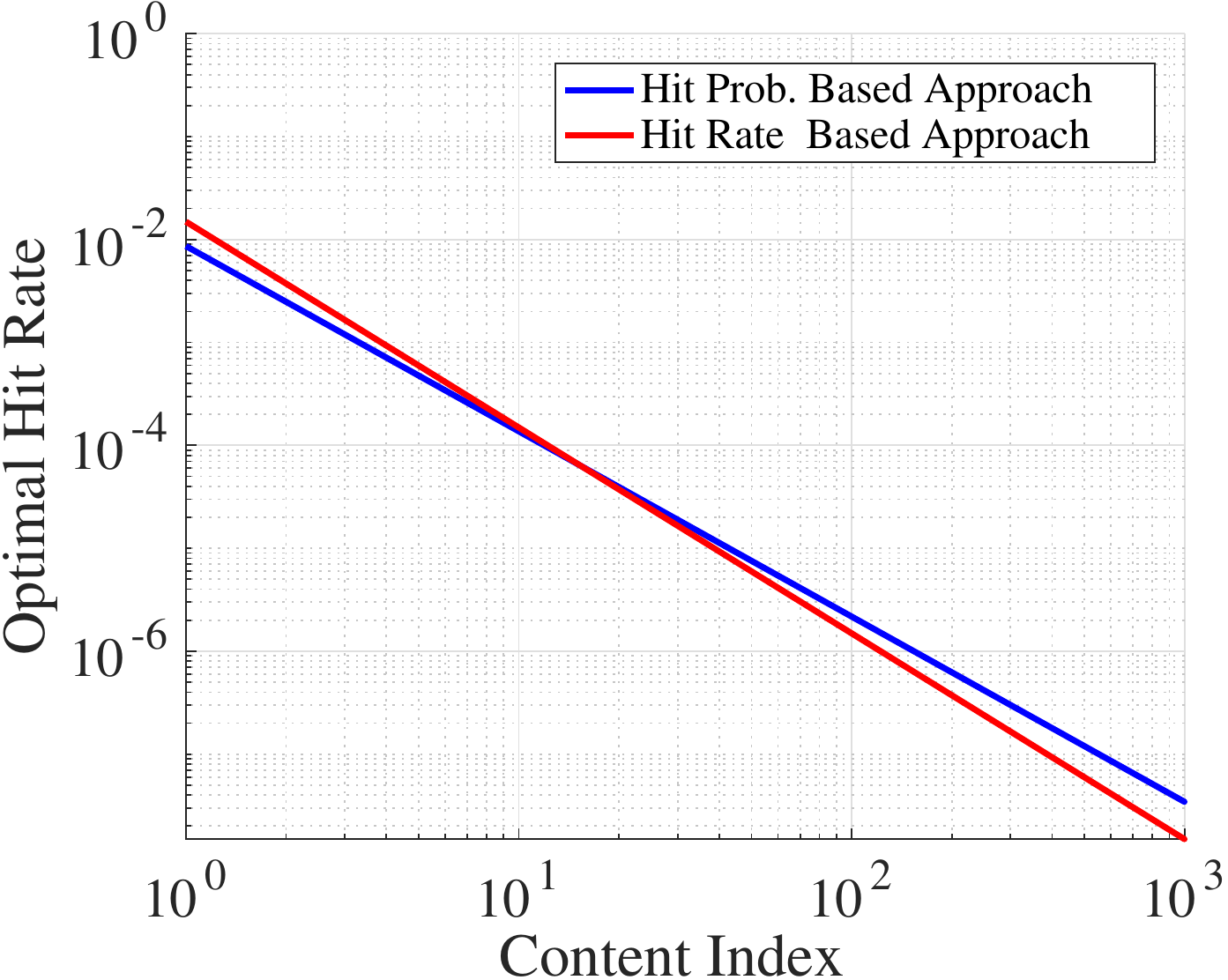}
\subcaption{}
\label{expoa}
\end{minipage}%\hfill
\begin{minipage}{0.47\textwidth}
\includegraphics[width=1\textwidth]{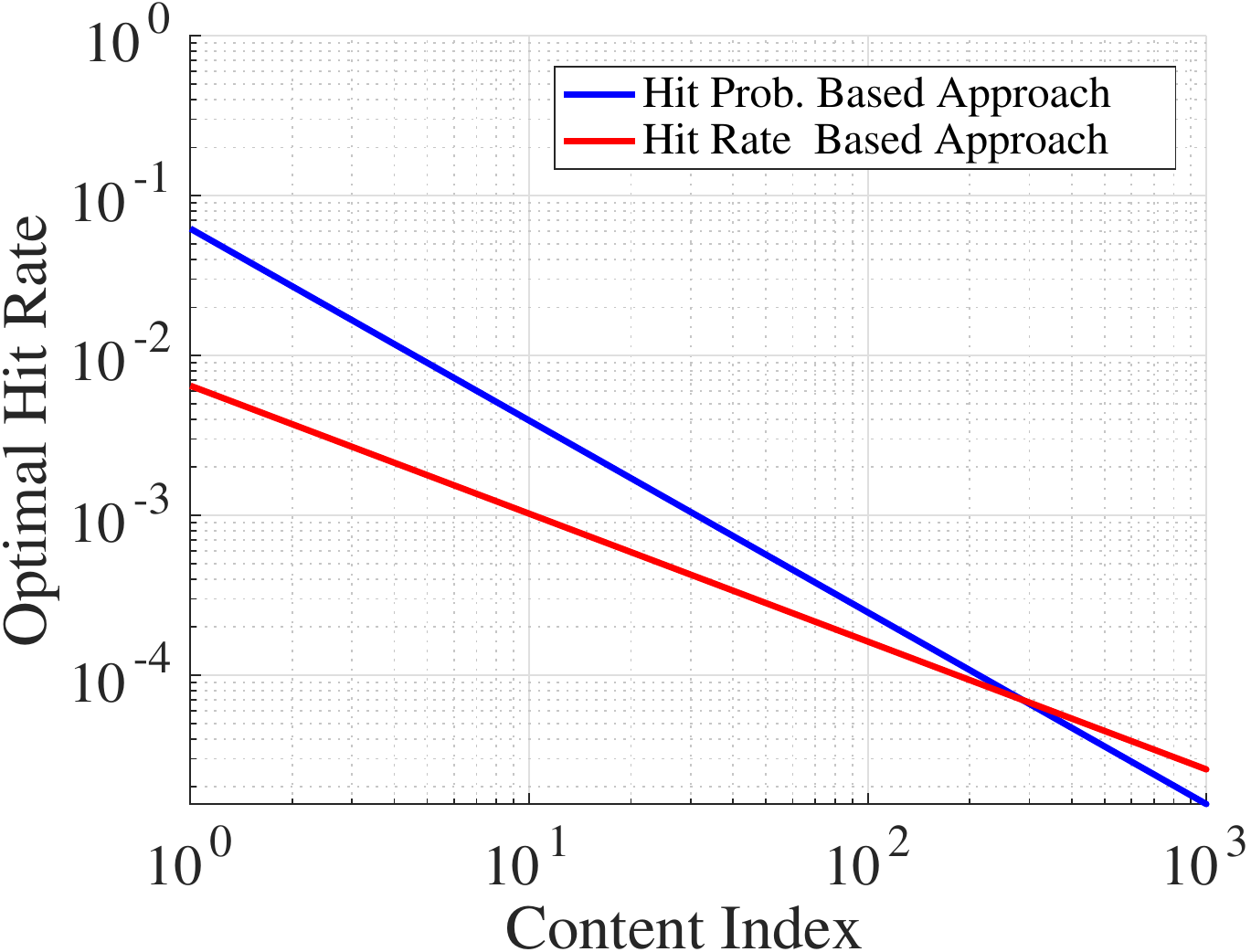}
\subcaption{}
\label{expob}
\end{minipage}
%\begin{minipage}{0.25\textwidth}
%\includegraphics[width=1\textwidth]{figures/pareto/central/betaltone_central_rate.eps}
%\subcaption{$\beta = 0.8$}
%%\label{paretoa}
%\end{minipage}
%\begin{minipage}{0.25\textwidth}
%\includegraphics[width=1\textwidth]{figures/pareto/central/mpd_central_rate.eps}
%\subcaption{$\beta = 2$}
%%\label{paretob}
%\end{minipage}
%\vspace{-0.1in}
\caption{HRB-CUM vs. HPB-CUM under exponential distribution: (a) $\beta=0.8$ and (b) $\beta=2.$}%and generalized Pareto distribution (c), (d).}
\label{mpd_rate_beta_decreasing_wt}
\vspace{-0.25in}
\end{figure}

Figures~\ref{mpd_rate_beta_decreasing_wt} (a) and (b) illustrate the case that $w_i=\mu_i$, $\beta=0.8$ and $\beta=2$, respectively. We consider the Zipf popularity distribution with parameter $\alpha=0.8$, $n=10^3$ and $B=100$ in our numerical studies.

We make a similar comparison of the hit probabilities under HRB-CUM and HPB-CUM.  
\begin{theorem}\label{thm:decreasing-weight-hit-prob}
When $w_1\geq\cdots\geq w_n,$ (i) for $\beta<1,$ $\exists j\in (1, n)$ s.t. $h^r_i>h^p_i,$ $\forall i<j$, and $h^r_i<h^p_i,$ $\forall i>j;$ and (ii) for $\beta>1,$ $\exists l\in (1, n)$ s.t. $h^r_i<h^p_i,$ $\forall i<l$, and $h^r_i>h^p_i,$ $\forall i>l.$ In particular, if $j, l\in\mathbb{Z}^+,$ then $h^r_j=h^p_j,$ and $h^r_l=h^p_l.$
\end{theorem}

The following corollary applies to the Zipf popularity distribution.
\begin{corollary}
If the popularity distribution is Zipfian: 
(a)$h^r_i>h^p_i,$ for $i=1,\cdots, i_0,$ and $h^r_i<h^p_i,$ for $i=i_0+1, \cdots, n$ when $\beta<1;$ 
(b)$h^r_i<h^p_i,$ for $i=1, \cdots, i_0,$ and $h^r_i>h^p_i,$ for $i=i_0+1, \cdots, n,$ where $i_0=\Bigg\lfloor\left(\frac{\sum_{j}w_j^{\frac{1}{\beta}}j^{\alpha(1-\frac{1}{\beta})}}{\sum_{j} w_j^{\frac{1}{\beta}}}\right)^{\frac{1}{\alpha(1-\frac{1}{\beta})}}\Bigg\rfloor,$ when $\beta>1.$ 
\end{corollary}
We numerically verify our results, and observe that they exhibit similar trends as in Figures~\ref{mpd_rate_beta_decreasing_wt} (a) and (b), hence we omit them here due to space constraints.

We are unable to achieve explicit expressions for $h^r_i$, $h^p_i$, $\lambda^r_i$ and $\lambda_i^p$ for HRB-CUM and HPB-CUM when inter-request times are characterized by other distributions.  However, from Section~\ref{sec:distr}, we know that HRB-CUM and HPB-CUM are convex optimization problems when the distribution is DHR. We numerically compare the performance of HRB-CUM and HPB-CUM under a Zipf-like distribution with parameter $\alpha=0.8$.
Similar results as that of the Poisson request process hold for these distributions, and we omit the results due to space limitation.

%\subsubsection{Generalized Pareto Distribution}
%We are unable to achieve closed form expressions for $h^p_i$, $h^r_i$, $\lambda^p_i$ and $\lambda_i^r$ for HPB-CUM and HRB-CUM when inter-request times follow a generalized Pareto distribution.  However, from Section~\ref{sec:distr}, we know~(\ref{eq:hrbcum2}) and~(\ref{eq:hpbcum2}) are convex optimization problems.  We numerically compare the performance of HRB-CUM~(\ref{eq:hrbcum}) and HPB-CUM~(\ref{eq:hpbcum2}) under a Zipf-like distribution with parameter $\alpha=0.8$, shown in Figures~\ref{mpd_rate_beta_decreasing_wt} (c) and (d).  Again, we observe that HRB-CUM favors more popular contents compared to HPB-CUM when $\beta<1;$ and less popular contents when $\beta>1,$ similar to the results under exponential inter arrival request process.

%Similar arguments and numerical results hold for Weibull and Hyperexponential distribution, and we omit the results due to space limitation.

\section{Decentralized Algorithms}\label{sec:online}
In Section~\ref{sec:utility}, we formulated an optimization problem with a fixed cache size under the assumption of a static known workload. However, system parameters (e.g. request processes) can change over time, {and as discussed in Section~\ref{sec:distr}, the optimization problem under some inter-request distributions cannot easily be solved}. Moreover, it is infeasible to solve the optimization problem offline and then implement the optimal strategy.  Hence decentralized algorithms are needed to implement the optimal strategy to adapt to these changes in the presence of limited information.

\begin{remark}
Note that, the proposed decentralized algorithms only require local information (such as irt distribution parameters for that particular content) to achieve global optimality whereas the centralized algorithm requires information about all content request processes.
\end{remark}

In the following, we develop decentralized algorithms for HRB-CUM and compare their performance to those for HPB-CUM under stationary request processes discussed in Section~\ref{sec:distr}.   We only present explicit algorithms for HRB-CUM, similar algorithms for HPB-CUM are available in Appendix~\ref{appd}. %Appendix~\ref{appd}.  % \cite{nitishjian-mobihoc18-tech}. 
  We drop the superscript $r$ in this section for brevity.

\subsection{Dual Algorithm}\label{subsec:dual}
For a request arrival process with a DHR inter-request time distribution,~(\ref{eq:hrbcum2}) becomes a convex optimization problem as discussed in Section~\ref{sec:distr}, and hence solving the dual problem produces the optimal solution.  Since $0<t_i<\infty,$ then $0<\lambda_i/\mu_i<1$ and $0<g_i(\lambda_i/\mu_i)<1.$ Therefore, the Lagrange dual function is 
\begin{align}
D(\eta)=\max_{\lambda_i}\left\{\sum_{i=1}^n U_i(\lambda_i)-\eta\left[\sum_{i=1}^n g_i\left(\lambda_i/\mu_i\right)-B\right]\right\},\label{eq:dualfunct}
\end{align}
and the dual problem is 
\begin{align}
\min_{\eta\geq 0}\quad D(\eta).
\end{align}

Following the standard \emph{gradient descent algorithm} by taking the derivative of $D(\eta)$ w.r.t. $\eta,$  the dual variable $\eta$ should be updated as 
\begin{align}\label{eq:dual-intermidiate}
\eta^{{(k+1)}}\leftarrow \max\left\{0, \eta^{{(k)}}+\gamma\left[\sum_{i=1}^n g_i(\lambda_i^{(k)}/\mu_i)-B\right]\right\}, 
\end{align}
where $k$ is the iteration number, $\gamma>0$ is the step size at each iteration and $\eta\geq0$ due to KKT conditions.

Based on the results in Section~\ref{sec:utility}, in order to achieve optimality, we must have
\begin{small}
\begin{align}\label{eq:eta-dual}
\eta^{{(k)}} =\frac{\mu_i U_i^\prime(\lambda_i^{(k)})}{g_i^\prime(\lambda_i^{(k)}/\mu_i)}\triangleq y_i(\lambda_i^{(k)}/\mu_i),\;\text{i.e.},\; \lambda_i^{{(k)}}=\mu_iy_i^{-1}(\eta^{{(k)}}).
\end{align}
\end{small}

Since $g_i(\lambda_i^{(k)}/\mu_i)$ indicates the probability that content $i$ is in the cache, $\sum_{i=1}^n g_i(\lambda_i^{(k)}/\mu_i)$ represents the number of contents currently in the cache, denoted as $B_{\text{curr}}$.  Therefore, the dual algorithm for a reset TTL cache is
\begin{subequations}\label{eq:dualfinal}
\begin{align}
&t_i^{(k)}=F_i^{-1}( y_i^{-1}(\eta^{{(k)}})),\label{eq:dualfinal-timer}\\
&\eta^{{(k+1)}}\leftarrow \max\left\{0, \eta^{{(k)}}+\gamma(B_{\text{curr}}-B)\right\},\label{eq:dualfinal-eta}
\end{align}
\end{subequations}
which is executed every time a request is made.

\begin{remark}
{From~(\ref{eq:eta-dual}) and~(\ref{eq:dualfinal}), it is clear that if the explicit form of $g_i(\cdot)$ or $g_i^\prime(\cdot)$ is available, then the dual algorithm can be directly implemented.  This is the case for Poisson and generalized Pareto inter-request distributions, see Section~\ref{sec:distr} and the following for details. However, neither is available for the hyberexponential distribution and the $2$-MMPP.  In the following, we will show that the dual algorithm can still be implemented without this information.}
\end{remark}

\noindent{\bf\textit{Poisson Process:}} %Under a Poisson request process, 
We have $g^\prime_i(\lambda_i^{(k)}/\mu_i) = 1,$ and $\lambda_i^{(k)} = U_i^{\prime-1}(\eta^{{(k)}}/\mu_i).$

\noindent{\bf\textit{Generalized Pareto Distribution:}} {When inter-request times are described by a generalized Pareto distribution and utilities are $\beta$-fair, $\lambda_i^{{(k)}}$ is the solution of
\begin{align}\label{eq:paretofpt1}
\mu_i^{1-\beta}w_i(1-(\lambda_i^{(k)}/\mu_i))^{k_i}/[\eta^{(k)}(1-k_i)] - (\lambda_i^{(k)}/\mu_i)^{\beta} = 0.
\end{align}
We can show that there exists a solution in $[0,\mu_i]$ for any $\eta^{{(k)}} > 0$; details are given in Appendix~\ref{appd1}. %{\bf Where is this in Appendix??}}

\noindent{\bf\textit{Hyperexponential Distribution:}} {Under a hyperexponential distribution, we have {$g^\prime_i(x)= \mu_i(1-x)/f_i(F_i^{-1}(x)).$ Since we do not have a closed form expression for $F_i^{-1}(x)$}, no explicit form exists for $g^\prime_i(x)$.  Given~(\ref{eq:eta-dual}) and a $\beta$-fair utility, timer $t_i^{{(k)}}$ is obtained as a solution of the following fixed point equation}
\begin{align}\label{eq:hypexpfpt}
(F_i(t_i^{(k)}))^{\beta+1} - (F_i(t_i^{(k)}))^{\beta} + f_i(t_i^{(k)})/ [\eta^{{(k)}}\mu_i^{\beta-1} ]= 0,
\end{align}
where $F_i(t) = 1-\sum\limits_{j=1}^{l}p_{ji}e^{-\theta_{ji}t}$ and $f_i(t) = \sum\limits_{j=1}^{l}p_{ji}\theta_{ji}e^{-\theta_{ji}t}.$

\noindent{\bf\textit{$2$-MMPP:}} From Section~\ref{sec:distr}, the inter-request times of a $2$-MMPP are described by a second order hyperexponential distribution. Hence timer $t_i^{{(k)}}$ can be updated  from~(\ref{eq:hypexpfpt}) with $F_i(t) = 1-\sum_{j=1}^{2}q_{ji}e^{-u_{ji}t}$ and $f_i(t) = \sum_{j=1}^{2}q_{ji}u_{ji}e^{-u_{ji}t}.$

\begin{remark}
We can similarly design primal and primal-dual algorithms by adding a convex and non-decreasing cost function $C(\cdot)$ to the sum of utilities, denoting the cost for extra cache storage.   For ease of exposition, we relegate their description to Appendix~\ref{appd2}.  In the remainder of the paper, we refer to these distributed algorithms as Dual, Primal and Primal-Dual, respectively.  %Furthermore, it can be easily shown that all the above distributed algorithms converge to the optimal solutions, respectively, using Lyapunov techniques.  Again, the corresponding algorithms for Weibull distribution are available in Appendix~\ref{appd3}.%\cite{nitishjian-mobihoc18-tech}.}
\end{remark}

\subsection{Performance Evaluation}\label{sec:rd}
\begin{figure*}[htbp]
\centering
\hspace{-0.3cm}
\begin{minipage}{0.32\textwidth}
\includegraphics[width=1\textwidth]{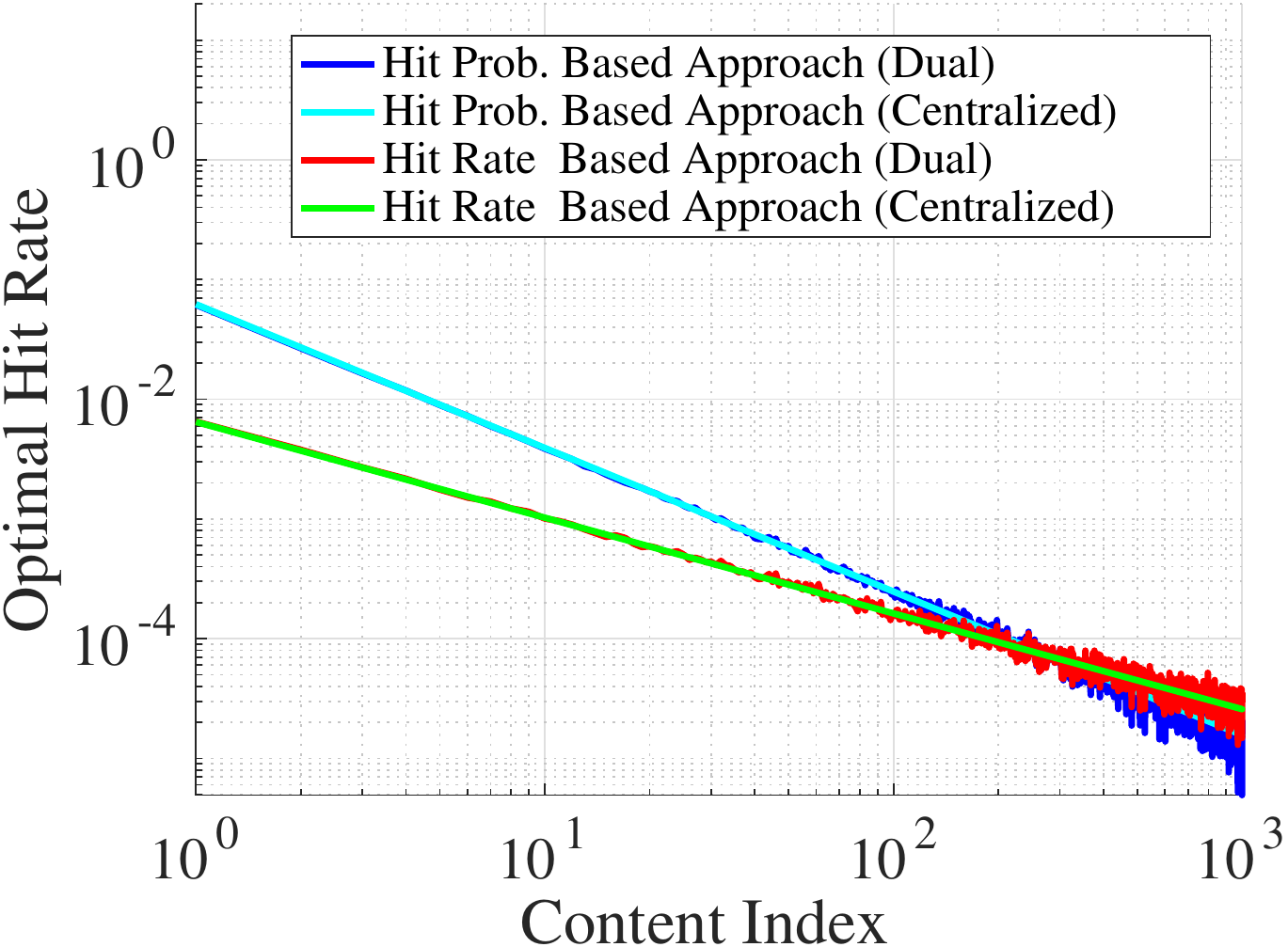}
\subcaption{}
\end{minipage}%\hfill
\begin{minipage}{0.32\textwidth}
\includegraphics[width=1\textwidth]{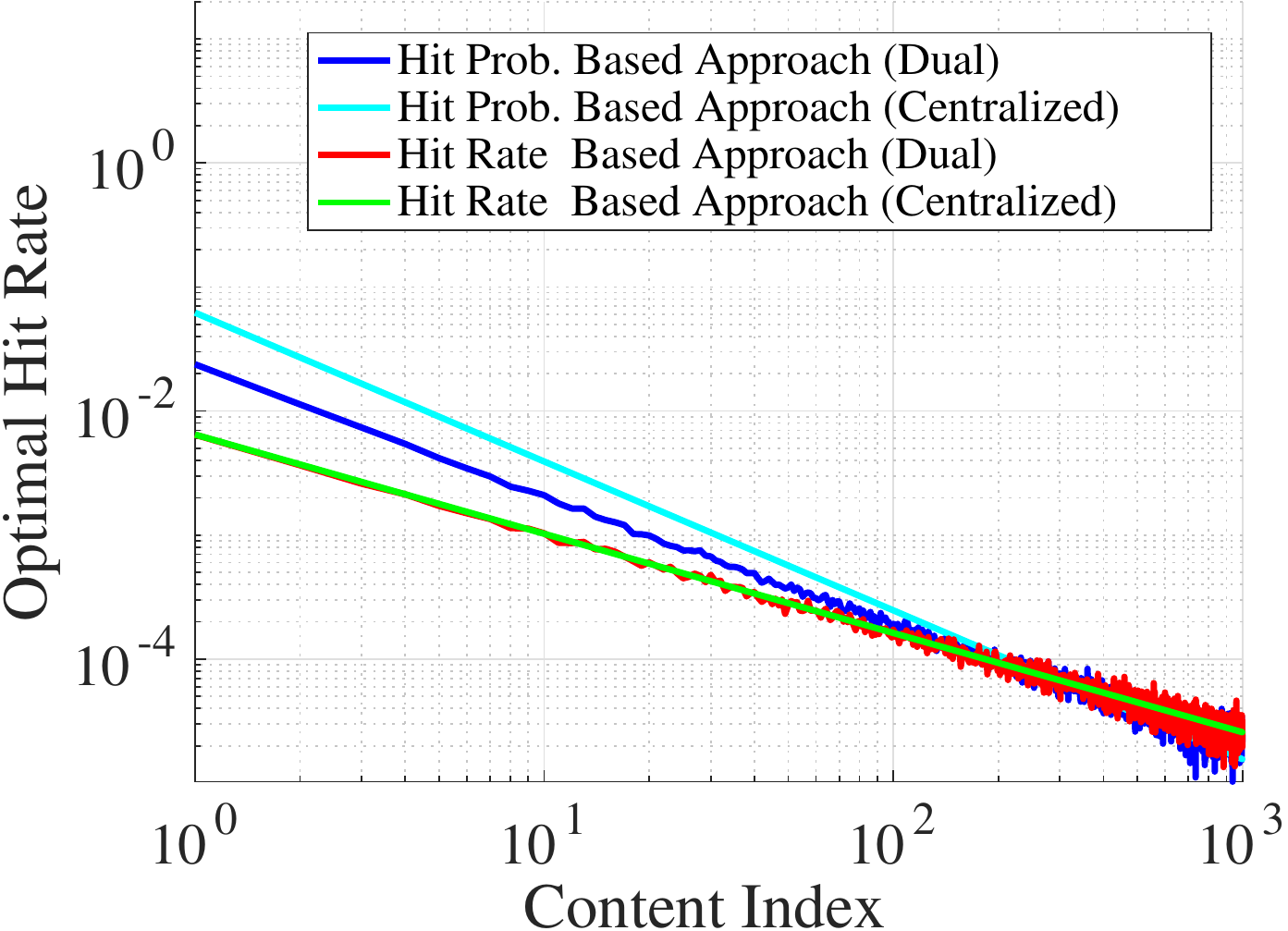}
\subcaption{}
\end{minipage}%\hfill
\begin{minipage}{0.32\textwidth}
\includegraphics[width=1\textwidth]{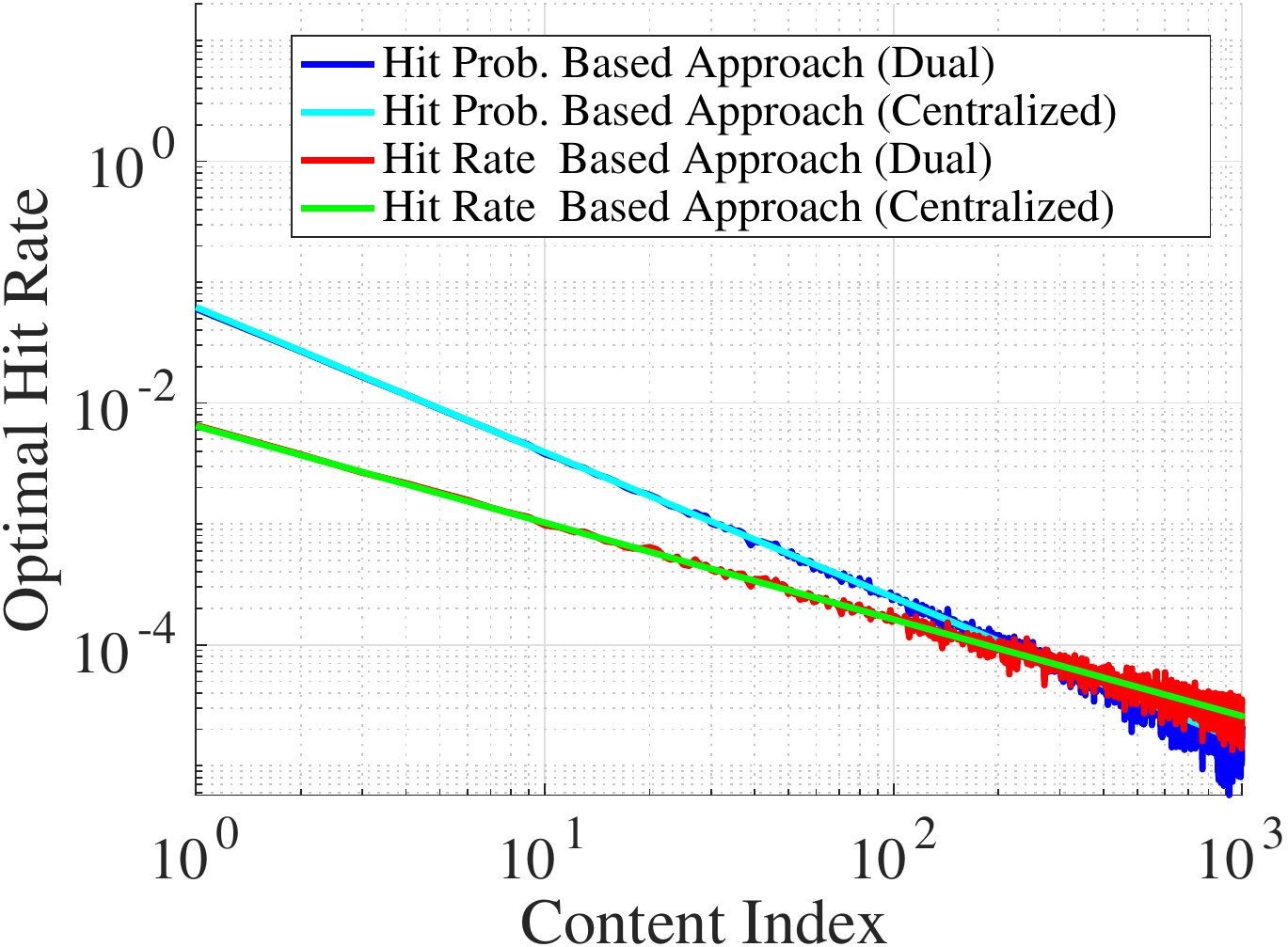}
\subcaption{}
\end{minipage}
\begin{minipage}{0.32\textwidth}
\includegraphics[width=1\textwidth]{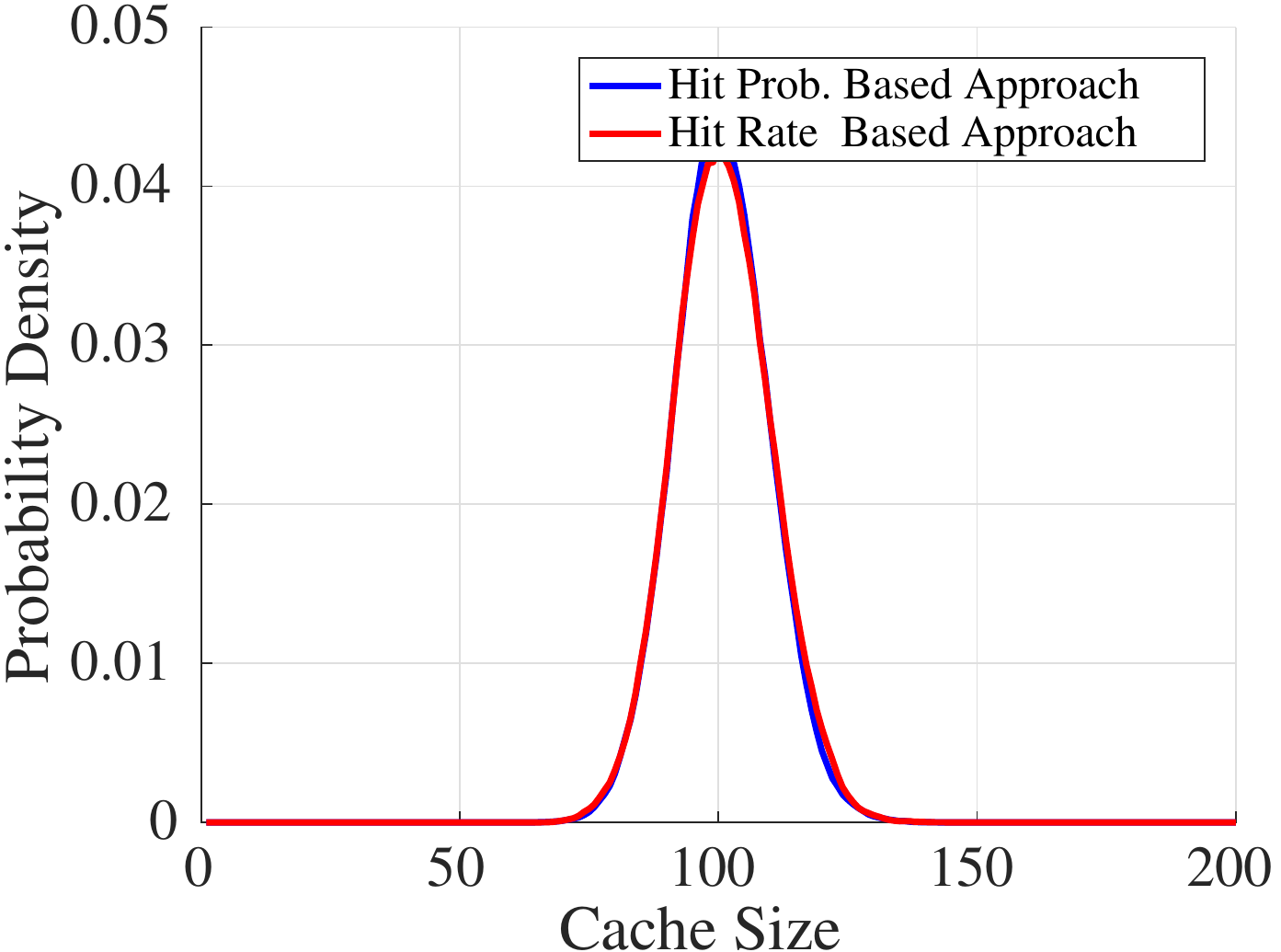}
\subcaption{}
\end{minipage}%\hfill
\begin{minipage}{0.32\textwidth}
\includegraphics[width=1\textwidth]{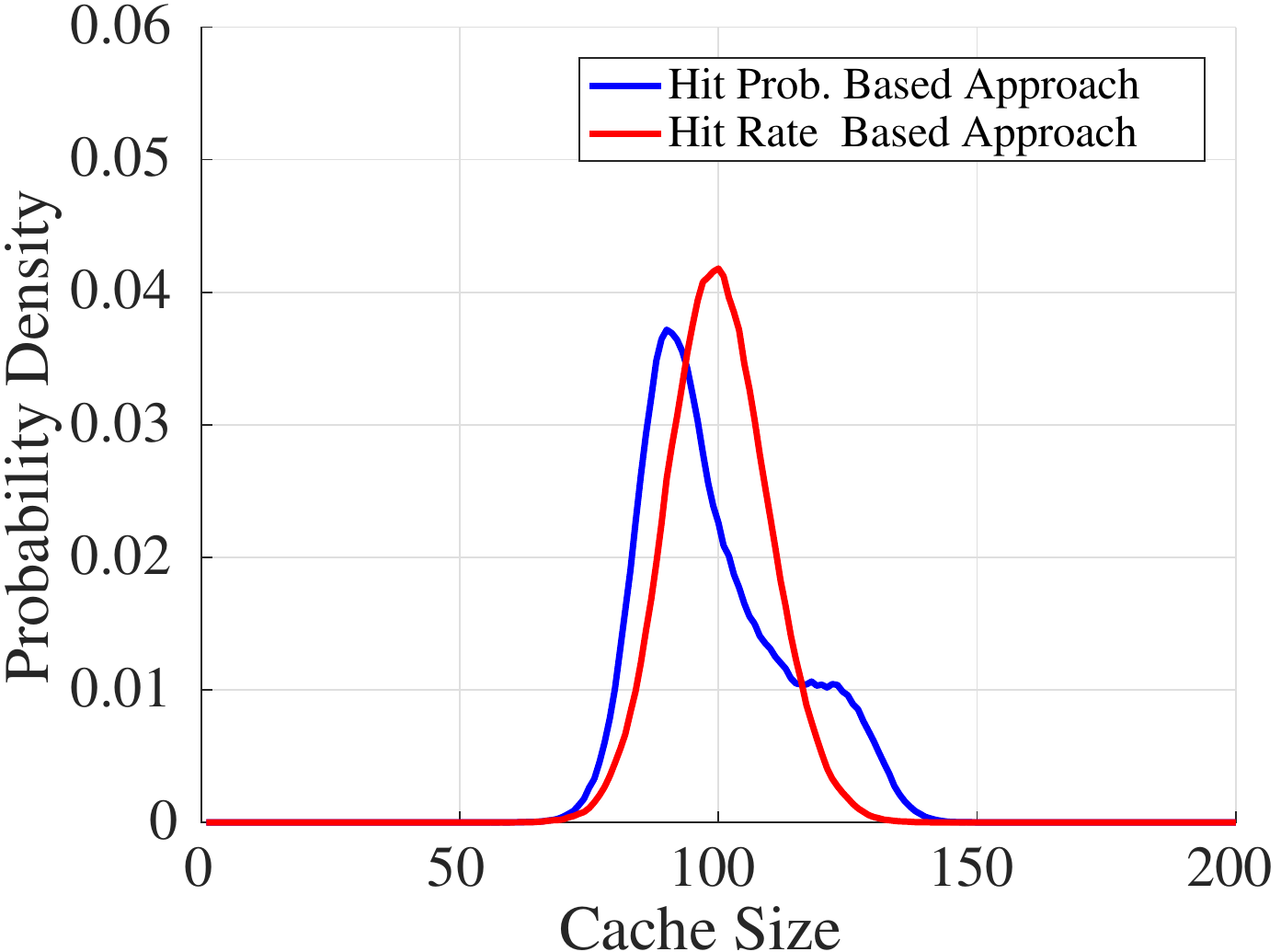}
\subcaption{}
\end{minipage}%\hfill
\begin{minipage}{0.32\textwidth}
\includegraphics[width=1\textwidth]{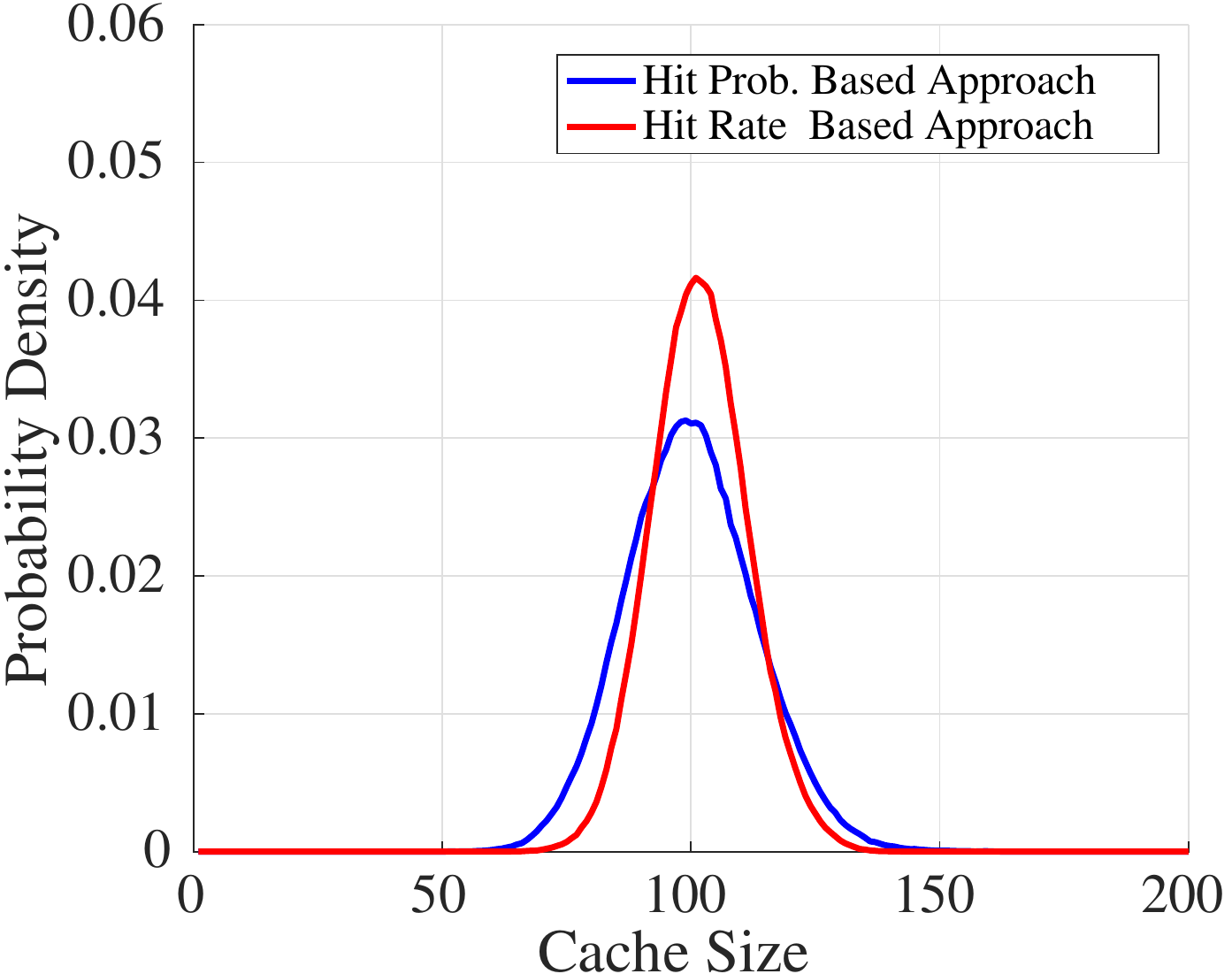}
\subcaption{}
\end{minipage}
\vspace{-0.1in}
\caption{Dual for HRB-CUM and HPB-CUM with under exponential distribution with minimum potential delay fairness; Hit rate (Fig. $a$-$c$) and cache size distribution (Fig. $d$-$f$) comparisons for Dual with $\gamma = 10^{-7}$ ($a, d$), $\gamma = 10^{-3}$ ($b, e$) and $\gamma = 10^{-5}$ ($c, f$) under exponential inter-request process.}
\label{dualsoln}
\vspace{-0.2in}
\end{figure*}

In this Section,  we evaluate the performance of the decentralized algorithms for both HRB-CUM and HPB-CUM when inter-request times are described by stationary request processes with an exponential irt distribution when utility functions are $\beta$-fair.  Due to space restrictions, we limit our study to minimum potential delay fairness, i.e., $\beta=2$.

\subsubsection{Experiment Setup}
In our studies, we consider a Zipf popularity distribution with $\alpha=0.8,$ $n=1000$ and $B=100.$ We consider the inter-request time distributions described in Section~\ref{sec:distr} with an aggregate request rate $\mu=1$ such that $\mu_i=p_i$ from~(\ref{eq:popularity}).  In particular, for exponential distribution, the rate parameter is set to $\mu_i=p_i$.
%(ii) for generalized Pareto distribution, let $k_i=0.48,$ $\sigma_i=(1-k_i)/p_i$ and $\theta_i=0$.  % (iii) for Weibull distribution, consider $k_i = 0.5$, $\theta_i = 1/(p_i\Gamma(1+1/k_i))$, where $\Gamma(\cdot)$ is the gamma function. 
We relegate discussions of generalized Pareto, hyperexponential and $2$-MMPP to Section~\ref{sec:expoapprox} %since no centralized solution is available for them.

\subsubsection{Exactness}\label{sec:numerical-convergence}
We first consider the dual algorithm described in Section~\ref{subsec:dual}. Note that the dual algorithm for generalized Pareto involves solving nonlinear equation~(\ref{eq:paretofpt1}).  % Note that the dual algorithm for generalized Pareto and Weibull distributions involve solving nonlinear equations~(\ref{eq:paretofpt}) and~(\ref{eq:wblfpt}).
 We solve it efficiently with Matlab routine {\it fsolve} using a step size\footnote{Note that the step size has an impact on the convergence and its rate, more details are discussed in Section~\ref{sec:numerical-con-rate}.} $\gamma= 10^{-7}.$  The performance of dual under exponential is shown in Figure~\ref{dualsoln} ($a$), where  ``Centralized" means solutions from solving~(\ref{eq:hrbcum2}). %generalized Pareto and Weibull distributions are shown in Figure~\ref{dualsoln}, {where  ``Centralized" means solutions from solving~(\ref{eq:hrbcum2}).}

From Figure~\ref{dualsoln} ($a$), we observe that the decentralized algorithms yield the exact hit rates under both HRB-CUM and HPB-CUM.  Similarly results hold for hit probabilities, omitted here due to space limits.  Figure~\ref{dualsoln} ($d$) shows the probability density for the number of contents in the cache across these distributions. As expected the density is highly concentrated around the cache size $B$.  Similar results hold for generalized Pareto distribution and we omit the results due to space limits. 

%\begin{figure}
%\centering
%\begin{minipage}{0.25\textwidth}
%\includegraphics[width=0.9\textwidth]{figures/exponential/dual_online/hit_rate_mpd_expo_dual_10-3.eps}
%\end{minipage}%\hfill
%\begin{minipage}{0.25\textwidth}
%\includegraphics[width=0.9\textwidth]{figures/exponential/dual_online/hit_rate_mpd_expo_dual_10-5.eps}
%\end{minipage}
%\begin{minipage}{0.25\textwidth}
%\includegraphics[width=0.9\textwidth]{figures/exponential/dual_online/cachesize_pdf_mpd_expo_dual_10-3.pdf}
%\end{minipage}%\hfill
%\begin{minipage}{0.25\textwidth}
%\includegraphics[width=0.9\textwidth]{figures/exponential/dual_online/cachesize_pdf_mpd_expo_dual_10-5.pdf}
%\end{minipage}
%\vspace{-0.1in}
%\caption{Hit Rate and cache size distribution comparisons for online dual algorithm with $\gamma = 10^{-3}$ (Left) and $\gamma = 10^{-5}$ (Right) under exponential inter request arrival process.}
%\label{hit_prob_cache_size_dual}
%\vspace{-0.3in}
%\end{figure}

We also use primal and primal-dual distributed algorithms to implement minimum potential delay fairness.  In particular, as discussed in Section \ref{sec:online}, primal is associated with a penalty function $C(\cdot)$.  Choosing an appropriate penalty function plays an important role in the performance of primal, since we need to evaluate the gradient at each iteration through  $C^\prime(\cdot)$. Here, we use $C(x)= \max\{0,x - B\log(B+x)\}$ \cite{srikant13}.  Another reasonable choice can be $C(x) = \max\{0,x^m\},\;m\geq1$. We observe that both primal and primal-dual yield exact hit probabilities and hit rates under HRB-CUM and HPB-CUM for minimum potential delay fairness.   We omit the plots due to space constraints. 

\subsubsection{Convergence Rate and Robustness}\label{sec:numerical-con-rate} 
Although the decentralized algorithms converge to the optimal solution as shown in Section~\ref{sec:numerical-convergence}, the rate of convergence is also important from a service provider's perspective.  {Due to space limits, we only focus on the dual here.}  From \eqref{eq:dualfinal}, it is clear that the step size\footnote{Here we use superscript $p$ and $r$ to distinguish the step size of corresponding dual algorithms under HRB-CUM and HPB-CUM, respectively.} $\gamma^p$ (or $\gamma^r$) plays a significant role in the convergence rate.  We choose different values of $\gamma^p$ and $\gamma^r$ and compare the performance of HRB-CUM and HPB-CUM under Poisson request processes, shown in Figure~\ref{dualsoln} ($b$) and ($c$).  On one hand, we find that  when a larger value of $\gamma^p=\gamma^r=10^{-3}$ is chosen, the dual for HRB-CUM easily converges after a few iterations (more than a few $5\times10^5$ iterations), i.e., the simulated hit rates exactly match numerically computed values,  while those of HPB-CUM do not converge.  On the other hand, when a smaller value $\gamma^p=\gamma^r=10^{-5}$ is chosen, both converge in the same number of iterations.  We also used $\gamma^p=\gamma^r=10^{-1}, 10^{-7}$, which exhibit similar behaviors to $10^{-3}$ and $10^{-5},$ respectively, and are omitted due to space constraints.

We also explored the expected number of contents in the cache, shown in Figure~\ref{dualsoln} ($e$) and ($f$).  It is obvious that under HRB-CUM, the probability of violating the target cache size $B$ is quite small, while it is larger for HPB-CUM especially for $\gamma^p=\gamma^r=10^{-3},$ and even for $\gamma^p=\gamma^r=10^{-5},$ HRB-CUM is more concentrated on the target size $B.$  These results indicate that the dual algorithm associated with HRB-CUM is more robust to changes in the step size and converges much faster under exponential inter-requests.

%\vspace{-0.1in}
\subsubsection{Comparison of Decentralized Algorithms}
From the above analysis, we know that at each iteration, the dual algorithm needs to solve a non-linear equation to obtain a timer value, which might be computationally intensive compared to primal and primal-dual. However, for primal, some choices of penalty function $C(\cdot)$ and arrival process $g_i(\cdot)$ may result in large gradients and abrupt function change \cite{smith96}.  Similarly for primal-dual, two scaling parameters $\delta_i$ and $\gamma$ need to be carefully chosen, otherwise the algorithm might diverge.  These demonstrate the pro-and-cons of these distributed algorithms, and one algorithm may be favorable than others in specific situations. 

\section{Stability Analysis of Decentralized Algorithms}\label{sec:stabdec}
In this section, we derive stability results for the decentralized algorithms proposed in Section \ref{subsec:dual}. In particular, we establish stability of the update rule \eqref{eq:dual-intermidiate} around its equilibrium $\eta^*.$ A continuous time approximation to \eqref{eq:dual-intermidiate} was studied in \cite{dehghan16} and for it, stability results were established using Lypaunov theory. Motivated by the fact that this approximation is neither necessary nor sufficient for the stability of the actual discrete-time update rule \eqref{eq:dual-intermidiate}, we propose to analyze its stability directly in the discrete time domain. W.l.o.g., we consider the log utility function. We assume requests for each content arrive according to a Poisson process and give conditions on $\gamma$ guaranteeing stability of the update rule \eqref{eq:dual-intermidiate}. We also perform stability analysis with other irt distributions, such as Pareto distribution, as discussed in Appendix \ref{stab:pareto23}.

%\subsection{Poisson Arrivals}
\subsubsection{Local stability analysis} \label{sec:ltsb}
When requests arrive according to a Poisson process we have
\begin{align}
D(\eta)&= \sum_{i=1}^n w_i\log(\mu_iw_i/\eta)-\eta\left[\sum_{i=1}^n w_i/\eta-B\right] = \overline{W}-W\log(\eta)-W+\eta B,\label{eq:dualfunct2}
\end{align}
\noindent where $\overline{W} = \sum_{i=1}^n w_i\log(\mu_iw_i)$ and $W = \sum_{i=1}^n w_i.$ Let $\eta^* = W/B$ be the unique minimizer of the dual function $D(\eta)$ defined in \eqref{eq:dualfunct2}. We have the following dual algorithm.

\begin{align}\label{eq:dual-alg5}
\eta^{(k+1)}\leftarrow\max\left\{0, \eta^{(k)}+\gamma\left(W/\eta^{(k)}-B\right)\right\}.
\end{align}
Any differentiable function $f(\eta)$ can be linearized around a point $\eta^*$ as $L(\eta) = f(\eta^*) + f^{\prime}(\eta^*)(\eta-\eta^*).$ Denote $f(\eta) = \eta+\gamma\left(W/\eta-B\right)$ with $f:\mathbb{R}^+\rightarrow \mathbb{R}.$ We have $f(\eta^*) = \eta^*$. %{\bf[Comments: it goes a bit too fast, I do some algebra to get these results. maybe good to add a little detail].}} 
We also know that $\eta^* = W/B.$ Under linearization, substituting $f^{\prime}(\eta^*) = 1-(\gamma B^2/W)$ in $L(\eta)$ we get %\red{\bf[Comments: substitute to where? incomplete sentence]}  
\begin{align}
\eta^{(k+1)} = \eta^* + \left(1-\frac{\gamma B^2}{W}\right) (\eta^{(k)}-\eta^*). \label{eq:lin12}
\end{align}
\noindent Denote $\eta_{\delta}^{(k)} = \eta^{(k)}-\eta^*$ as deviation from $\eta^*$ at $k^{th}$ iteration. Hence we have%\red{\bf[Comments: first, it is better to not use $\implies$ in the equations. Second, given the left part of equation 33, define the $\eta_{\delta}^{(k)} $ first, then another equation to get the right part of equation 33 would be better logic. ]}

\begin{align}
\eta_{\delta}^{(k+1)} =\left(1-\frac{\gamma B^2}{W}\right) \eta_{\delta}^{(k)}=\left(1-\frac{\gamma B^2}{W}\right)^k \eta_{\delta}^{(0)}.\label{eq:lin2}
\end{align}

and \eqref{eq:lin12} is locally asymptotically stable if 
%Assuming $\eta_{\delta}^{(0)}\approx 0$, $\eta_{\delta}^{(k+1)}\approx \eta_{\delta}^{(k)}$ only when $|1-\frac{\gamma B^2}{W}| < 1$ with $\gamma,B,W >0.$ Thus for local asymptotic stability we have
\begin{align}
\gamma < \frac{2W}{B^2}\label{eq:gamma_lstb}.
\end{align}
Thus, if \eqref{eq:gamma_lstb} holds, then the update rule \eqref{eq:dual-intermidiate} converges to $\eta^*$ as long as $\eta^{(0)}$ is sufficiently close to $\eta^*$.
\subsubsection{Global stability guarantees}
%Consider a candidate Lyapunov function $V(\eta) = D(\eta) - D(\eta^*)$\footnote {It can be shown that $V(\eta)$ is a candidate by showing the convexity property and $V(\eta^*) = 0$.}. A continuous time Lyapunov analysis with $V(\cdot)$ as candidate Lyapunov function yields global stability of  the dual algorithm \eqref{eq:dualfinal-eta} from any starting condition and across different scaling parameters \cite{dehghan16}. 
%However, for higher values of $\gamma$, the dual diverge from any initial value of $\eta.$ Thus we propose to use discrete time Lyapunov techniques for stability analysis of the dual algorithm. 
In \cite{dehghan16}, a Lyapunov function was constructed for a continuous-time approximation to \eqref{eq:dual-intermidiate}. Now, we consider a discrete-time Lyapunov candidate $V(\eta) = D(\eta) - D(\eta^*)$ for \eqref{eq:dual-intermidiate} directly. By discrete time Lyapunov function theory \cite{hahn58}, for global asymptotic stability, we must have $V(f(\eta))-V(\eta) < 0, \forall\;\eta>0, \eta\ne\eta^*$ and for some candidate Lyapunov function $V.$ We evaluate $\Delta V(\eta) = V(f(\eta))-V(\eta)$ with $V(\eta) = D(\eta) - D(\eta^*).$

%\red{\bf[Comments: Suddenly comes a new function $G(\eta)$, it is better to have a connection, like a short sentence]}
\begin{align}
\Delta V(\eta) &= V(f(\eta))-V(\eta)=D(f(\eta)) - D(\eta) =-W\log(f(\eta))+f(\eta)B+W\log\eta-\eta B,\nonumber\\
&=-W\log\left[1+\frac{\gamma}{\eta}\bigg(\frac{W}{\eta}-B\bigg)\right]+B\gamma\left(\frac{W}{\eta}-B\right).\label{eq:dualfunct3}
\end{align}
Here, we are interested in finding a scaling parameter such that $\Delta V(\eta) < 0$ for all $\eta > 0$ thereby proving that the online algorithm \eqref{eq:dual-intermidiate} is stable for any initial starting value $\eta^{(0)}$. We consider two cases: where $\gamma$ is constant and another when $\gamma$ is a function of the dual variable. 
%the candidate Lyapunov function $V$ would provide global stability guarantees for the online algorithm given by \eqref{eq:dual-alg5}.\\ 

\noindent {\bf Scaling Parameter as a function of dual variable:} 
We consider the case $\gamma = \gamma(\eta)$ such that the update rule \eqref{eq:dual-intermidiate} is globally stable around its equilibrium $\eta^*.$ Such a function is constructed in Appendix \ref{stab:poisson234}.

\noindent {\bf Constant Scaling Parameter:}
When $\gamma$ is fixed, independent of $\eta,$ we can show that $V$ is not a Lyapunov function. We consider the following theorem.

%\begin{proof}
%Consider the derivative of the function $G(\eta)$.
%\begin{align}
%G^\prime(\eta) = \frac{-\frac{W\gamma}{\eta^2}\left(\frac{2W}{\eta}-B\right)}{1+\frac{\gamma}{\eta}\left(\frac{W}{\eta}-B\right)} \label{eq:deriv_g}
%\end{align}
%
%\noindent We also have
%
%\begin{align}
%\eta < \eta^* &\implies W/\eta-B > W/\eta^*-B = 0\nonumber\\
%&\implies W/\eta-B > 0 \;\texttt{and}\; 2W/\eta-B > 0 \label{eq:deriv_g2}
%\end{align}
%
%\noindent Combining Equations \eqref{eq:deriv_g} and \eqref{eq:deriv_g2} and considering $W, \eta, \gamma > 0$ we have
%
%\begin{align}
%G^\prime(\eta) < 0 \nonumber
%\end{align}
%
%\noindent Hence $G(\eta)$ is a decreasing function of $\eta$ in the domain $[0,\eta^*].$ Note that here we implicitly assume the scaling parameter $\gamma$ to be independent of $\eta$ while taking derivative of $G(\eta)$ in Equation \eqref{eq:deriv_g}.
%
%\end{proof}

\begin{theorem}\label{thm:const_gamma}
Given $\gamma>0$, $\Delta V(\eta) > 0\; \forall\; \eta < \eta^*$.  %\red{\bf [Comments: there is no ``," in this long sentence]}
\end{theorem}

%\begin{proof}
%Evaluating the function $G(\eta)$ at $\eta = \eta^*$ and putting $W/\eta^* = B$ we get
%
%\begin{align}
%G(\eta^*) = -W\log\left[1+\frac{\gamma}{\eta^*}\bigg(\frac{W}{\eta^*}-B\bigg)\right]\nonumber\\+B\gamma\left(\frac{W}{\eta^*}-B\right) = 0
%\end{align}
%
%\noindent From Lemma \ref{lemma:const_gamma} it is clear that $G(\eta) > G(\eta^*)\;\forall\;\eta<\eta^*.$ Hence $G(\eta) > 0\; \forall\; \eta < \eta^*.$
%
%\end{proof}

\noindent From Theorem \ref{thm:const_gamma}, it is clear that when $\gamma$ is independent of $\eta$, the candidate Lyapunov function: $V(\eta) = D(\eta) - D(\eta^*)$ cannot guarantee global asymptotic stability. See proof in Appendix \ref{app-stab-dec}.

\section{Poisson Online Approximation}\label{sec:expoapprox}
 \begin{figure*}[htbp]
\centering
\begin{minipage}{.32\textwidth}
\centering
\includegraphics[width=1\linewidth]{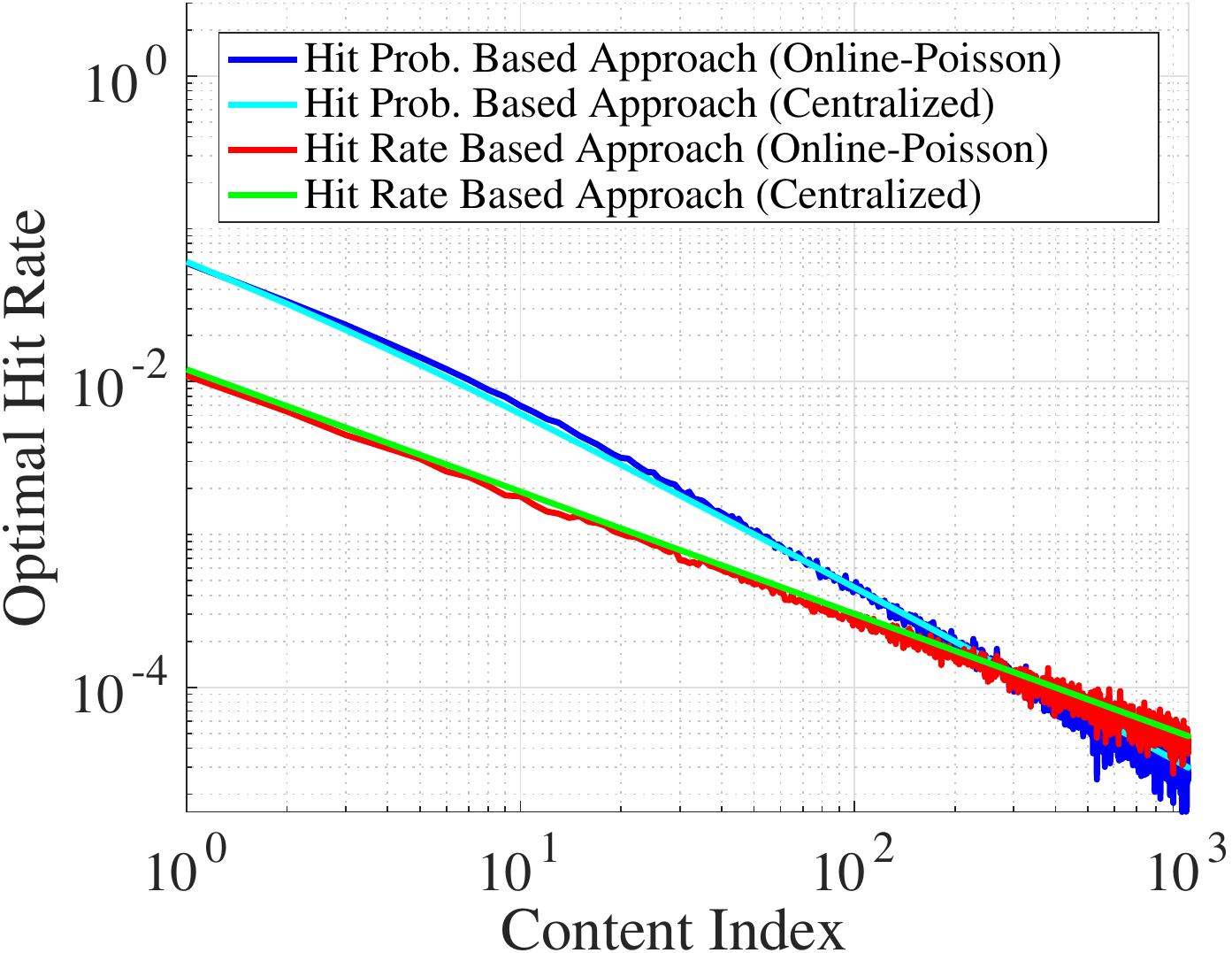}
\caption{Poisson online approximation to Generalized Pareto inter-requests.}
\label{poissonapprox-pareto}
\end{minipage}\hfill
\begin{minipage}{.32\textwidth}
\centering
\includegraphics[width=1\textwidth]{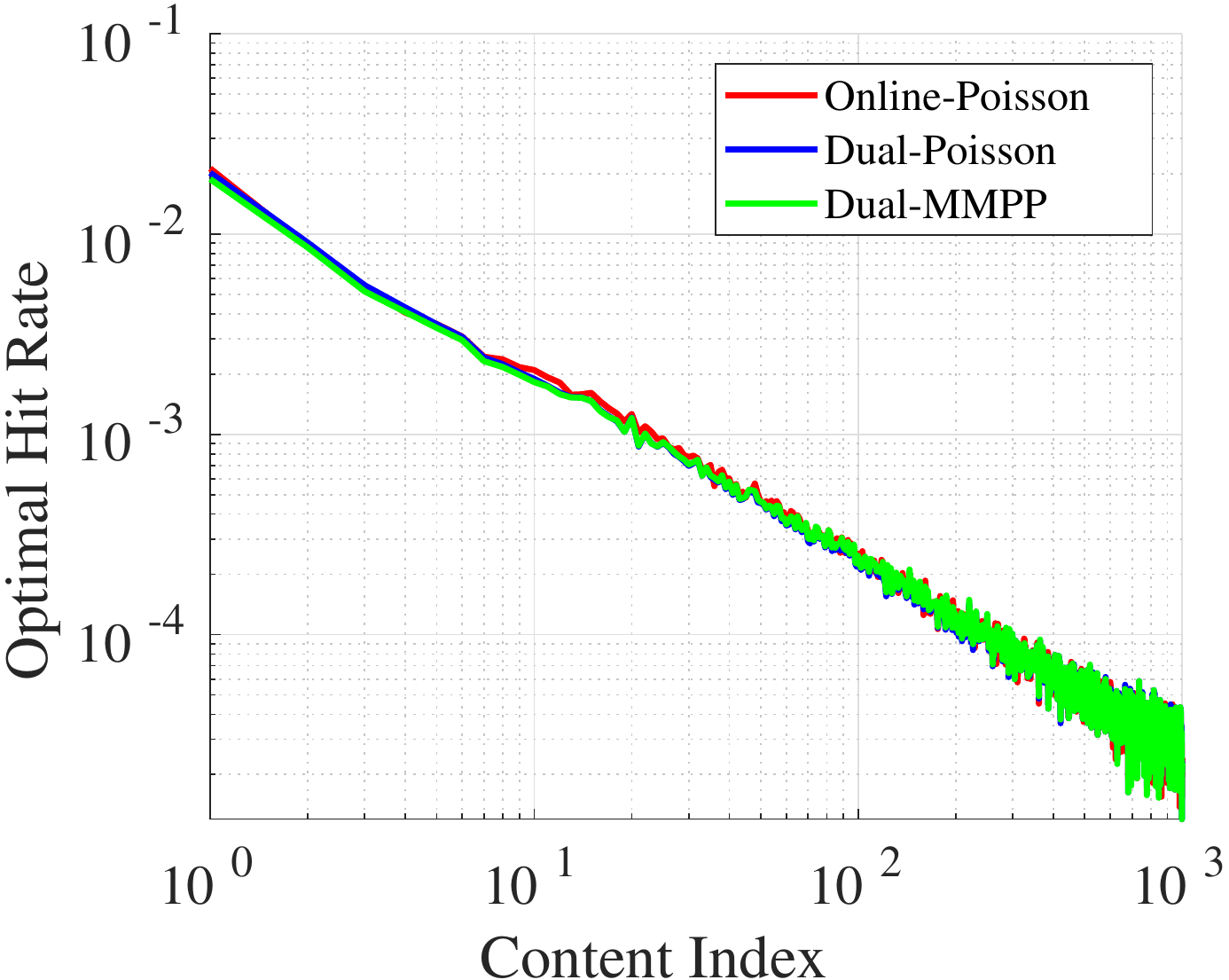}
\caption{Poisson online approximation to 2-MMPP inter-requests: $x=10^{-3}$.}
\label{poissonapprox}
\end{minipage}\hfill
\begin{minipage}{.32\textwidth}
\centering
\includegraphics[width=1\textwidth]{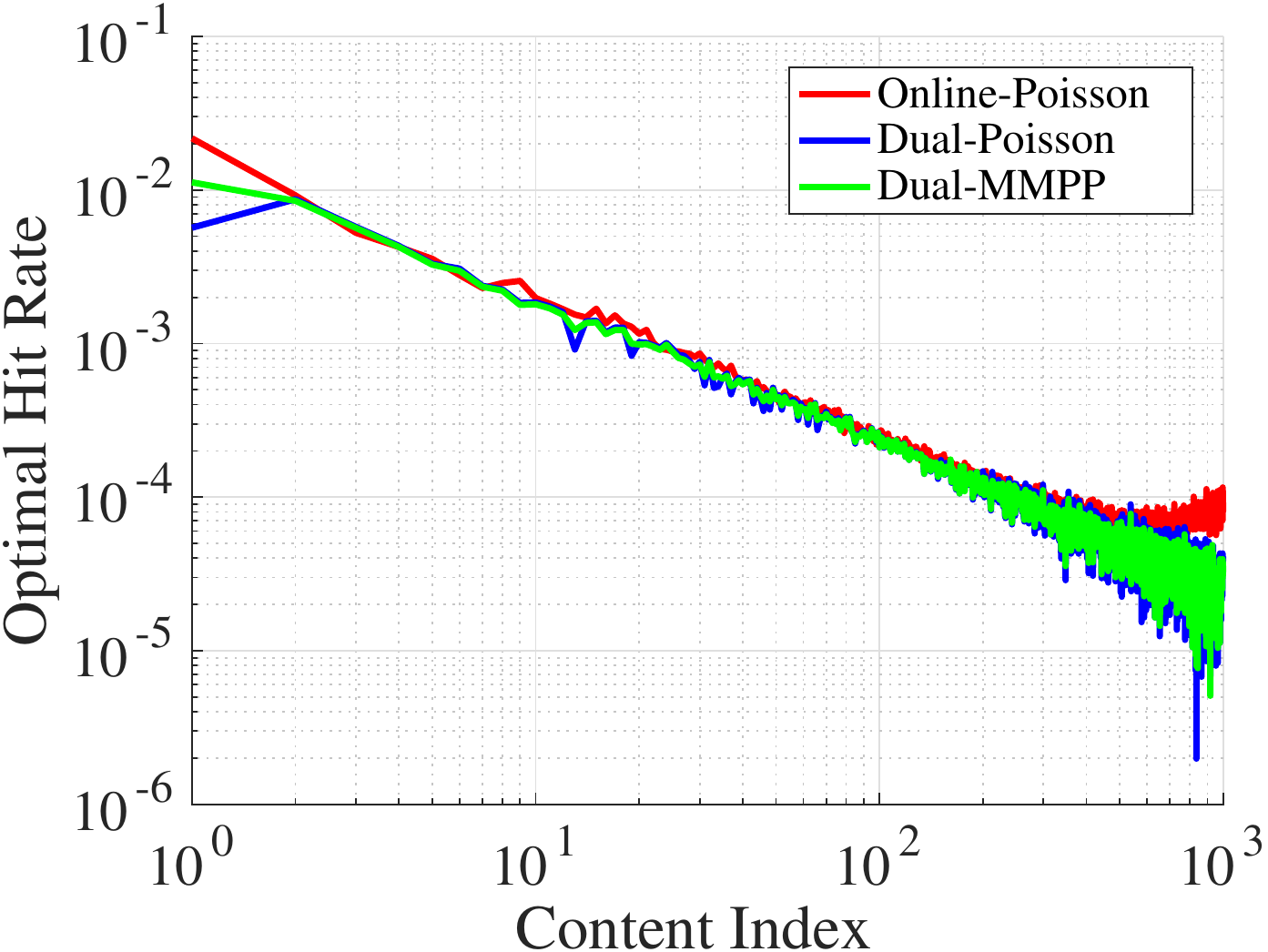}
\caption{Poisson online approximation to 2-MMPP inter-requests: $x=10^{-7}$.}
\label{poissonapprox1}
\end{minipage}
\vspace{-0.15in}
\end{figure*}

From Section~\ref{sec:online}, it is clear that the implementation of Dual under generalized Pareto, hyperexponential distributions and $2$-MMPP involves solving non-linear fixed point equations, which are computationally intensive.  However, the Dual for the case of requests governed by Poisson processes is simple.  Furthermore, knowledge of the inter-request distribution is also required.  However, this is not always available to the service provider.   

In this section,  we apply the Dual designed for Poisson request processes to a workload where requests are described by a {\emph non-Poisson} stationary request processes. Such an algorithm does not require solving non-linear equations and hence is computationally efficient.  Moreover, we also use estimation techniques introduced in \cite{dehghan16} to approximate request rates which makes these distributed algorithms work in an online fashion.  %W.l.o.g., we consider content $i.$

\vspace{-0.1in}
\subsection{Online Algorithm}\label{sec:estimator}
We consider the problem of estimating the arrival rate $\mu_i$ for content $i$ adopting techniques used in \cite{dehghan16} described as follows. Denote the remaining TTL time for content $i$ as $\tau_i$.  This can be computed given $t_i$ and a time-stamp for the last request time for content $i.$  Recall that $X_{ik}$ is a random variable corresponding to the inter-request times for requests for content $i.$  Let $\bar{X}_{ik}$ be the mean.  Then we approximate the mean inter-request time as $\hat{\bar{X}}_{ik}=t_i-\tau_i.$ Clearly $\hat{\bar{X}}_{ik}$ is an unbiased estimator of $\bar{X}_{ik}$, and hence an unbiased estimator of $1/\mu_i.$  In this section, we use this estimator to implement the distributed algoritms, which now becomes an online algorithm.

Given this estimator and Dual~(\ref{eq:dualfinal}), we propose the following \emph{Poisson approximate online algorithm}
\begin{subequations}\label{eq:mod-dual-general}
\begin{align}
&t_i^{(k)}= -\frac{1}{\hat{\mu}_{iP}}\log\Bigg(1-\frac{1}{\hat{\mu}_{iP}}U_i^{\prime-1}\left(\frac{\eta^{(k+1)}}{\hat{\mu}_{iP}}\right)\Bigg), \label{eq:timer-poisson-online}\displaybreak[0]\\
&\eta^{{(k+1)}}\leftarrow \max\{0, \eta^{{(k)}}+\gamma(B_{\text{curr}}-B)\}. \label{eq:eta-poisson-online} 
\end{align}
\end{subequations}
There are two differences between our proposed algorithm~(\ref{eq:mod-dual-general}) and Dual~(\ref{eq:dualfinal}).  First, the explicit form of~(\ref{eq:dualfinal-timer}) is different for different inter-request distributions as discussed in Section~\ref{subsec:dual}, while we always adopt the explicit form of Poisson process in~(\ref{eq:timer-poisson-online}). Second,  $\mu_i$ in~(\ref{eq:eta-dual}) is the exact value of the mean arrival rate of the corresponding inter-request distribution, while we estimate its value as discussed above and denote it as $\hat{\mu}_{iP}.$  However, the value of $B_{\text{curr}}$ denotes the number of contents currently in the cache under the real inter-request distribution under both~(\ref{eq:mod-dual-general}) and~(\ref{eq:dualfinal}).   In the following, we consider the performance of~(\ref{eq:mod-dual-general}) under different inter-request distributions.  

\subsection{Generalized Pareto Distribution} 

In this section, we apply the online algorithm~(\ref{eq:mod-dual-general}) %designed for the case of requests described by a Poisson process
 to a workload where requests are described by stationary request process under generalized Pareto distribution with shape parameter $k_i = 0.48.$
 %We consider the arrivals follow a generalized Pareto distribution with shape parameter $k_i = 0.48.$  
  The performance is shown in Figure~\ref{poissonapprox-pareto} .  It is clear that the approximation is accurate.  Furthermore, it has been theoretically characterized in \cite{weinberg16} that for any given generalized Pareto model with finite variance, the exponential approximation that minimizes the K-L divergence between these two distributions has the same mean as that of the generalized Pareto distribution, i.e. $\mu_i=(1-k_i)/\sigma_i$. The estimator we use  in our online algorithm~(\ref{eq:mod-dual-general}), i.e., $1/\hat{\mu}_{iP}$, is an unbiased estimator of mean inter-request time of the generalized Pareto arrival process, thus explaining the better performance of our Poisson approximation in accordance with the theoretical results provided in \cite{weinberg16}. Moreover, we notice that when $k_i$ becomes smaller, the accuracy has been improved.  However, this approximation has poor performance when $k_i>0.5$ since the generalized Pareto distribution has infinite variance for $k_i>0.5.$
  
%In our algorithm, we do not need this requirement. 

\subsection{$m$-state MMPP}
%We observe similar limiting behavior for an $m$-state MMPP arrival process. Below theorems justify our observation. We relegate the proofs to Appendix~\ref{app-limit-mmmpp}.
Under the general $m$-state MMPP, we can theoretically characterize the limit behaviors of the irts. W.l.o.g. denote the transition rate for content $i$ from state $j$ to $k$ as $r_{jki}$.  Let $Q_i = [r_{jki}, 1\le j \le m, 1\le k \le m],$ be the corresponding generator matrix. Arrivals for content $i$ at state $j$ are described by a Poisson process with rate $\theta_{ji}$.  Then the steady state distribution, $\pi_i = [\pi_{ji}, 1\le j \le m]$, satisfies $\pi_i Q_i= 0$. We represent $r_{jki}=a_{jki}x_i$, where $a_{jki}$ are constants, and $0\leq x_i\leq\infty.$ We summarize the results in the following theorems and relegate the proofs to \ref{app-limit-mmmpp}. 

\begin{theorem}\label{thm:mstate-zero}
When $r_{jki}\rightarrow0$, i.e., $x_i\rightarrow0$, the inter-request times are  described by an $m^{th}$ order hyperexponential distribution.
\end{theorem}
\begin{theorem}\label{thm:mstate-infinite}
When $r_{jki}\rightarrow\infty$, i.e., $x_i\rightarrow\infty$,  the inter-request times are exponentially distributed with mean arrival rate
\begin{align}\label{eq:mmpp-infty}
\bar{\theta}_i=\sum_{j=1}^m \theta_j \pi_{ji}.
\end{align}
i.e., $m$-MMPP is equivalent to a Poisson process with rate $\bar{\theta}_i$, i.e., our approximation is exact.
\end{theorem}
Since there is no explicit form of the inter-request time distribution for a general $m$-state MMPP, we focus on a $2$-MMPP in our numerical studies.
\subsubsection{2-MMPP} \label{sec:appro-2mmpp}
The optimal hit rates under $2$-MMPP can be obtained through solving Dual for a second order hyperexponential distribution with parameters $q, u_1$ and $u_2$ defined in \eqref{eq:mmppparams}. However, from Section~\ref{sec:online}, Dual requires solving a non-linear equation \eqref{eq:hypexpfpt}.   Instead, we consider Poisson approximation~(\ref{eq:mod-dual-general}) under $2$-MMPP.    W.l.o.g., we assume the phase rates $\theta_{1i}$ and $\theta_{2i}$ for $i=1,\cdots, n$ to be Zipf distributed with parameters $0.4$ and $0.8,$ respectively.  

\noindent{\bf Limiting Behavior: }%\label{sec:mmppbehavior}
We first evaluate the performance of Poisson online approximation algorithm~(\ref{eq:mod-dual-general}) for different transition rates $r_{12i}$ and $r_{21i}$.

\begin{theorem}\label{thm:2mmpp}
(1) When $r_{12i}, r_{21i}\rightarrow\infty$, i.e., $x_i\rightarrow\infty$, $2$-MMPP is equivalent to a Poisson process with rate $\frac{\theta_{1i}a_{21i}+\theta_{2i}a_{12i}}{a_{12i}+a_{21i}}$, i.e., our approximation is exact.\\

%\begin{align}\label{eq:2mmpp-infty}
%u_{1i} \rightarrow \frac{\theta_{1i}a_{21i}+\theta_{2i}a_{12i}}{a_{12i}+a_{21i}}, \quad u_{2i} \rightarrow \infty, \quad q_{1i} \rightarrow 1,\quad q_{2i} \rightarrow 0,
%\end{align}
%i.e., $2$-MMPP is equivalent to a Poisson process with rate $u_{1i}$, i.e., our approximation is exact.\\
(2) When $r_{12i}, r_{21i}\rightarrow0$, i.e., $x_i\rightarrow0$,
\begin{align}\label{eq:2mmpp-zero}
u_{1i} \rightarrow \theta_{2i},\quad u_{2i} \rightarrow \theta_{1i}, \quad q_{1i} \rightarrow \frac{\theta_{2i}a_{12i}}{\theta_{1i}a_{21i}+\theta_{2i}a_{12i}},\quad q_{2i} \rightarrow \frac{\theta_{1i}a_{21i}}{\theta_{1i}a_{21i}+\theta_{2i}a_{12i}},
\end{align}
%i.e., $2$-MMPP can be treated as a weighted sum of two Poisson processes with rates $\theta_{1i}$, $\theta_{2i}$ and weights $q_{1i}$, $q_{2i},$ respectively. \\
%(3) When $0<r_{1i}, r_{2i}<\infty$, i.e., $0< x_i<\infty$:  As discussed in Section~\ref{sec:distr}, $2$-MMPP is equivalent to a second order hyperexponential distribution with parameters given in~(\ref{eq:mmppparams}).
\end{theorem}
The proof is relegated to Appendix~\ref{app-limit-2mmpp}.

%
%\noindent{\textit{\textbf{Case $1$: When $r_{1i}\rightarrow\infty$ and $r_{2i}\rightarrow\infty$, i.e., $x_i\rightarrow\infty$:}}}
%By applying L'Hospital's rule to~(\ref{eq:mmppparams}), we have 
%\begin{align}
%u_{1i} &= \frac{\theta_{1i}a_{2i}+\theta_{2i}a_{1i}}{a_{1i}+a_{2i}}, \quad & u_{2i} &=\infty, \nonumber\displaybreak[0]\\
%q_{1i} &=1,\quad &q_{2i} &=0.
%\end{align}
%\noindent{Therefore, $2$-MMPP is equivalent to a Poisson process with rate $u_{1i}$, i.e., our approximation is exact. }
%
%
%\noindent{\textit{\textbf{Case $2$: When $r_{1i}\rightarrow0$ and $r_{2i}\rightarrow0$, i.e., $x_i\rightarrow0$:}}} W.l.o.g., we assume $\theta_{1i}\geq\theta_{2i}.$
%Similarly, by applying L'Hospital's rule to~(\ref{eq:mmppparams}), we have 
%\begin{align}
%u_{1i} &= \theta_{2i},\quad &u_{2i} &=\theta_{1i}, \nonumber\displaybreak[1]\\
%q_{1i} &=\frac{\theta_{2i}a_{1i}}{\theta_{1i}a_{2i}+\theta_{2i}a_{1i}},\quad &q_{2i} &=\frac{\theta_{1i}a_{2i}}{\theta_{1i}a_{2i}+\theta_{2i}a_{1i}}.
%\end{align}
%\noindent{Therefore, $2$-MMPP can be treated as a weighted sum of two Poisson processes with rates $\theta_{1i}$, $\theta_{2i}$ and weights $q_{1i}$, $q_{2i},$ respectively.}  
%
%
%\noindent{\textit{\textbf{Case $3$: When $0<r_{1i}, r_{2i}<\infty$, i.e., $0< x_i<\infty$:}}} As discussed in Section~\ref{sec:distr}, $2$-MMPP is equivalent to a second order hyperexponential distribution with parameters given in~(\ref{eq:mmppparams}).   

\noindent{\bf Numerical Validation: }%\label{sec:mmppnvld}
%\begin{figure}%[htbp]
%\centering
%%\hspace{-0.5cm}
%\begin{minipage}{0.25\textwidth}
%\includegraphics[width=1\textwidth]{figures/expoapprox/figs/online/hr_dual_mpd_mmpp_expo_online_HRB_10-3.eps}
%\end{minipage}%\hfill
%\begin{minipage}{0.25\textwidth}
%\includegraphics[width=1\textwidth]{figures/expoapprox/figs/online/hr_dual_mpd_mmpp_expo_online_HRB.eps}
%\end{minipage}
%\vspace{-0.1in}
%\caption{Poisson online approximation to 2-MMPP inter-arrivals: \textit{(Left)} $x=10^{-3}$ and \textit{(Right)} $x=10^{-7}$}%and generalized Pareto distribution (c), (d).}
%\label{poissonapprox}
%\vspace{-0.1in}
%\end{figure}
We numerically verify the results in Theorem ~\ref{thm:2mmpp} %all three cases 
 by taking different values of transition rates. The performance comparison between two limiting cases are shown in Figures~\ref{poissonapprox} and~\ref{poissonapprox1}, respectively,% \textit{\textbf{Case $1$}} and \textit{\textbf{Case $2$}} is shown in Figure~\ref{poissonapprox}, 
  where ``Dual-MMPP" is obtained from Dual~(\ref{eq:dualfinal}) in Section~\ref{sec:online}, ``Dual-Poisson" is obtained from~(\ref{eq:mod-dual-general}) with the exact mean $\mu_i=(\theta_{1i}r_{21i}+\theta_{2i}r_{12i})/(r_{12i}+r_{21i})$ is known and ``Online-Poission" is obtained from~(\ref{eq:mod-dual-general}) with estimated arrival rates as discussed in Section~\ref{sec:estimator}.   We can see that with large transition rates, the Poisson approximation performs better as compared to small transition rates.  This is due to the fact that our approximation becomes exact when transition rates go to infinity. However, our approximation yields similar optimal aggregate hit rate as compared to ``Dual-MMPP" even for small transition rates as shown in Table \ref{tbl:aggrRate}. We also numerically verify the case for intermediate transition rates by taking  $r_{12i} = 5\times10^{-5}$ and $r_{21i} = 2\times10^{-5}.$    Again, we can see that the optimal hit rates obtained through~(\ref{eq:mod-dual-general}) match those obtained from Dual under $2$-MMPP.  We omit the plot due to space limits.

\begin{small}
\begin{table}[tbp]
%\vspace{-0.05in}
\begin{tabular}{c c c|| c c c }
\hline
$x$&$n$&$B$&Dual-MMPP& Dual-Poisson &Online-Poisson\\%&Centralized\\
\hline
$10^{-3}$&$1000$&$100$&$0.1591$&$0.1612$&$0.1655$\\%&$0.02\%$\\
$10^{-7}$&$1000$&$100$&$0.1474$&$0.1427$&$0.1540$\\%&$-2.63$\\%&$11.8\%$\\
%$10^{-7}$&$1000$&$100$&$-8.9$&$-8.3$  &$-9.2256$ \\%&$-7.61$\\%&$6.7\%$\\
%$10^{-9}$&$30$&$10$&$-2.63$ &$-2.82$ ($7.2\%$) &$-3.33$ ($26.6\%$)\\%&$-7.61$\\%&$6.7\%$\\
\hline
\end{tabular}
\caption{Optimal aggregate hit rates for large ($x=10^{-3}$) and small ($x=10^{-7}$) state transition rates.}% ``Modified" and ``Online"  are obtained by~(\ref{eq:mod-dual-mmpp}) and ~(\ref{eq:mod-dual-general}), respectively.  }
\label{tbl:aggrRate}
\vspace{-0.25in}
\end{table}
\end{small}

\begin{remark}
We also considered the case when irts follow hyperexponential and weibull distributions.  Equation ~(\ref{eq:hrbcum2}) can be solved with Dual for both distributions. We compare results using ~(\ref{eq:hrbcum2}) with those obtained using ~(\ref{eq:mod-dual-general}) and we find that the optimal hit rates obtained through~(\ref{eq:mod-dual-general}) match those obtained solving ~(\ref{eq:hrbcum2}).  For ease of exposition, these results are relegated to Appendix~\ref{appd6}.
\end{remark}

\section{Trace-driven Simulation}\label{sec:tracesim}

In this section, we evaluate the accuracy of the reverse engineered dual implementation of LRU and compare the performance of LRU to that of Poisson approximate online algorithm through trace-driven simulation. We use requests from a web access trace collected from a gateway router at IBM research lab \cite{zerfos13}. The trace contains $3.5\times10^6$ requests with a content catalog of size $n = 5638.$ We consider a cache size $B = 1000.$

\begin{figure*}[htbp]
\centering
\begin{minipage}{.32\textwidth}
\centering
\includegraphics[width=1\linewidth]{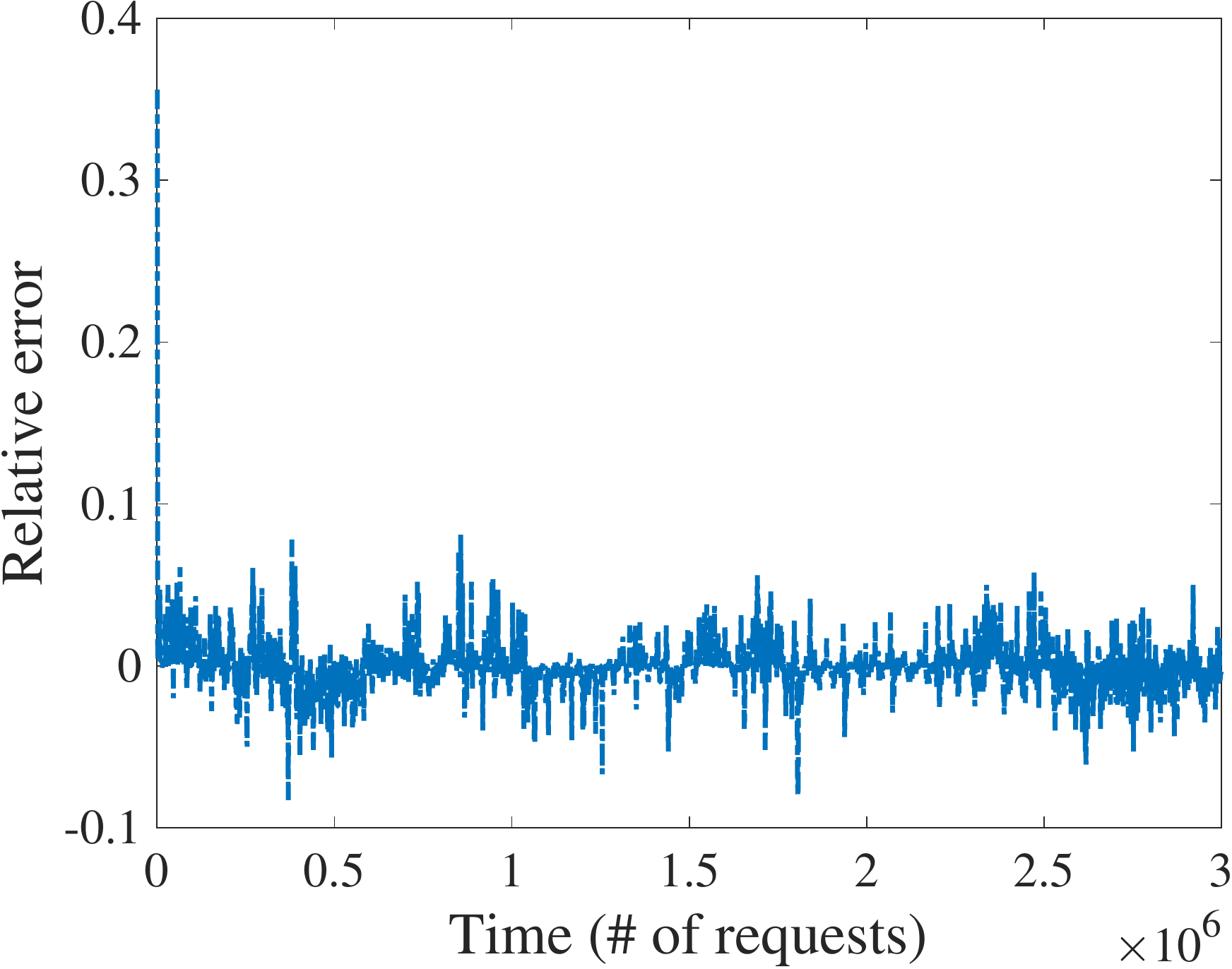}
\caption{Relative error in hit counts of LRU and the reverse engineered dual algorithm.}
\label{lru-reverse}
\end{minipage}\hfill
\begin{minipage}{.32\textwidth}
\centering
\includegraphics[width=1\textwidth, height = 0.8\textwidth]{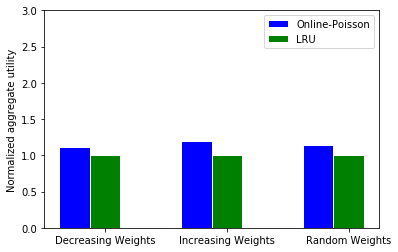}
\caption{Trace-driven comparison for online-Poisson and LRU caching policy.}
\label{poissonapprox3}
\end{minipage}\hfill
\begin{minipage}{.32\textwidth}
\centering
\includegraphics[width=1\textwidth, height = 0.8\textwidth]{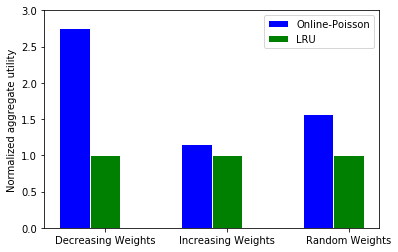}
\caption{Synthetic trace-driven comparison for online-Poisson and LRU caching policy.}
\label{poissonapprox4}
\end{minipage}
\vspace{-0.15in}
\end{figure*}

\subsection{Reverse Engineering}
We use the trace to compute cache hits for the replacement-based implementation of LRU and the implementation based on reverse engineered dual algorithm. We count the number of hits from each implementation over windows of $3000$ requests and compute the relative error.  From Figure \ref{lru-reverse}, it is clear that the relative error is small over time. Thus the implementation based on the reverse engineered dual algorithm performs close to its replacement-based implementation.

\subsection{Effect of content weights}
The utility function defined in \eqref{eq:utility} involves content weights, $w_i$, associated with each content $i.$
Classical cache replacement policies such as LRU are oblivious to content weights. However, the Poisson approximation based online algorithm updates the TTL timer by considering the content weight at each time step.  Thus the Poisson approximation based online algorithm is more robust to variation in content weights. Figure \ref{poissonapprox3} compares the performance of online Poisson algorithm to that of LRU across different sets of content weights, i.e. we consider the following three cases: (a) $w_i  = \mu_i$ (decreasing weights and decreasing request rates) (b) $w_i  = 1/\mu_i$ (increasing weights and decreasing request rates) (c) $w_i  = \texttt{rand}(0,1)$ (random weights and decreasing request rates). Let $U_{P}$ and $U_{L}$ denote the aggregate content utility for online Poisson algorithm and for LRU policy, respectively. We normalize both utilities w.r.t. LRU policy as $U_{P}/U_{L}$ and $U_{L}/U_{L} = 1$, respectively.  From Figure \ref{poissonapprox3}, it is clear that in each case online Poisson algorithm performs better than LRU, i.e. online Poisson algorithm achieves larger aggregate utility as compared to the LRU policy. %{\bf \blue{[better in what? in achieve larger utilities? ]}}
%{\bf \blue{[Comments: define the weights you use, e.g., what the increasing/decreasing weights are? Random weights?]}}  

We also consider a synthetic trace generated with a content catalog of size $n = 1000$ and irt distribution following a generalized Pareto distribution. The results are shown in Figure \ref{poissonapprox4}. It is clear from Figure \ref{poissonapprox4} that the Poisson approximation based online algorithm performs even better as compared to LRU when the request process is stationary. We also get similar performance benefits when compared to other classical replacement based caching policies such as FIFO and RANDOM. We omit them due to space constraints. %{\bf \blue{[Did you do all of them? LFU? I add the red sentence here. People may ask if you have them since we refer to an arXiv online report. If you have the results, maybe put them into the arXiv version.]}}

\section{Conclusion}\label{sec:concl}

In this paper, we associated each content with a utility that is a function of the corresponding content hit rate or hit probability, and formulated a cache utility maximization problem under stationary requests. We showed that this optimization problem is convex when the request process has a DHR. We presented explicitly optimal solutions for HRB-CUM and HPB-CUM, and made a comparison between them both theoretically and numerically. We also developed decentralized algorithms to implement the optimal policies. We found that HRB-CUM is more robust and stable than HPB-CUM w.r.t. convergence rate.  Finally, we proposed Poisson approximate online algorithms to different inter-request distributions, which is accurate and lightweight.  Going further, we aim at extending our results to consider \emph{Non-reset TTL Cache} where the timer is set only on a cache miss. {Non-reset TTL Caches} might have different implications on the design and performance analysis of distributed and online algorithms.  Establishing these results will be our future goal.

\section{Appendix}\label{appendix}
\subsection{HPB-CUM}\label{sec:hpb-cum}
Following a similar argument in Section~\ref{sec:utility},  we can formulate the following hit probability based optimization problem
\begin{align}\label{eq:hpbcum2}
\text{\bf{HPB-CUM:}}\quad \max_{0\leq h_i^p\leq 1} \sum_{i=1}^n U_i(h_i^p),\quad\text{s.t.} \sum_{i=1}^n g_i(h_i^{p}) \leq B,
\end{align}

The Lagrangian function can be written as 
\begin{align}
\mathcal{L}^p(\boldsymbol h^p, \eta^p)=\sum_{i=1}^n U_i(h_i^p)-\eta^p\left[\sum_{i=1}^n g_i(h_i^{p})-B\right],
\end{align}
where $\eta^p$ is the Lagrangian multiplier {and $\boldsymbol h^p=(h_1^p,\cdots,h_n^p)$}.  Similarly, the derivative of $\mathcal{L}^p(\boldsymbol h^p, \eta^p)$ w.r.t. $h_i^p$ for $i=1, \cdots, n,$ should satisfy the following condition so as to achieve its maximum
\begin{align}\label{eq:funv}
\eta^p=U_i^\prime(h_i^p)/g^\prime_i(h_i^p) \triangleq v_i(h_i^p),
\end{align}
where $v_i(\cdot)$ is a continuous and differentiable function on $[0, 1],$ i.e., there exists a one-to-one mapping between $\eta^p$ and $h_i^p$ if $0\le v_i^{-1}(\eta^p)\le1.$  Again, by the cache capacity constraint, we can compute $\eta^p$ through the following fixed-point equation
\begin{align}\label{eq:capacity-constraint-hpb}
\sum_{i=1}^n g_i(h_i^{p}) &= \sum_{i=1}^ng_i(v_i^{-1}(\eta^p)) = B.
\end{align}
Finally, given $\eta^p$, the timer, hit probability, and hit rate are
\begin{align}
 t_i = F_i^{-1}( v_i^{-1}(\eta^p)), \quad h_i^p= v_i^{-1}(\eta^p), \quad &\lambda_i^p=\mu_i v_i^{-1}(\eta^p), \quad i=1,\cdots,n.
\end{align}

\subsection{Proofs in Section~\ref{sec:utility}}\label{appa}
\subsubsection{Convexity of HRB-CUM~(\ref{eq:hrbcum}) and HPB-CUM~(\ref{eq:hpbcum2})} \label{appa-convex}
%In this section we present convexity property of the optimization formulation in \eqref{eq:hpbcum}.
In this section, we show that HRB-CUM~(\ref{eq:hrbcum}) and HPB-CUM~(\ref{eq:hpbcum2})} in terms of timers are non-convex. 
\begin{theorem*}
HRB-CUM~(\ref{eq:hrbcum}) and HPB-CUM~(\ref{eq:hpbcum2}) in terms of timers are non-convex.
\end{theorem*}
\begin{proof}
Recall that 
\begin{align*}
\hat{F}_i(t_i)=\mu_i\int_{0}^{t_i} (1-F(x))dx.
\end{align*}
Take the derivative w.r.t. $t_i,$ we have  
\begin{align}
\frac{\partial \hat{F}_i(t_i)}{\partial t_i}=\mu_i(1-F(t_i)),\quad \text{and} \quad \frac{\partial^2 \hat{F}_i(t_i)}{\partial t_i^2}=-\mu_if_i(t_i).
\end{align}
Since $f_i(\cdot)$ is the p.d.f. for the inter-request arrival time with $\mu_i\geq 0$, we have $f_i(t_i) \geq 0.$ Thus $\partial^2 \hat{F}_i(t_i)/\partial t_i^2 \le 0.$ Therefore, $\hat{F}_i(t_i)$ is concave in $t_i$ and  then~(\ref{eq:hrbcum})  is a non-convex optimization problem. Similarly, we can show that~(\ref{eq:hpbcum2}) is non-convex.  
\end{proof}

\subsubsection{Proof for Lemma \ref{lm:hazard}}\label{appa-lemma}
%\noindent {\bf Proof for Lemma \ref{lm:hazard}:}
%\begin{proof}
 Given~(\ref{eq:cdf-renewal}) and~(\ref{eq:age-distribution}), we have 
\begin{align}
\frac{\partial \hat{F}_i(t_i)}{\partial t_i}=\mu_i(1-F_i(t_i)).
\end{align}
Then
\begin{align}
\frac{\partial g_i(h_i^p)}{\partial h_i^p}&=\frac{\partial \hat{F}_i(F_i^{-1}(h_i^p))}{\partial h_i^p}
\stackrel{(a)}{=}\mu_i(1-F_i(F_i^{-1}(h_i^p)))\cdot\frac{\partial F_i^{-1}(h_i^p)}{\partial h_i^p}\nonumber\displaybreak[1]\\
&\stackrel{(b)}{=}\frac{\mu_i(1-F_i(F_i^{-1}(h_i^p)))}{f_i(F_i^{-1}(h_i^p))} =\frac{\mu_i}{\zeta_i( F_i^{-1}(h_i^p))},
\end{align}
where (a) and (b) hold true based on the chain-rule and the inverse function theorem over continuously differentiable function $F_i$, respectively.
%\end{proof}

\subsection{Proofs in Section~\ref{sec:distr}}\label{appb}
%\subsection{Evaluation of Age Distribution}\label{appb}
In this section, we derive expressions for the age distribution of different inter-request distributions, which are summarized in Table~\ref{tbltraff}.

\noindent{\textbf{\textit{Exponential Distribution:}}}
The c.d.f. for exponential distribution is 
\begin{equation}\label{eq:exponential}
F_i(t)=1-e^{-\mu_i t}, \quad t\geq 0,
\end{equation}
where $\mu_i$ is the rate parameter. Then the age distribution $\hat{F}_i(t)$ for $t\geq 0$ is 
\begin{align}
\hat{F}_i(t)&=\mu_i\int_0^t (1-F_i(\tau))d\tau =\mu_i\int_0^t e^{-\mu_i\tau}d\tau =\frac{\mu_i(1-e^{-\mu_i t})}{\mu_i} = 1-e^{-\mu_i t} = F_i(t).
\end{align} 

\noindent{\textbf{\textit{Generalized Pareto Distribution:}}}
The c.d.f. of generalized Pareto distribution is
\begin{equation}\label{eq:pareto45}
    F_i(t)=1 - \left[1+k_i(t-\theta_i)/\sigma_i\right]^{-1/k_i},\quad  t\geq \theta_i,
\end{equation}
where $k_i, \sigma_i$ and $\theta_i$ are shape, scale and location parameters, respectively.  We consider the case that  $0\leq k_i<1$, $\sigma_i \geq 0$, and $\theta_i = 0$ such that~(\ref{eq:pareto45}) has a DHR. It is well known that the mean satisfies $\mu_i = (1-k_i)/\sigma_i$ and the age distribution $\hat{F}_i(t)$ is 
\begin{align}
\hat{F}_i(t) &= \mu_i\int_0^t (1-F_i(\tau))d\tau =  \mu_i\int_0^t\left(1+k_it/\sigma_i\right)^{-1/k_i}\nonumber\displaybreak[1]\\
&= \mu_i\frac{\left(1+k_it/\sigma_i\right)^{-\frac{1}{k_i}+1}-1}{(k_i/\sigma_i)(-\frac{1}{k_i} + 1)} = 1 - \left(1 + {k_it/\sigma_i}\right)^{\frac{k_i-1}{k_i}}.
\end{align}

\noindent{\textbf{\textit{Hyperexponential Distribution:}}}
The c.d.f. of hyperexponential distribution is 
\begin{equation}\label{eq:hyperexpo45}
    F_i(t)=1-\sum\limits_{j=1}^{l}p_{ji}e^{-\theta_{ji}t},
\end{equation}
where $p_{ji}$ are phase probabilities and $\theta_{ji}$ are phase rates. The age distribution $\hat{F}_i(t)$ is
\begin{align}
\hat{F}_i(t) &= \mu_i\int_0^t (1-F_i(\tau))d\tau =  \mu_i\int_0^t\sum\limits_{j=1}^{l}p_{ji}e^{-\theta_{ji}\tau}d\tau = \mu_i\sum\limits_{j=1}^{l}\frac{p_{ji}}{\theta_{ji}}(1-e^{-\theta_{ji}t}).
\end{align}

\noindent{\textbf{\textit{Weibull Distribution:}}}
The c.d.f. for Weibull distribution is 
\begin{align}
F_i(t)= 1-e^{-\big(t/\theta_i\big)^{k_i}},
\end{align}
where $\theta_{i}$ and $k_i$ are scale and shape parameters, respectively. Then the age distribution $\hat{F}_i(t)$ is  
\begin{align}
\hat{F}_i(t) &= \mu_i\theta_i\int_{0}^{\frac{t}{\theta_i}}e^{-x^{k_i}}\;dx.
\end{align}
It is difficult to get a closed form of $\hat{F}_i(t)$ in general.  However, for a special case, $k_i = 0.5,$ we have
\begin{align}
\hat{F}_i(t) &= \mu_i\theta_i\int_{0}^{\frac{t}{\theta_i}}e^{-\sqrt{x}}\;dx = 2\mu_i\theta_i\left[1-e^{-\sqrt{\frac{t}{\theta_i}}}\left(\sqrt{\frac{t}{\theta_i}}+1\right)\right].
\end{align}

\noindent{\textbf{\textit{Uniform Distribution:}}}
The c.d.f. for uniform distribution is 
\begin{equation}
    F_i(t)=\frac{t}{b_i},\quad 0 \le t \le b_i,
\end{equation}
where $b_i$ is the uniform parameter and $\mu_i = 2/b_i.$ Then we have
\begin{align}
\hat{F}_i(t) &= \mu_i\int_{0}^{t}\left(1-\frac{t}{b_i}\right)d\tau =\mu_i\bigg(t - \frac{t^2}{2b_i}\bigg) = 2\left(t/b_i\right) - \left(t/b_i\right)^2.
\end{align}

\subsection{Proofs in Section~\ref{sec:dist}}\label{appc}
%\subsection{Performance Comparison}\label{appc}
In this section, we compare the performance of HRB-CUM and HPB-CUM under different utility functions and inter-request processes.
\subsubsection{Identical Distributions}
Here, we consider the performance comparison of HRB-CUM and HPB-CUM under identical inter-request process.

\noindent{\bf Proof of Theorem \ref{thm:identical}}
%\begin{proof}
 Under identical inter-request process, we have $F_i(\cdot)=F(\cdot)\;\forall i$. Hence $\hat{F}_i(\cdot)=\hat{F}(\cdot)$, i.e., $g_i(\cdot)=g(\cdot)\;\forall i$. Also $\mu_i = \mu\;\forall\;i.$ In HRB-CUM \eqref{eq:hrbcum}, we aim to maximize the objective $\sum_{i=1}^n U_i(\lambda_i^r).$  We can scale the objective as $\sum_{i=1}^n U_i(\lambda_i^r/\mu)$, while the solution of problem \eqref{eq:hrbcum} remains the same. By substituting $\lambda_i^r/\mu = h_i^p$,  \eqref{eq:hrbcum} and \eqref{eq:hpbcum2}, i.e. HRB-CUM and HPB-CUM are identical.
%\end{proof}

\subsubsection{$\beta$-fair Utility Functions}
Here, we consider $\beta$-fair utilities. First, we consider log utilities, i.e., $\beta = 1.$

\noindent {\bf Proof for Theorem \ref{thm:log}}
%\begin{proof}
 Consider $U_i(x) = w_i\log x$, i.e., $U_i^\prime(x) = w_i/x.$ Under HRB-CUM, from \eqref{eq:ydef}, it is clear that
\begin{align}
y_i(\lambda_i^r/\mu_i) = \frac{\mu_iw_i}{\lambda_i^rg_i^\prime(\lambda_i^r/\mu_i)} = \frac{U_i^\prime(\lambda_i^r/\mu_i)}{g_i^\prime(\lambda_i^r/\mu_i)} = v_i(\lambda_i^r/\mu_i).
\end{align}
Again by substituting $\lambda_i^r/\mu_i = h_i^p$, HRB-CUM and HPB-CUM are identical.
%\end{proof}

\noindent{\textbf{\textit{Exponential Distribution:}}} We compare HPB-CUM and HRB-CUM under exponential inter-request process.
%\red{\bf[Change all the things to non-increasing weights. Give a proof of that instead of uniform weights. The ideas are very similar, not a hard proof.]}

\noindent{\textit{Uniform weights:}} First we consider uniform weights, i.e., $w_i\equiv w$ for $i=1, \cdots, n.$  Then we have 
\begin{align}
&h^p_i = \frac{B}{n}, \quad h^r_i = \frac{\mu_i^{\frac{1}{\beta} - 1}}{\sum_{j=1}^n \mu_j^{\frac{1}{\beta} - 1}}B.
\end{align}
It is easy to check that $h^r_i$ is decreasing in $i$ for $\beta<1,$ and increasing in $i$ for $\beta>1.$ 

\begin{theorem*}\label{thm:uniform1}
When weights are uniform, (i) for $\beta<1,$ HRB-CUM favors more popular item compared to HPB-CUM, i.e., $\exists j\in(1,n)$ s.t. $h^r_i>h^p_i,,$ $\forall i<j$, and $h^r_i<h^p_i,$ $\forall i>j$; and (i) for $\beta>1,$ HRB-CUM favors less popular item compared to HPB-CUM, i.e., $\exists l\in(1,n)$ s.t. $h^r_i<h^p_i,$ $\forall i<l$, and $h^r_i>h^p_i,$ $\forall i>l$. In particular, if $j, l\in\mathbb{Z}^+,$ then $h^r_j=h^p_j,$ and $h^r_l=h^p_l.$
\end{theorem*}
\begin{proof}
We first consider $\beta<1$, i.e., $h^r_i$ is decreasing in $i$.  We have 
\begin{align}
&h^r_1= \frac{\mu_1^{\frac{1}{\beta} - 1}}{\sum_{j=1}^n\mu_j^{\frac{1}{\beta} - 1}}B> \frac{\mu_1^{\frac{1}{\beta} - 1}}{n\mu_1^{\frac{1}{\beta} - 1}}B=\frac{B}{n}=h^p_1, \quad h^r_n= \frac{\mu_n^{\frac{1}{\beta} - 1}}{\sum_{j=1}^n \mu_j^{\frac{1}{\beta} - 1}}B< \frac{\mu_n^{\frac{1}{\beta} - 1}}{n\mu_n^{\frac{1}{\beta} - 1}}B=\frac{B}{n}=h^p_n.
\end{align}
Since $h^r_i$ is decreasing in $i$ and $h^p_i=\frac{B}{n}$ for any $i=1, \cdots, n,$ thus, there must exist an intersection point $1<j<n$ such that $h^r_j=h^p_j,$ satisfying that $h^r_k>h^p_k$ for $k=1, \cdots, j-1$ and $h^r_k<h^p_k$ for $k=j+1, \cdots, n$.

Therefore, when $\beta<1,$ we know that HRB-CUM favors more popular item compared to HPB-CUM.

Similarly, when $\beta>1,$ $h^r_i$ is increasing in $i$.  We have 
\begin{align}
&h^r_1= \frac{\mu_1^{\frac{1}{\beta} - 1}}{\sum_{j=1}^n\mu_j^{\frac{1}{\beta} - 1}}B< \frac{\mu_1^{\frac{1}{\beta} - 1}}{n\mu_1^{\frac{1}{\beta} - 1}}B=\frac{B}{n}=h^p_1, \quad h^r_n= \frac{\lambda_n^{\frac{1}{\beta} - 1}}{\sum_{j=1}^n \mu_j^{\frac{1}{\beta} - 1}}B> \frac{\mu_n^{\frac{1}{\beta} - 1}}{n\mu_n^{\frac{1}{\beta} - 1}}B=\frac{B}{n}=h^p_n.
\end{align}
Again, as $h^r_i$ is increasing in $i$ and $h^p_i=\frac{B}{n}$ for any $i=1, \cdots, n,$ thus, there must exist an intersection point $1<l<n$ such that $h^r_i=h^p_i,$ satisfying that $h^r_k<h^p_k$ for $k=1, \cdots, l-1$ and $h^r_k>h^p_k$ for $k=l+1, \cdots, n$.

Therefore, when $\beta>1,$ we know that HRB-CUM favors less popular item compared to HPB-CUM. 
\end{proof}

Now we make a comparison between the hit rate under these two approaches. 
\begin{theorem*}\label{thm:uniform2}
Under the uniform weight distribution, (i) for $\beta<1,$ hit rate based utility maximization approach favors more popular item compared to hit probability based utility approach, i.e., $\exists \tilde{j}\in(1,n)$ s.t. $\lambda^r_i>\lambda^p_i$ $\forall i<\tilde{j}$; and (i) for $\beta>1,$ hit rate based utility maximization approach favors less popular item compared to hit probability based utility approach, i.e., $\exists \tilde{l}\in(1, n)$ s.t. $\lambda^r_i>\lambda^p_i,$ $\forall i>\tilde{l}.$   In particular, if $\tilde{j}, \tilde{l}\in\mathbb{Z}^+,$ then $\lambda^r_{\tilde{j}}=\lambda^p_{\tilde{j}},$ and $\lambda^r_{\tilde{l}}=\lambda^p_{\tilde{l}}.$
\end{theorem*}

\begin{proof}
We know that
\begin{align}
\lambda^p_i = \mu_ih^p_i,\;\lambda^r_i = \mu_ih^r_i.\nonumber
\end{align}
Based on the relation between $h^p_i$ and $h^r_i$ proved in previous theorem, similar results can be obtained for hit rates: $\lambda^p_i$ and $\lambda^r_i$.
\end{proof}

\noindent{\textit{Monotone non-increasing weights:}} 
%Consider the case where the weights increase with request rate.  
%It is reasonable to consider a monotone weight distribution 
Since we usually weight more on more popular content, we consider monotone non-increasing weights, i.e., $w_1\geq\cdots\geq w_n,$ given $\mu_1\geq\cdots\geq\mu_n.$ In such a case, we have
\begin{align}
&h^p_i = \frac{w_i^{1/\beta}}{\sum_{j} w_j^{1/\beta}}B, \quad h^r_i = \frac{w_i^{1/\beta} \mu_i^{1/\beta-1}}{\sum_{j} w_j^{1/\beta} \mu_j^{1/\beta - 1}}B.
\end{align}
It is easy to check that $h^p_i, h^r_i$ are decreasing in $i$ for $\beta<1,$ and increasing in $i$ for $\beta>1.$ 

Following the same arguments in proofs of Theorems~\ref{thm:uniform1} and~\ref{thm:uniform2}, we can prove Theorema~\ref{thm:decreasing-weight-hit-rate} and~\ref{thm:decreasing-weight-hit-prob}, hence are omitted here. 

\noindent {\bf Proof for Theorem \ref{thm:decreasing-weight-hit-prob}:} %{\it When weights are monotone decreasing, (i) for $\beta<1,$ HRB-CUM favors more popular content compared to HPB-CUM, i.e., $\exists j\in (1, n)$ s.t. $h^r_i>h^p_i,$ $\forall i<j$, and $h^r_i<h^p_i,$ $\forall i>j;$ and (ii) for $\beta>1,$ HRB-CUM favors less popular content compared to HPB-CUM, i.e., $\exists l\in (1, n)$ s.t. $h^r_i<h^p_i,$ $\forall i<l$, and $h^r_i>h^p_i,$ $\forall i>l.$ In particular, if $j, l\in\mathbb{Z}^+,$ then $h^r_j=h^p_j,$ and $h^r_l=h^p_l.$}

\begin{proof}
Above theorem can be proved in a similar manner to that of uniform distribution.
\end{proof}

\noindent {\bf Proof for Theorem \ref{thm:decreasing-weight-hit-rate}:} %{\it When weights are monotone decreasing, (i) for $\beta<1,$ HRB-CUM favors more popular contents compared to HPB-CUM, i.e., $\exists \tilde{j}\in(1,n)$ s.t. $\lambda^r_i>\lambda^p_i,$ $\forall i<\tilde{j}$; and (ii) for $\beta>1,$ HRB-CUM favors less popular contents compared to HPB-CUM, i.e., $\exists \tilde{l}\in(1, n)$ s.t. $\lambda^r_i>\lambda^p_i,$ $\forall i>\tilde{l}.$ In particular, if $\tilde{j}, \tilde{l}\in\mathbb{Z}^+,$ then $\lambda^r_{\tilde{j}}=\lambda^p_{\tilde{j}},$ and $\lambda^r_{\tilde{l}}=\lambda^p_{\tilde{l}}.$}

\begin{proof}
Above theorem can be proved in a similar manner to that of uniform distribution.
\end{proof}

\subsection{Decentralized Algorithms in Section~\ref{sec:online}}\label{appd}

\subsubsection{Non-linear Equations for Dual in HRB-CUM}\label{appd1}
We show the existence of a solution of~(\ref{eq:paretofpt}). 
\begin{theorem*}
For any $\eta^{(k)}, w_i > 0$ and $0\le k_i\le1,$ there always exists a unique solution in $[0, 1]$ for 
\begin{align}\label{eq:appendfxdpt}
e(h_i^{{(k)}}) = \mu_i^{1-\beta}\frac{w_i\big(1-h_i^{{(k)}}\big)^{k_i}}{\eta^{(k)}(1-k_i)} - \big(h_i^{{(k)}}\big)^{\beta} = 0.
\end{align}
\end{theorem*}
\begin{proof}
For $h_i^{{(k)}}=0$ and $h_i^{{(k)}}=1,$ we have 
\begin{align}
e(0) = \frac{\mu_i^{1-\beta}w_i}{\eta^{(k)}(1-k_i)},\quad e(1) = -1. \nonumber
\end{align}
Furthermore,  
\begin{align}
e^\prime(h_i^{{(k)}}) = -\frac{\mu_i^{1-\beta}w_ik_i(1-h_i^{{(k)}})^{k_i-1}}{\eta^{(k)}(1-k_i)} -& \beta (h_i^{{(k)}})^{\beta-1} < 0, \quad \forall h_i^{{(k)}} \in [0, 1]. 
\end{align}
Thus $e(\cdot)$ is decreasing in $h_i^{{(k)}}$. Since $h_i^{{(k)}} \in [0, 1]$,  $e(0) > 0$ and $e(1) < 0,$ therefore, there always exists a unique solution to~(\ref{eq:appendfxdpt}) in $[0, 1].$
\end{proof}

%\begin{theorem*}
%For any $\eta^{(k)}, w_i, \theta_i, \mu_i > 0,$ there always exists a unique solution in $[0, 1]$ for 
%\begin{align}
%e(h_i^{{(k)}}) = \frac{w_i}{2\theta_i\mu_i\eta^{(k)}}\left[\frac{1}{\log \bigg(\frac{1}{1-h_i^{{(k)}}}\bigg)}\right] - \big(h_i^{{(k)}}\big)^{\beta} = 0.
%\end{align}
%\end{theorem*}
%
%\begin{proof}
%%A similar argument as of the previous theorem proves the above theorem.
%The proof follows a similar argument as above, hence are omitted here. 
%\end{proof}

\subsubsection{Decentralized Algorithms for HRB-CUM}\label{appd2}
In the following, we develop primal and primal-dual algorithms for HRB-CUM under stationary request processes.

\noindent{\textbf{\textit{Primal Algorithm:}}}
Under the primal approach, we append a cost to the sum of utilities as
\begin{align}\label{eq:primal}
W(\boldsymbol \lambda) = \sum_{i=1}^n U_i(\lambda_i)-C\left(\sum_{i=1}^n g_i\left(\lambda_i/\mu_i\right)-B\right),
\end{align}
where $C(\cdot)$ is a convex and non-decreasing penalty function denoting the cost for extra cache storage.  When $g_i(\cdot)$ is convex, by the composition property, $W(\cdot)$ is strictly concave in $\boldsymbol \lambda$. We use standard \emph{gradient ascent} as follows.  

The gradient is given as
\begin{align}
\frac{\partial W(\boldsymbol \lambda)}{\partial \lambda_i}= U_i^\prime(\lambda_i)-\frac{g_i^\prime\left(\lambda_i/\mu_i\right)}{\mu_i}C^\prime\left(\sum_{i=1}^n g_i\left(\lambda_i/\mu_i\right)-B\right).
\end{align}
We also have $\partial \lambda_i/\partial t_i = \partial \mu_iF_i(t_i)/\partial t_i = \mu_if_i(t_i) > 0.$ Hence we move $t_i$ in the direction of gradient and the primal algorithm is given by
\begin{small}
\begin{align}
&\lambda_i^{(k)} = \mu_iF_i\left(t_i^{(k)}\right), \nonumber\\
&t_i^{(k+1)}\leftarrow \max\bigg\{0, t_i^{(k)}+\delta_i\bigg[U_i^\prime(\lambda_i^{(k)})- \frac{g_i^\prime\left(\lambda_i/\mu_i\right)}{\mu_i}C^\prime\big(B_{\text{curr}}-B\big)\bigg]\bigg\},
\end{align}
\end{small}

\noindent where $\delta_i = \rho_i(\partial \lambda_i^{(k)}/\partial t_i) = \rho_i\mu_if_i(t_i^{(k)})$, and $\rho_i\ge0$ is the step size.

\noindent{\textbf{\textit{Primal-Dual Algorithm:}}}
The dual and primal algorithms can be combined to form the primal-dual algorithm. For HRB-CUM, we have 
\begin{align}
&t_i^{(k+1)}\leftarrow \max\bigg\{0, t_i^{(k)}+\delta_i\bigg[U_i^\prime(\lambda_i^{(k)})- \frac{g_i^\prime\left(\lambda_i/\mu_i\right)}{\mu_i}\eta^{{(k)}}\bigg]\bigg\}, \nonumber\\
&\eta^{{(k+1)}}\leftarrow \max\left\{0, \eta^{{(k)}}+\gamma\left[B_{\text{curr}}-B\right]\right\}.
\end{align}

\subsubsection{Decentralized Algorithms for HPB-CUM}\label{appd3}
In the following, we develop decentralized algorithms for HPB-CUM under stationary request processes.

\noindent{\textbf{\textit{Dual Algorithm:}}}
%\subsubsection{HPB-CUM}
{For a request arrival process with a DHR inter-request distribution,~(\ref{eq:hpbcum2}) becomes a convex optimization problem as discussed in Section~\ref{sec:distr}, and hence solving the dual problem produces the optimal solution. } Since $0<t_i<\infty,$ then $0<h_i<1$ and $0<g_i(h_i)<1.$ Therefore, the Lagrange dual function is 
\begin{align}
D(\eta)=\max_{h_i}\left\{\sum_{i=1}^n U_i(h_i)-\eta\left[\sum_{i=1}^n g_i(h_i)-B\right]\right\},
\end{align}
and the dual problem is 
\begin{align}
\min_{\eta\geq 0}\quad D(\eta).
\end{align}

%Here, we use the standard \emph{gradient descent algorithm} to minimize $D(\eta^p)$ by moving the decision variables gradually along the gradient direction.  

%Taking the derivate of $D(\eta^p)$ with respect to $\eta^p,$ we obtain the gradient
%\begin{align}
%\frac{\partial D(\eta^p)}{\partial \eta^p}=-\left[\sum_{i=1}^n g_i(h_i^p)-B\right],
%\end{align}
{Following the standard \emph{gradient descent algorithm} by taking the derivate of $D(\eta)$ w.r.t. $\eta,$  the dual variable $\eta$ should be updated as }
\begin{align}
\eta^{{(k+1)}}\leftarrow \max\left\{0, \eta^{{(k)}}+\gamma\left[\sum_{i=1}^n g_i(h_i^{{(k)}})-B\right]\right\},
\end{align}
where $k$ is the iteration number, $\gamma>0$ is the step size at each iteration and $\eta\geq0$ due to KKT conditions.%$k$ is the iteration number.

Based on the results in Section~\ref{sec:utility}, in order to achieve optimality, we must have
\begin{align*}
\eta^{{(k)}}=\frac{U_i^\prime(h_i^{{(k)}})}{g^\prime_i(h_i^{{(k)}})} \triangleq v_i(h_i^{{(k)}}),\quad\text{i.e.},\quad h_i^{{(k)}}=v_i^{-1}(\eta^{{(k)}}).
\end{align*}
%i.e., $h_i^{p^{(k)}}=v_i^{-1}(\eta^{p^{(k)}}).$ 

\noindent{\textit{\textbf{Poisson Process:}}} {Under a Poisson request process, we have $g^\prime_i(h_i^{{(k)}}) = 1,$ and $h_i^{{(k)}} = U_i^{\prime-1}(\eta^{{(k)}})$, consistent with the results in \cite{dehghan16}. }

\noindent{\textit{\textbf{Generalized Pareto Distribution:}}} {When inter-request times are described by a generalized Pareto distribution and utilities are $\beta$-fair, $h_i^{{(k)}}$ can be obtained through 
\begin{align}%\label{eq:paretofpt}
\frac{w_i\left(1-h_i^{{(k)}}\right)^{k_i}}{\eta^{{(k)}}(1-k_i)} - \left(h_i^{{(k)}}\right)^{\beta}= 0.
\end{align}
%We can show that there exists a finite valued solution in $[0,1]$ for any $\eta^{{(k)}} > 0$; details given in \cite{nitishjian-infocom18-tech}.

\noindent{\textit{\textbf{Weibull Distribution:}}} {When inter-request times are described by a Weibull distribution with shape parameter $k_i=0.5$ and utilities are $\beta$-fair, $h_i^{{(k)}}$ can be obtained through
\begin{align}%\label{eq:wblfpt}
\frac{w_i}{2\theta_i\mu_i\eta^{{(k)}}}\left[\frac{1}{\log\left(\frac{1}{1-h_i^{{(k)}}}\right)}\right] -\left(h_i^{{(k)}}\right)^{\beta} = 0.
\end{align}
%Similarly, there always exists a solution in $[0,1]$ for any $\eta^{{(k)}} > 0$, details given in \cite{nitishjian-infocom18-tech}.

%For example, under generalized pareto inter arrival request process with $\beta$ fair utility functions, $h_i^{p^{(k)}}$ is the solution of the equation
%
%\begin{align}\label{eq:paretofpt}
%\frac{w_i(1-x)^{k_i}}{\eta^{p^{(k)}}(1-k_i)} - x^{\beta} = 0
%\end{align}

%\noindent Similarly, for weibull inter arrival request process with shape parameter $= 0.5$,  $h_i^{p^{(k)}}$ is the solution of the equation
%\begin{align}\label{eq:wblfpt}
%\frac{w_i}{2\theta_i\mu_i\eta^{p^{(k)}}}\left[\frac{1}{\log (\frac{1}{1-x})}\right] - x^{\beta} = 0
%\end{align}
%
%\noindent It can be shown that both Equations \eqref{eq:paretofpt} and \eqref{eq:wblfpt} have a finite solution in domain $[0,1]$ for any $\eta^{p^{(k)}} > 0$  (See Appendix \ref{sec:appendix}). 

%Note that, under exponential inter arrival request process, $g^\prime_i(h_i^{p^{(k)}}) = 1.$ Hence $h_i^{p^{(k)}} = U_i^{\prime-1}(\eta^{p^{(k)}})$, which is in accordance with the results derived in \cite{dehghan16}. 

{Since $g_i(h_i^{{(k)}})$ indicates the probability that content $i$ is in the cache, $\sum_{i=1}^n g_i(h_i^{{(k)}})$ represents the number of contents currently in the cache, denoted as $B_{\text{curr}}$.}  Therefore, the dual algorithm for reset TTL caches is
\begin{align}%\label{eq:dualfinal}
&t_i^{(k)}=F_i^{-1}( v_i^{-1}(\eta^{{(k)}})),\nonumber\\
&\eta^{{(k+1)}}\leftarrow \max\left\{0, \eta^{{(k)}}+\gamma\left[B_{\text{curr}}-B\right]\right\},
\end{align}
where the iteration number $k$ is incremented upon each request arrival.

\noindent{\textbf{\textit{Primal Algorithm:}}}
%\subsubsection{HPB-CUM}
Under the primal approach, we append a cost to the sum of utilities as
\begin{align}\label{eq:primal}
W(\boldsymbol h) = \sum_{i=1}^n U_i(h_i)-C\left(\sum_{i=1}^n g_i(h_i)-B\right),
\end{align}
{where $C(\cdot)$ is a convex and non-decreasing penalty function denoting the cost for extra cache storage.  {When $g_i(\cdot)$ is convex, by the composition property \cite{boyd04}}, $W(\cdot)$ is strictly concave in $\boldsymbol h$. We use standard \emph{gradient ascent} as follows.}  
%Assuming $g_i(\cdot)$ as convex, the composition of $C(\cdot)$ and $g_i(\cdot)$ would be convex \cite{boyd04}. Hence $W(\cdot)$ is strictly concave in $h^p$. We can adopt standard gradient ascent algorithm as follows.

The gradient is given as
\begin{align}
\frac{\partial W(\boldsymbol h)}{\partial h_i}= U_i^\prime(h_i)-g_i^\prime(h_i)C^\prime\left(\sum_{i=1}^n g_i(h_i)-B\right).
\end{align}
We also have $\partial h_i/\partial t_i = \partial F_i(t_i)/\partial t_i = f_i(t_i) > 0.$ Hence we move $t_i$ in the direction of gradient and the primal algorithm is given by
%\begin{footnotesize}
\begin{small}
\begin{align}
&h_i^{{(k)}} = F_i\left(t_i^{(k)}\right), \nonumber\\
&t_i^{(k+1)}\leftarrow \max\bigg\{0, t_i^{(k)}+\delta_i\bigg[U_i^\prime(h_i^{{(k)}})- g_i^\prime(h_i^{{(k)}})C^\prime\big(B_{\text{curr}}-B\big)\bigg]\bigg\},
\end{align}
\end{small}
%\end{footnotesize}
%\red{\bf [Comments: when are these operations taken?]}

\noindent where $\delta_i = \rho_i(\partial h_i/\partial t_i) = \rho_i f_i(t_i^{(k)})$, $\rho_i\ge0$ is the step size, and $k$ is the iteration number   incremented upon each request arrival.

\noindent{\textbf{\textit{Primal-Dual Algorithm:}}}
The dual and primal algorithms can be combined to form the primal-dual algorithm. For HPB-CUM
\begin{align}
&t_i^{(k+1)}\leftarrow \max\bigg\{0, t_i^{(k)}+\delta_i\bigg[U_i^\prime(h_i^{{(k)}})- g_i^\prime(h_i^{{(k)}})\eta^{{(k)}}\bigg]\bigg\}, \nonumber\\
&\eta^{{(k+1)}}\leftarrow \max\left\{0, \eta^{{(k)}}+\gamma\left[B_{\text{curr}}-B\right]\right\}.
\end{align}

\subsection{Stability Analysis for Poisson Arrivals} \label{stab:poisson234}
We consider the case $\gamma = \gamma(\eta)$ such that the update rule \eqref{eq:dual-intermidiate} is globally stable around its equilibrium $\eta^*.$ We define the such a function $\gamma(\eta)$ below. 
 
 \noindent {\bf Case 1: $\eta < \eta^*$:} W.l.o.g consider $\eta = \frac{\eta^*}{m}$ where $m>1.$ In this scenario, we find a scaling parameter $\gamma = \gamma_m$ for which $\Delta V(\frac{\eta^*}{m}) < 0.$ We evaluate $\Delta V(\frac{\eta^*}{m})$ as follows. 
%\red{\bf [Comments: do not use one word, it should be a short sentence.  for the following equation, break it into two line, matching the two lines with ``=".  Similar comments to equation 35 as equation 33. make the logic a bit smooth]}
\begin{align}
\Delta V\left(\frac{\eta^*}{m}\right) &= -W\log\left[1+\frac{m\gamma_m}{\eta^*}\bigg(\frac{mW}{\eta^*}-B\bigg)\right]+B\gamma_m\left(\frac{mW}{\eta^*}-B\right)\nonumber\\ 
%&=-W\log\left[1+\frac{k(k-1)B^2\gamma_k}{W}\right]+B^2\gamma_k(k-1)\nonumber\\
&=-W\log\bigg[1+m(m-1)\hat{\gamma}_m\bigg]+W(m-1)\hat{\gamma}_m.\nonumber
\end{align}

\noindent where $\hat{\gamma}_m = \frac{B^2\gamma_m}{W}$. Thus $\Delta V\left(\frac{\eta^*}{m}\right) < 0$ only if                    
\begin{align}
&-W\log\bigg[1+m(m-1)\hat{\gamma}_m\bigg]+W(m-1)\hat{\gamma}_m < 0\nonumber\\
%&\implies \log\bigg[1+k(k-1)\hat{\gamma_k}\bigg]-(k-1)\hat{\gamma_k} > 0\nonumber\\
&\iff 1+m(m-1)\hat{\gamma}_m>e^{(m-1)\hat{\gamma}_m}.\label{eq:k_alpha}
\end{align}

Denote the left hand side ({\it L.H.S.}) of \eqref{eq:k_alpha} as function $l_1(\hat{\gamma}_m) = 1+m(m-1)\hat{\gamma}_m$ and the right hand side ({\it R.H.S.}) as $l_2(\hat{\gamma}_m) = e^{(m-1)\hat{\gamma}_m}$. $l_1(\hat{\gamma}_m)$ is a straight line with positive slope, say,  $m_1 = m(m-1)$ and $y$-intercept $1.$ $l_2(\hat{\gamma}_m)$ is an exponentially growing function with slope, say, $m_2 = m-1$ at $\hat{\gamma}_m = 0$ and $y$-intercept $1.$ We have $m>1, m_1>m_2.$ Also as $l_2(\hat{\gamma}_m)$ grows exponentially, it eventually intersects $l_1(\hat{\gamma}_m)$ at some point $\hat{\gamma}_m = \hat{\gamma}_m^*$. Thus for $\hat{\gamma}_m \in (0,\hat{\gamma}_m^*)$, $\Delta V(\frac{\eta^*}{m}) < 0$ is satisfied. Here $\hat{\gamma}_m^*$ is the non-zero solution of the fixed point equation $1+m(m-1)x = e^{(m-1)x}$.%\red{\bf [Comments: strongly suggest to polish this paragraph. it is too odd.]}

\noindent {\bf Case 2: $\eta > \eta^*$:} Again, consider $\eta = \frac{\eta^*}{m}$ where $m<1.$ In this scenario, we can likewise find a scaling parameter $\gamma = \gamma_m$ for which $\Delta V(\frac{\eta^*}{m}) < 0.$ Proceeding similar to the analysis as that in case $1$ we get:

\begin{align}
\Delta V\left(\frac{\eta^*}{m}\right) < 0 \;\text{or}\; 1+m(m-1)\hat{\gamma}_m>e^{(m-1)\hat{\gamma}_m}.\label{eq:k_alpha3}
\end{align}

\noindent Again for $\hat{\gamma}_m \in (0,\hat{\gamma}_m^*)$, \eqref{eq:k_alpha3} is satisfied where $\hat{\gamma}_m^*$ is the non-zero solution of the fixed point equation $1+m(m-1)x = e^{(m-1)x}.$\\

\subsection{Stability Analysis for Pareto Arrivals} \label{stab:pareto23}
When requests arrive according to a Pareto distribution we have
\begin{align}
D(\eta)=\sum_{i=1}^n w_i\log (\mu_iy_i^{-1}(\eta))-\eta\left[\bigg(\sum_{i=1}^n 1 - \left(1 -y_i^{-1}(\eta)\right)^{1-k_i}\bigg) -B\right]. \label{eq:dualfunctpareto}
\end{align}
\noindent where $y_i(x) \triangleq \frac{w_i(1-x)^{k_i}}{x(1-k_i)}$ and $k_i$ is the Pareto scaling parameter. Assume $k_i \le 0.5$ for finite variance of the Pareto distribution. %\red{\bf [Comments: $k_i$ is.... do not use ":"]}  
Thus we have the following online dual algorithm.
\begin{align}\label{eq:paretofpt}
\eta^{(m+1)}\leftarrow \eta^{(m)}+\gamma\left[\bigg(\sum_{i=1}^n 1 - \left(1 -y_i^{-1}(\eta^{(m)})\right)^{1-k_i}\bigg)-B\right].
\end{align}

\subsubsection{Local Stability Analysis}
Here we focus on the local stability analysis for problem \eqref{eq:paretofpt}. Denote $f(\eta) = \eta+\gamma\big[\big(\sum_{i=1}^n 1 - \big(1 -y_i^{-1}(\eta)\big)^{1-k_i}\big)-B\big]$ with $f:\mathbb{R}^+\rightarrow \mathbb{R}.$ The function $f(\eta)$ can be linearized around the equilibrium point $\eta^*$ as:
\begin{align}
\eta_{\delta}^{(m+1)} =f^{\prime}(\eta^*)\eta_{\delta}^{(m)}=[f^{\prime}(\eta^*)]^m \eta_{\delta}^{0},\label{eq:lin2}
\end{align}
\noindent where $\eta_{\delta}^{(m)} = \eta^{(m)}-\eta^*$ is deviation from $\eta^*$ at $m^{th}$ iteration. Assuming $\eta_{\delta}^{(0)}\approx 0$, $\eta_{\delta}^{(m+1)}\approx \eta_{\delta}^{(m)}$ only when $|f^{\prime}(\eta^*)| < 1.$ Thus for local asymptotic stability we have
\begin{align}
|f^{\prime}(\eta^*)| < 1\iff \left|1+\gamma\left[\sum_{i=1}^n (1-k_i)(1-y_i^{-1}(\eta^*))^{-k_i}(y_i^{-1}(\eta^*))^\prime\right]\right| < 1.\label{eq:gamma_lstb1}
\end{align}
\begin{remark}
Note that the local stability condition $|f^{\prime}(\eta^*)| < 1$ is valid across any irt distribution with decreasing hazard rate.
\end{remark}

We have the following lemma and theorem.

\begin{lemma}\label{lm:pareto}
Suppose that $y$ has an inverse function $y^{-1}$. If $y$ is differentiable at $y^{-1}(\eta)$ and $y^\prime[y^{-1}(\eta)]\ne 0$, then $y^{-1}$ is differentiable at $\eta$ and the following differentiation formula holds.
\begin{align}
(y_i^{-1}(\eta))^\prime = \frac{d}{d\eta}y^{-1}(\eta) = \frac{1}{y^\prime[y^{-1}(\eta)]}.\label{eq:gamma_lstb10}
\end{align}
\end{lemma}

Combining \eqref{eq:gamma_lstb1} and \eqref{eq:gamma_lstb10} yields the following condition for local stability.
\begin{align}
\gamma < \frac{2}{\sum_{i=1}^n \frac{(1-k_i)^2[x_i^*]^2[1-x_i^*]^{1-2k_i}}{w_i[1-x_i^*(1-k_i)]}} = \frac{2}{\sum_{i=1}^n A^*_i(x_i^*)}, \label{eq:loc_stab_pareto10}
\end{align}
\noindent where $x_i^* = y_i^{-1}(\eta^*).$ It can be shown that $\frac{d^2}{d(x_i^*)^2}A_i^*(x_i^*) < 0$ when $0 < x_i^* < 1.$ Also $A_i^*(0)=A_i^*(1) = 0.$ The function $A_i^*(x_i^*)$ has a unique maximum in $(0,1).$ Let the maximum occurs at $x_i^* = \bar{x}_i$ which can be found out by solving the equation $\frac{d}{d\bar{x}_i}A_i^*(\bar{x}_i) = 0.$ Thus, when
\begin{align}
\gamma < \frac{2}{n\max_{i} A^*_i(\bar{x}_i)}\label{eq:loc_stab_pareto100}
\end{align}
\eqref{eq:paretofpt} is locally  asymptotically stable.

\subsubsection{Evidence of Global Stability}

\begin{figure}%[htbp]
\centering
%\hspace{-0.5cm}
\begin{minipage}{0.45\textwidth}
\includegraphics[width=1\textwidth]{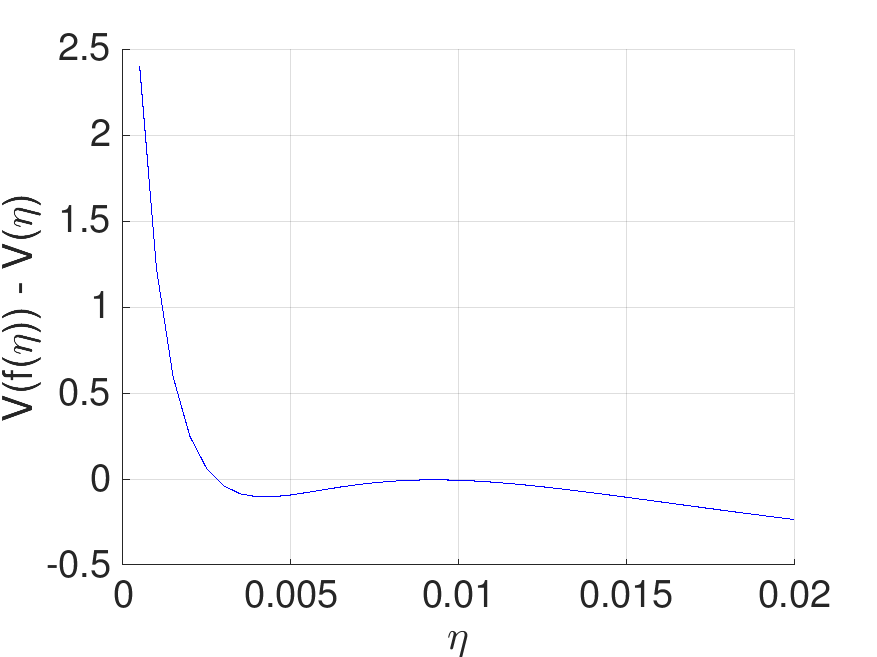}
\end{minipage}%\hfill
\begin{minipage}{0.45\textwidth}
\includegraphics[width=1\textwidth]{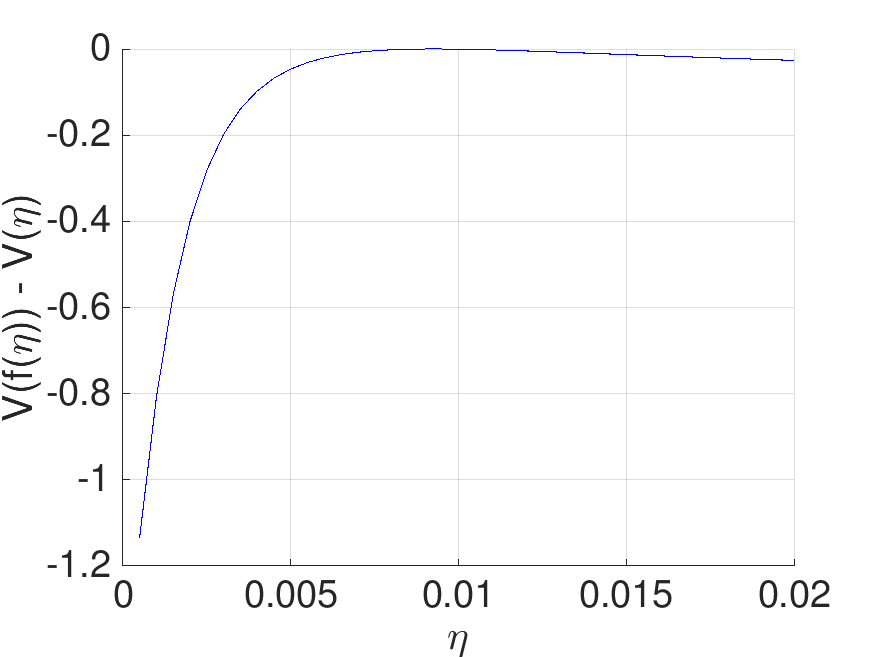}
\end{minipage}
\vspace{-0.05in}
\caption{Behavior of Lyapunov Derivative with $\gamma = 10^{-4}$ (Left) and $\gamma = 10^{-5}$ (Right).}
\label{fig:globex}
\vspace{-0.15in}
\end{figure}

Again we consider the candidate Lyapunov function $V(\eta) = D(\eta) - D(\eta^*)$. %\red{\bf [Comments: it is better to use $\mathbb{R}$]} 
By discrete time Lyapunov function theory \cite{hahn58}, for global asymptotic stability, we require $\Delta V(\eta) = V(f(\eta))-V(\eta) < 0, \forall\;\eta>0, \eta\ne\eta^*$. We evaluate $V(f(\eta))-V(\eta)$ with $V(\eta) = D(\eta) - D(\eta^*)$ for the case when $k_i \le 0.5$ as follows. %\red{\bf [Comments: similar comments on $V()$ or $V[]$]}

We plot $\Delta V(\eta)$ versus $\eta$ for a particular choice of parameters ($W = 1, \gamma = 10^{-4}\; \text{(Figure \ref{fig:globex}: Left) and}\;\gamma = 10^{-5}\;\text{(Figure \ref{fig:globex}: Right)}, B = 100, k_i = 0.5$). From the plot, it is clear that $\Delta V(\eta) \nless 0 \;\forall\;\eta>0, \eta\ne\eta^*$ when $\gamma = 10^{-4}.$ However, for sufficiently small $\gamma$ (say, $\gamma \le 10^{-5}$) then $\Delta V(\eta)$  is sufficiently negative and \eqref{eq:paretofpt} is globally  asymptotically stable as depicted in Figure \ref{fig:globex} (Right).

\subsection{Proofs in Section~\ref{sec:stabdec}}\label{app-stab-dec}

\begin{lemma}\label{lemma:const_gamma}
$\Delta V(\eta)$ is a decreasing function of $\eta$ for $\eta \in [0,\eta^*].$ %\red{\bf[Comments:``for $\eta\in[0,\eta^*]$"]}
\end{lemma}

\begin{proof}
Consider the derivative of the function $\Delta V(\eta):$
\begin{align}
\Delta V^\prime(\eta) = \frac{-\frac{W\gamma}{\eta^2}\left(\frac{2W}{\eta}-B\right)}{1+\frac{\gamma}{\eta}\left(\frac{W}{\eta}-B\right)} \label{eq:deriv_g}
\end{align}

\noindent for $\eta < \eta^*$ we have

\begin{align}
&W/\eta-B > W/\eta^*-B = 0 \nonumber\\ &\implies W/\eta-B > 0 \;\texttt{and}\; 2W/\eta-B > 0 \label{eq:deriv_g2}
\end{align}

\noindent Combining Equations \eqref{eq:deriv_g} and \eqref{eq:deriv_g2} and considering $W, \eta, \gamma > 0$ we have

\begin{align}
\Delta V^\prime(\eta) < 0 \nonumber
\end{align}

\noindent Hence $\Delta V(\eta)$ is a decreasing function of $\eta$ for $\eta \in [0,\eta^*].$

\end{proof}

\noindent{\bf Proof of Theorem \ref{thm:const_gamma}}
\begin{proof}
Evaluating the function $\Delta V(\eta)$ at $\eta = \eta^*$ and setting $W/\eta^* = B$ we get

\begin{align}
\Delta V(\eta^*) = -W\log\left[1+\frac{\gamma}{\eta^*}\bigg(\frac{W}{\eta^*}-B\bigg)\right]+B\gamma\left(\frac{W}{\eta^*}-B\right) = 0
\end{align}

\noindent From Lemma \ref{lemma:const_gamma} it is clear that $\Delta V(\eta) > \Delta V(\eta^*)$ for all $\eta<\eta^*.$ Hence $\Delta V(\eta) > 0$ for all $\eta < \eta^*.$

\end{proof}

\subsection{Poisson Approximation for various distributions}\label{appd6}
\begin{figure}%[htbp]
\centering
%\hspace{-0.5cm}
\begin{minipage}{0.45\textwidth}
\includegraphics[width=1\textwidth]{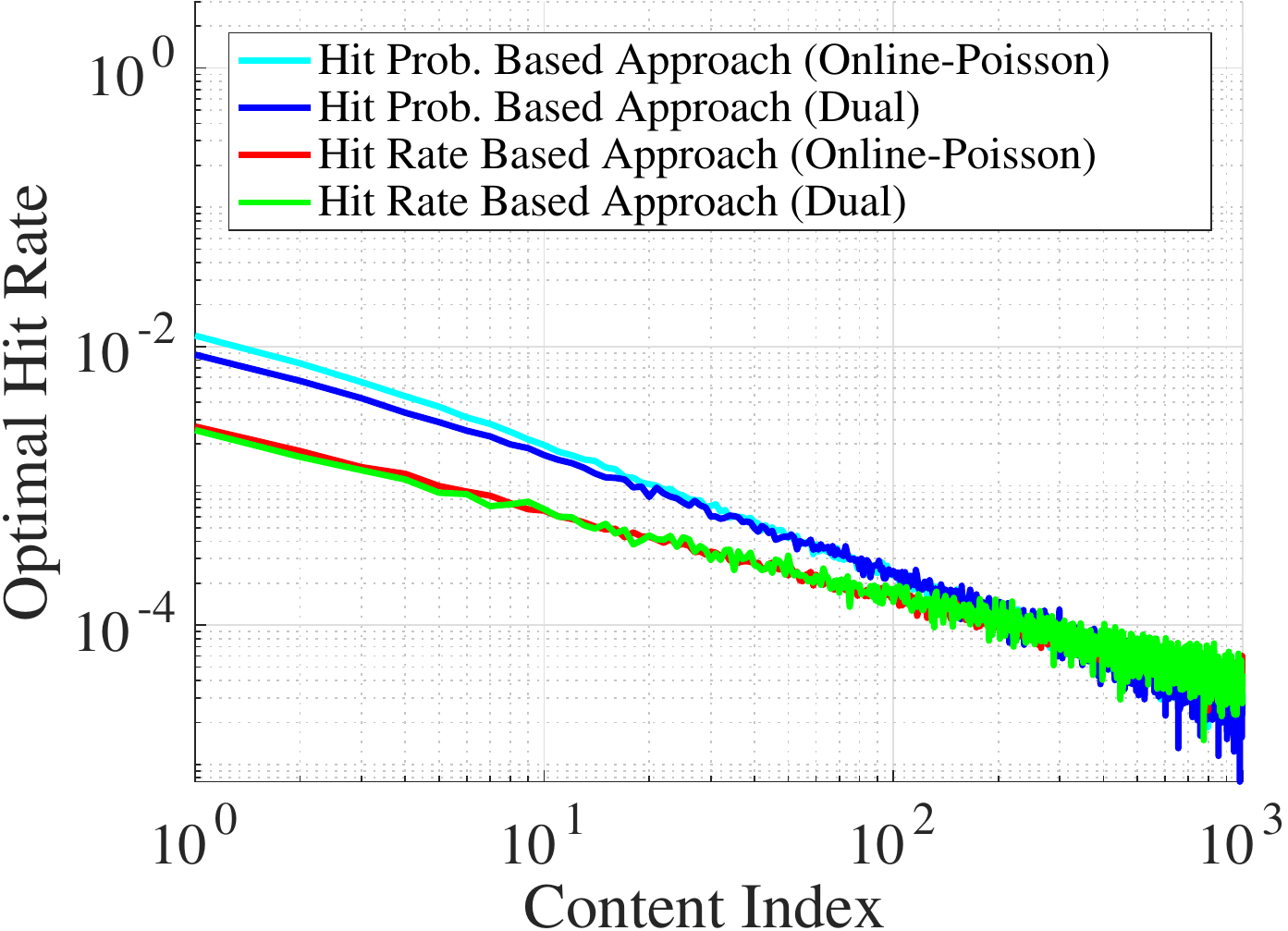}
\end{minipage}%\hfill
\begin{minipage}{0.45\textwidth}
\includegraphics[width=1\textwidth]{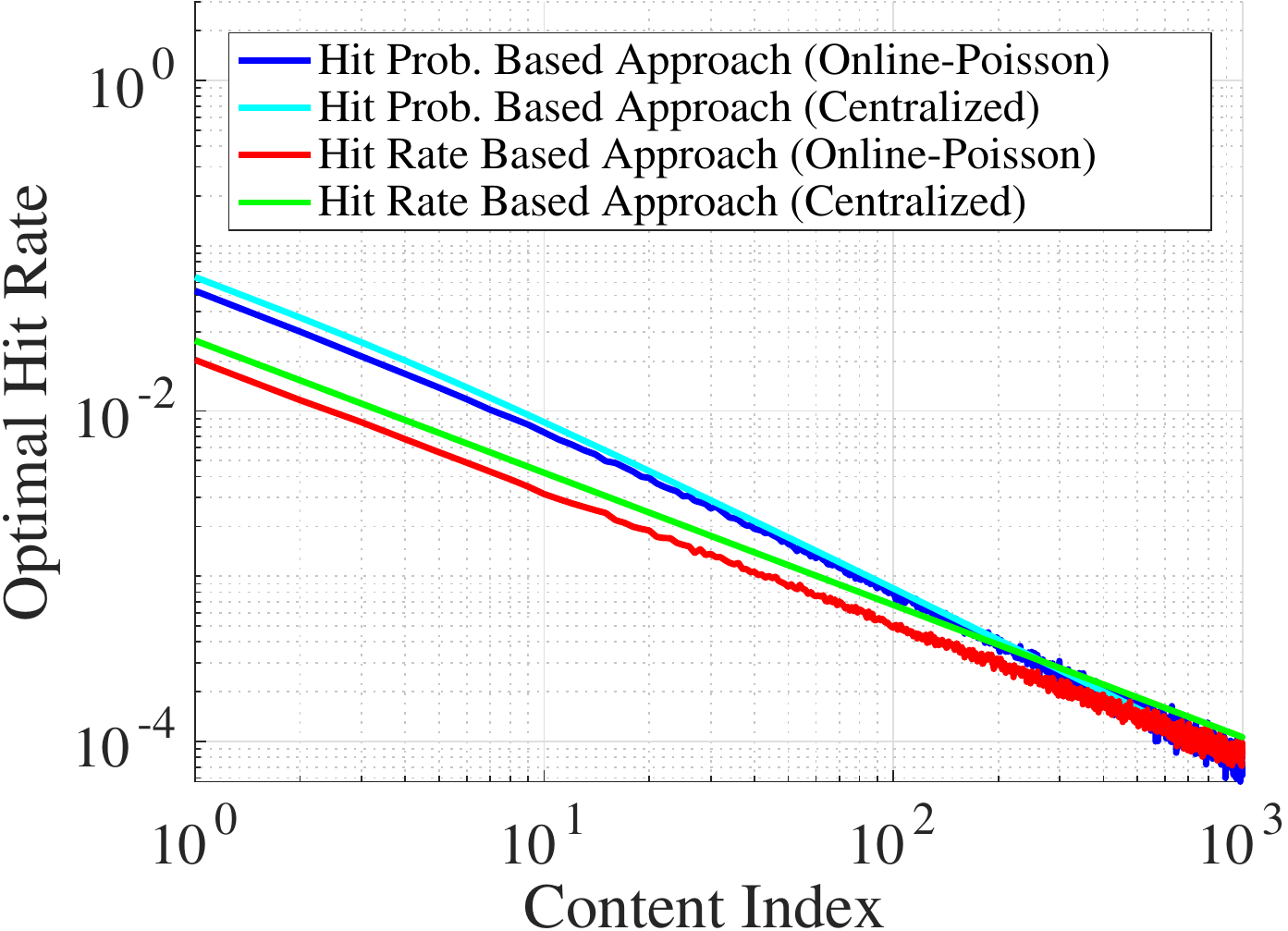}
\end{minipage}
\vspace{-0.05in}
\caption{Poisson online approximation to Hyperexponential \textit{(Left)} and weibull \textit{(Right)} inter-requests.}%and generalized Pareto distribution (c), (d).}
\label{poissonapprox-hywb}
\vspace{-0.15in}
\end{figure}

\subsubsection{Hyperexponential Distribution}  \label{sub:hypexpo}
W.l.o.g., we consider the arrivals follow a hyperexponential distribution with phase probabilities $p_{1i} = p_{2i} = 0.5,$ and phase rate parameters $\theta_{1}$ and $\theta_{2}$ Zipf distributed with rates $0.4$ and $0.8,$ respectively.   %From Section~\ref{sec:distr}, ~(\ref{eq:hrbcum2}) cannot be solved centrally here, hence we consider the Dual, and denote the solution as ``Dual".  %Again, we approximate this process with arrivals following a Poisson process with same mean, i.e., $\mu_i=1/(p_{1i}/\theta_{1i} + p_{2i}/\theta_{2i}).$ 
  From Figure~\ref{poissonapprox} (b), we can see that the optimal hit rates obtained through~(\ref{eq:mod-dual-general}) exactly match those obtained from Dual under hyperexponential distribution by solving fixed point equation in \eqref{eq:hypexpfpt}.  Similar performance was obtained for other parameters, especially for $p_{1i}\ll p_{2i}$, hence are omitted here. 
  
\subsubsection{Weibull Distribution}
Similarly, we consider a Weibull distribution with shape parameter $k_i = 0.5$.  %Again, from Section~\ref{sec:distr}, ~(\ref{eq:hrbcum2}) can be solved centrally here.  
  From Figure~\ref{poissonapprox} (c), it is clear that this approximation is accurate.  Note that when $k_i \rightarrow 1$, Weibull behaves more closed to exponential distribution, hence the accuracy of this approximation can be further improved.  For smaller value of $k_i,$ it has been shown that Weibull can be well approximated by hyperexponential distribution \cite{jin10}.  The performance of Poisson approximation to hyperexponential distribution has been discussed in Section \ref{sub:hypexpo}. 
  
\subsection{Limiting Behavior of 2-MMPP} \label{app-limit-2mmpp}
\noindent{\textit{\textbf{Case $1$: When $r_{1i}\rightarrow\infty$ and $r_{2i}\rightarrow\infty$, i.e., $x_i\rightarrow\infty$:}}}
From Equation \eqref{eq:mmppparams}, we get
{\footnotesize
%\begin{small}
\begin{align}
&u_{1i} = \frac{1}{2}\bigg[\theta_{1i}+\theta_{2i}+r_{1i}+r_{2i}-\sqrt{(\theta_{1i}-\theta_{2i}+r_{1i}-r_{2i})^2+4r_{1i}r_{2i}}\bigg]\nonumber\displaybreak[0]\\
&= \frac{1}{2}\bigg[\frac{(\theta_{1i}+\theta_{2i}+r_{1i}+r_{2i})^2-(\theta_{1i}-\theta_{2i}+r_{1i}-r_{2i})^2-4r_{1i}r_{2i}}{\theta_{1i}+\theta_{2i}+r_{1i}+r_{2i}+\sqrt{(\theta_{1i}-\theta_{2i}+r_{1i}-r_{2i})^2+4r_{1i}r_{2i}}}\bigg]\nonumber\displaybreak[1]\\
&= \frac{2\theta_{1i}\theta_{2i}+2\theta_{1i}a_{2i}x_i+2\theta_{2i}a_{1i}x_i}{\theta_{1i}+\theta_{2i}+(a_{1i}+a_{2i})x_i+\sqrt{(\theta_{1i}-\theta_{2i}+(a_{1i}-a_{2i})x_i)^2+4a_{1i}a_{2i}x_i^2}}.
\end{align} 
%\end{small}
}

When $x_i \rightarrow\infty$, by applying L'Hospital's rule, we have
{\footnotesize
\begin{align}
&u_{1i} = \frac{2\theta_{1i}a_{2i}+2\theta_{2i}a_{1i}}{a_{1i}+a_{2i}+\lim\limits_{x_i\rightarrow\infty}\frac{(\theta_{1i}-\theta_{2i}+(a_{1i}-a_{2i})x_i)(a_{1i}-a_{2i})+{4}a_{1i}a_{2i}x_i}{\sqrt{(\theta_{1i}-\theta_{2i}+(a_{1i}-a_{2i})x_i)^2+4a_{1i}a_{2i}x_i^2}}}\nonumber\displaybreak[0]\\
&={ \frac{2\theta_{1i}a_{2i}+2\theta_{2i}a_{1i}}{a_{1i}+a_{2i}+\lim\limits_{x_i\rightarrow\infty}\frac{(\theta_{1i}-\theta_{2i})(a_{1i}-a_{2i})+(a_{1i}+a_{2i})^2x_i}{\sqrt{(\theta_{1i}-\theta_{2i})^2+2(\theta_{1i}-\theta_{2i})(a_{1i}-a_{2i})x_i+(a_{1i}+a_{2i})^2x_i^2}}}}\nonumber\displaybreak[1]\\
&={ \frac{2\theta_{1i}a_{2i}+2\theta_{2i}a_{1i}}{a_{1i}+a_{2i}+\lim\limits_{x_i\rightarrow\infty}\frac{(a_{1i}+a_{2i})^2x_i}{(a_{1i}+a_{2i})x_i}}} ={ \frac{2\theta_{1i}a_{2i}+2\theta_{2i}a_{1i}}{a_{1i}+a_{2i}+a_{1i}+a_{2i}}} =\frac{\theta_{1i}a_{2i}+\theta_{2i}a_{1i}}{a_{1i}+a_{2i}}.
\end{align} 
}
Similarly, we obtain $u_{2i}=\infty.$
%Simplifying further we get, 
%{\footnotesize
%\begin{align}
%u_{1i} &= \frac{2\theta_{1i}a_2+2\theta_{2i}a_1}{a_1+a_2+a_1+a_2} = \frac{\theta_{1i}a_2+\theta_{2i}a_1}{a_1+a_2}
%\end{align} 
%}
%A similar procedure can be applied to obtain other parameters. Thus we have 
%{\footnotesize
%\begin{align}
%u_{2i}= \infty,\; q_{1i} = 1,\;q_{2i} = 0\nonumber
%\end{align} 
%}

{Again, from Equation \eqref{eq:mmppparams},
{\footnotesize
\begin{align}
q_{1i} &= \frac{\theta_{2i}^2r_{1i}+\theta_{1i}^2r_{2i}}{(\theta_{1i}r_{2i}+\theta_{2i}r_{1i})(u_{1i}-u_{2i})} - \frac{u_{2i}}{u_{1i}-u_{2i}} =\frac{\theta_{2i}^2a_{1i}+\theta_{1i}^2a_{2i}}{(\theta_{1i}a_{2i}+\theta_{2i}a_{1i})(-\delta)}+\frac{\theta_{1i}+\theta_{2i}+(a_{1i}+a_{2i})x_i}{2\delta}+\frac{1}{2}\nonumber\displaybreak[1]\\
         &=\frac{(\theta_{1i}a_{2i}+\theta_{2i}a_{1i})(a_{1i}+a_{2i})x+\theta_{1i}\theta_{2i}(a_{1i}+a_{2i})-\theta_{1i}^2a_{2i}-\theta_{2i}^2a_{1i}}{2(\theta_{1i}a_{2i}+\theta_{2i}a_{1i})\sqrt{(\theta_{1i}-\theta_{2i}+(a_{1i}-a_{2i})x_i)^2+4a_{1i}a_{2i}x_i^2}}+\frac{1}{2},
\end{align}
}
when $x_i\rightarrow\infty$, by applying L'Hospital's rule, we have
{\footnotesize
\begin{align}
q_{1i}  &=\frac{1}{2}+\frac{(\theta_{1i}a_{2i}+\theta_{2i}a_{1i})(a_{1i}+a_{2i})}{2(\theta_{1i}a_{2i}+\theta_{2i}a_{1i})\lim\limits_{x_i\rightarrow\infty}\frac{(\theta_{1i}-\theta_{2i}+(a_{1i}-a_{2i})x_i)(a_{1i}-a_{2i})+4a_{1i}a_{2i}x_i}{\sqrt{(\theta_{1i}-\theta_{2i}+(a_{1i}-a_{2i})x_i)^2+4a_{1i}a_{2i}x_i^2}}}\nonumber\displaybreak[0]\\
&=\frac{1}{2}+\frac{(\theta_{1i}a_{2i}+\theta_{2i}a_{1i})(a_{1i}+a_{2i})}{2(\theta_{1i}a_{2i}+\theta_{2i}a_{1i})(a_{1i}+a_{2i})} =\frac{1}{2}+\frac{1}{2}=1,
\end{align}
thus $q_{2i}=0.$
}

{\noindent{\textit{\textbf{Case $2$: When $r_{1i}\rightarrow0$ and $r_{2i}\rightarrow0$, i.e., $x_i\rightarrow0$:}}}  W.l.o.g., we assume $\theta_{1i}\geq\theta_{2i}.$  From Equation \eqref{eq:mmppparams}, 
\begin{align}
\delta=\theta_{1i}-\theta_{2i},
\end{align}
then
\begin{align}
u_{1i}=\frac{1}{2}[\theta_{1i}+\theta_{2i}+(a_{1i}+a_{2i})x_i-\delta]=\theta_{2i},
\end{align}
similarly, we have $u_{2i}=\theta_{1i}.$}

{Again, from Equation \eqref{eq:mmppparams}, when $x_i\rightarrow0,$ we obtain 
{\footnotesize
\begin{align}
q_{1i} &= \frac{1}{2}+\frac{(\theta_{1i}a_{2i}+\theta_{2i}a_{1i})(a_{1i}+a_{2i})x+\theta_{1i}\theta_{2i}(a_{1i}+a_{2i})-\theta_{1i}^2a_{2i}-\theta_{2i}^2a_{1i}}{2(\theta_{1i}a_{2i}+\theta_{2i}a_{1i})\delta}\nonumber\displaybreak[0]\\
&=\frac{1}{2}+\frac{\theta_{1i}\theta_{2i}(a_{1i}+a_{2i})-\theta_{1i}^2a_{2i}-\theta_{2i}^2a_{1i}}{2(\theta_{1i}a_{2i}+\theta_{2i}a_{1i})(\theta_{1i}-\theta_{2i})} = \frac{\theta_{2i}a_{1i}}{\theta_{1i}a_{2i}+\theta_{2i}a_{1i}},
\end{align}
}
then $q_{2i}=\frac{\theta_{1i}a_{2i}}{\theta_{1i}a_{2i}+\theta_{2i}a_{1i}}.$}

\subsection{Limiting Behavior of m- state MMPP}\label{app-limit-mmmpp}
%\subsection{Evaluation of Age Distribution}\label{appb}
In this section, we derive limiting behavior  for an $m$- state MMPP arrival process.

Consider a $m$-state MMPP whose arrival rate is given by $\boldsymbol\theta[J(t)]$, where $J(t), t\geq 0$ is a $m$-state continuous-time Markov chain with generator $Q.$  Denote $\boldsymbol r=(r_1,\cdots, r_m).$   Let $X_k$ be the time between the $(k-1)$-th and $k$-th arrivals with $X_0=0.$   When the Markov chain is in state $i$, arrivals occur according to a Poisson process with rate $\theta_i.$ Denote $\boldsymbol \theta=(\theta_1,\cdots,\theta_m).$

Consider randomly chosen non-overlap time intervals $A_1=[a_1, b_1],$ $A_2=[a_2, b_2],$ $A_3=[a_3, b_3],\cdots $  with $a_1<b_1<a_2<b_2<\cdots$. Let $A=(A_1, A_2,\cdots, A_l).$

\subsubsection{The transition rate approaches zero}\label{app-lmt-m-zero}
%\begin{figure}
%\centering
%\includegraphics[width=0.9\linewidth]{figures/arrivals.pdf}
%\caption{Arrivals}
%\label{fig:arrivals}
%\vspace{-0.2in}
%\end{figure}

We first consider the case that the transition rate approaches zero, i.e., $\boldsymbol r\rightarrow \boldsymbol 0$ ($r_i\rightarrow 0$ for $i=1,\cdots, m$). 

Consider a time interval $[a, b]$ and $J(a)=i$ for $i=1,\cdots, m.$  Denote $M(a,b)$ be the number of state transition during $(a, b].$ For example, when $J(a)=i,$ $M_i(a, b)=0$ means all arrivals during $(a, b]$ follows the Poisson process with rate $\theta_i$. 

We first show that when $\boldsymbol r\rightarrow \boldsymbol 0$, the number of state transition approaches $0,$ i.e., $M(a, b)=0$ a.s.
\begin{prop}\label{prop1}
Suppose $J(a)=i,$ then we have 
\begin{align}
\mathbb{P}[J(b)=i, M(a, b)=0|J(a)=i]=e^{-r_i (b-a)}\rightarrow 1,\quad\text{as $\boldsymbol r\rightarrow\boldsymbol 0$.}
\end{align}
\end{prop}
\begin{proof}
From the property of $m$-state MMPP, we directly have 
\begin{align*}
\mathbb{P}[J(b)&=i, M(a, b)=0|J(a)=i] =\mathbb{P}[M(a, b)=0|J(a)=i]=e^{-r_i (b-a)},
\end{align*}
then take the limit as $\boldsymbol r\rightarrow\boldsymbol 0,$ which completes the proof.
\end{proof}

Then we immediately have 
\begin{corollary}\label{cor1}
\begin{align}
\mathbb{P}[J(b)=i, M(a, b)>0|J(a)=i]\rightarrow 0, \quad\text{as $\boldsymbol r\rightarrow\boldsymbol 0$.}
\end{align}
\end{corollary}
\begin{proof}
\begin{align}
\mathbb{P}[J(b)=i, M(a, b)>0|J(a)=i] &\leq \mathbb{P}[M(a, b)>0|J(a)=i]\nonumber\displaybreak[0]\\
&=1-\mathbb{P}[M(a, b)=0|J(a)=i] =1-e^{-r_i (b-a)} \rightarrow 0, \quad{\text{as $\boldsymbol r\rightarrow \boldsymbol 0.$ }}
\end{align}       
Since $\mathbb{P}[J(b)=i, M(a, b)>0|J(a)=i]\geq 0,$ we have $\mathbb{P}[J(b)=i, M(a, b)>0|J(a)=i]\rightarrow 0$  as $\boldsymbol r\rightarrow \boldsymbol 0.$                                
\end{proof}

From Proposition~\ref{prop1} and Corollary~\ref{cor1}, we also obtain
\begin{prop}\label{prop2}
As $\boldsymbol r\rightarrow \boldsymbol 0,$ we have
\begin{align}
\mathbb{P}^{\boldsymbol r}[J(b)=j|J(a)=i]\rightarrow \delta_{ij}, \quad \forall b>a,
\end{align}
where 
\begin{align}
\delta_{ij}=\begin{cases} 
                  1,\quad\text{if $i=j;$}\nonumber\\
                  0,\quad\text{otherwise.}
                  \end{cases}
\end{align}
\end{prop}
\begin{proof}
Suppose $j\neq i,$ then 
\begin{align}
\mathbb{P}^{\boldsymbol r}[J(b)=j|J(a)=i] &= 1-\mathbb{P}^{\boldsymbol r}[J(b)=i|J(a)=i]\nonumber\displaybreak[1]\\
&=1-\mathbb{P}^{\boldsymbol r}[J(b)=i, M(a,b)=0|J(a)=i] -\mathbb{P}^{\boldsymbol r}[J(b)=i, M(a,b)>0|J(a)=i]\nonumber\displaybreak[2]\\
 &\stackrel{(a)}{\rightarrow} 0, \quad\text{as $\boldsymbol r\rightarrow\boldsymbol 0$,}
 \end{align}                                         
where (a) holds from Proposition~\ref{prop1} and Corollary~\ref{cor1}.

Hence, we have $\mathbb{P}^{\boldsymbol r}[J(b)=j|J(a)=i]=1$ for $j= i,$ which completes the proof. 
\end{proof}

Then we immediately have
\begin{corollary}\label{cor2}
As $\boldsymbol r\rightarrow \boldsymbol 0,$ we have 
\begin{align}
\mathbb{P}^{r}[J(a)=i]\rightarrow\mathbb{P}^0[J(a)=i]=1_{\{J(0)=i\}}.
\end{align}
\end{corollary}
%\begin{proof}
%\begin{align}
%\mathbb{P}^{r}[J(a)=i]=&\mathbb{P}^{r}[J(a)=i, J(0)=i, M(0, a)=0]+ \mathbb{P}^{r}[J(a)=i, J(0)=i, M(0, a)>0]\nonumber\\
%                                            & +\sum_{j\neq i}\mathbb{P}^{r}[J(a)=i, J(0)=j, M(0, a)>0]\nonumber\\
%                                            =&\mathbb{P}^{r}[J(a)=i, M(0, a)=0|J(0)=i]\mathbb{P}^{r}[J(0)=i] \nonumber\\
%                                            & + \mathbb{P}^{r}[J(a)=i, M(0, a)>0|J(0)=i]\mathbb{P}^{r}[J(0)=i]\nonumber\displaybreak[0]\\
%                                            & +\sum_{j\neq i}\mathbb{P}^{r}[J(a)=i,  M(0, a)>0|J(0)=j]\mathbb{P}^{r}[J(0)=j]\nonumber\displaybreak[1]\\
%                                           % \stackrel{(a)}{=}& e^{-r_i a}\mathbb{P}^{r}[J(0)=i]\nonumber\displaybreak[2]\\
%                                            %&+(1-e^{-r_i a})\mathbb{P}^{r}[J(0)=i]+\sum_{j\neq i}\mathbb{P}^{r}[J(a)=i,  M(0, a)>0|J(0)=j]\mathbb{P}^{r}[J(0)=j]\nonumber\displaybreak[3]\\
%                                             \stackrel{(a)}{\rightarrow}& 1\cdot\mathbb{P}^{r}[J(0)=i]+0\cdot\mathbb{P}^{r}[J(0)=i]+\sum_{j\neq i}0\cdot\mathbb{P}^{r}[J(0)=j], \quad \text{as $\boldsymbol r\rightarrow\boldsymbol 0$,}\nonumber\displaybreak[2]\\
%                                            =& \mathbb{P}^{r}[J(0)=i]=1_{\{J(0)=i\}}, 
%\end{align}
%where (a) holds from Proposition~\ref{prop1} and Corollary~\ref{cor1}.
%\end{proof}

Now we are ready to prove our main result.  Our goal is to show that the interarrival times within one state is exponentially distributed as the transition rate approaches zero.

\noindent {\bf Proof for Theorem \ref{thm:mstate-zero}:}
\begin{proof}
Basically we need to show that 
\begin{align}
&\mathbb{P}^{\boldsymbol r}[A_1,\cdots, A_l]\triangleq\mathbb{P}^{\boldsymbol r}[A]\rightarrow \mathbb{P}^{\boldsymbol 0}[A]=\sum_i\mathbb{P}[J(0)=i]\prod_{j=1}^l \mathbb{P}^{0}[A_j|J(0)=i],\quad\text{as $\boldsymbol r\rightarrow\boldsymbol 0$.}
\end{align}

First we consider the event $A_1=[a_1, b_1]$. Denote $N(a_1,b_1)$ as the number of arrivals during interval $[a_1, b_1].$ Then,
\begin{align}\label{eq:condition-prob-one}
&\mathbb{P}^{\boldsymbol r}[N(a_1,b_1)=n_1] =\sum_i^m \mathbb{P}^{\boldsymbol r}[N(a_1,b_1)=n_1|J(a_1)=i]\cdot\mathbb{P}^{\boldsymbol r}[J(a_1)=i].
\end{align}

%Let $M_i(a_1,b_1)$ be the number of state transitions during interval $[a_1, b_1]$ when $J(a_1)=i,$ i.e., $M_i(a_1, b_1)=0$ means all arrivals during $[a_1, b_1]$ follows the Poisson process with rate $\theta_i$.  Here we have 
In the following, we characterize the two terms in the sum of~(\ref{eq:condition-prob-one}), respectively.   

On one hand, we have
\begin{align}\label{eq:part1}
\mathbb{P}^{\boldsymbol r}[N(a_1,b_1)=n_1|J(a_1)=i] &=\mathbb{P}^{\boldsymbol r}[N(a_1,b_1)=n_1, M_i(a_1, b_1)=0|J(a_1)=i]\nonumber\\&+\mathbb{P}^{\boldsymbol r}[N(a_1,b_1)=n_1, M_i(a_1, b_1)>0|J(a_1)=i]\nonumber\displaybreak[1]\\
&\stackrel{(a)}{=}\mathbb{P}[N_i(a_1,b_1)=n_1]\cdot e^{-r_i(b_1-a_1)}+o(1)\nonumber\displaybreak[2]\\
&=\mathbb{P}[N_i(a_1,b_1)=n_1], \quad\text{as $\boldsymbol r\rightarrow\boldsymbol 0$},
\end{align}
where (a) holds from Proposition~\ref{prop1} and Corollary~\ref{cor1}.
% since as $\boldsymbol r\rightarrow \boldsymbol 0,$
%\begin{align}\label{eq:jump-prop}
%\mathbb{P}^{\boldsymbol r}[N(a_1,b_1)=n, M_i(a_1, b_1)=0|J(a_1)=i]&=\mathbb{P}[N_i(a_1,b_1)=n]\cdot e^{-r_i(b_1-a_1)}\nonumber\displaybreak[0]\\
%&\rightarrow\mathbb{P}[N_i(a_1,b_1)=n]\cdot 1,\nonumber\displaybreak[1]\\
%\mathbb{P}^{\boldsymbol r}[N(a_1,b_1)=n, M_i(a_1, b_1)>0|J(a_1)=i]&=\mathbb{P}^{\boldsymbol r}[N(a_1,b_1)=n|J(a_1)=i]\cdot(1-e^{-r_i(b_1-a_1)})\nonumber\displaybreak[2]\\
%&\rightarrow\mathbb{P}^{\boldsymbol r}[N(a_1,b_1)=n|J(a_1)=i]\cdot 0 =0.
%\end{align}

On the other hand, from Proposition~\ref{prop2}, we have 
\begin{align}\label{eq:part2}
\mathbb{P}^{\boldsymbol r}[J(a_1)=i]\rightarrow 1_{\{J(0)=i\}},\quad\text{as $\boldsymbol r\rightarrow\boldsymbol 0.$}
\end{align}

Combing~(\ref{eq:part1}) and~(\ref{eq:part2}) into~(\ref{eq:condition-prob-one}), we have 
\begin{align}
\mathbb{P}^{\boldsymbol r}[N(a_1,b_1)=n_1|J(0)=i]\rightarrow\sum_{i}^m \mathbb{P}[N_i(a_1,b_1)=n_1], \quad\text{as $\boldsymbol r\rightarrow\boldsymbol 0$.}
\end{align}

Given the arrivals of interval $[a_1, b_1]$, we consider the second interval $[a_2, b_2].$ From Proposition~\ref{prop2}, as $\boldsymbol r\rightarrow\boldsymbol 0,$  we have 
\begin{align}
\mathbb{P}(J(a_2)=i|J(a_1)=i)=1.
\end{align}  
Therefore, following a similar argument, we have 
\begin{align}
\mathbb{P}^{\boldsymbol r}[N(a_2,b_2)&=n_2|N(a_1,b_1)=n_1,J(0)=i] =\sum_{i}^m \mathbb{P}[N_i(a_2,b_2)=n_2]\cdot 1_{\{J(0)=i\}},
\end{align}
and 
\begin{align}
&\mathbb{P}^{\boldsymbol r}[N(a_2,b_2)=n_2, N(a_1,b_1)=n_1,J(0)=i]\nonumber\displaybreak[0]\\
=&\mathbb{P}^{\boldsymbol r}[N(a_2,b_2)=n_2|N(a_1,b_1)=n_1,J(0)=i]\cdot\mathbb{P}^{\boldsymbol r}[N(a_1,b_1)=n_1|J(0)=i]\nonumber\displaybreak[1]\\
=&\sum_{i}^m \mathbb{P}[N_i(a_2,b_2)=n_2]\cdot\sum_{i}^m \mathbb{P}[N_i(a_1,b_1)=n_1]
\end{align}
By induction, we have
\begin{align}
&\mathbb{P}^{\boldsymbol r}[N(a_l,b_l)=n_l,\cdots N(a_1,b_1)=n_1|J(0)=i]\nonumber\displaybreak[0]\\
=&\mathbb{P}^{\boldsymbol r}[N(a_l,b_l)=n_l|N(a_{l-1},b_{l-1})=n_{l-1}, \cdots N(a_1,b_1)=n_1, J(0)=i]\nonumber\\&\cdots \cdot \mathbb{P}^{\boldsymbol r}[N(a_2,b_2)=n_2|N(a_1,b_1)=n_1,J(0)=i]\cdot\mathbb{P}^{\boldsymbol r}[N(a_1,b_1)=n_1|J(0)=i]\nonumber\displaybreak[2]\\
=&\prod_{j=1}^l \sum_{i}^m \mathbb{P}[N_i(a_j,b_j)=n_j],
\end{align}
therefore,
\begin{align}
\mathbb{P}^{\boldsymbol 0}[A_1,\cdots, A_l]=\sum_{i}\mathbb{P}[J(0)=i]\prod_{j=1}^l \sum_{i}^m \mathbb{P}[N_i(a_j,b_j)=n_j]
\end{align}
which completes the proof.
\end{proof}

%In the following, we present a short proof through sample path.
%
%Consider a stochastic process $N^{\boldsymbol r}(\omega, t)=N_{J(\omega, t)}(\omega, t)$. 
%
%Given the initial state $J(\omega, 0)$ of the $m$-state MMPP, let $T_1(\omega)$ be the length of the time interval before the first state transition.  For any $0<t< T_1(\omega),$ we have 
%\begin{align}
%N^{\boldsymbol r}(\omega, t)=N_{J(\omega, 0)}(\omega, t),
%\end{align}
%where $J(\omega, 0)$ is the initial state of the stochastic process. 
%
%Now consider $\boldsymbol r\rightarrow\boldsymbol 0,$ then 
%\begin{align}
%T_1(\omega)\rightarrow\infty, \quad\text{a.s.}
%\end{align}
%
%Let $\Omega^{\infty}=\{\omega: T_1(\omega)\rightarrow\infty\}.$  Then for $\forall \omega\in\Omega^{\infty}$, $t>0,$ and $\boldsymbol r$ small enough we have 
%\begin{align}
%T_1(\omega)>t,
%\end{align}
%thus
%\begin{align}
%N^{\boldsymbol r}(\omega, t)=N_{J(\omega, 0)}(\omega, t).
%\end{align}
%
%Therefore,
%\begin{align}
%\lim_{\boldsymbol r\rightarrow 0}N^{\boldsymbol r}(\omega, t)=N_{J(\omega, 0)}(\omega, t), \quad \forall \omega\in\Omega^{\infty}, t>0,
%\end{align}
%where the convergence is uniform for $t$ in any bounded set.

\subsubsection{The transition rate approaches infinity}\label{app-lmt-m-inf}
Now we consider the case that the transition rate approaches infinity, i.e., $\boldsymbol r\rightarrow \boldsymbol \infty$ ($r_i\rightarrow \infty$ for $i=1,\cdots, m$). 

\noindent {\bf Proof for Theorem \ref{thm:mstate-infinite}:}
\begin{proof}
W.l.o.g., consider a time interval $[a, b].$ Suppose that there are $k$ state transition occurs during $[a, b].$ As $\boldsymbol r\rightarrow \boldsymbol \infty$, we have $k\rightarrow\infty$ a.s.  

For simplicity, we denote the length of each time interval within one state as $L_j$ for $j=1,\cdots, k$ and $\sum_{j}L_j=b-a.$ Then we have 
\begin{align}
N(a, b)=\sum_j N_j(L_j),
\end{align}
where $N(a,b)$ is the number of arrivals in $(a, b]$ and $N_j(L_j)$ is the number of arrivals during a time interval of length $L_j$ in $(a, b].$

W.l.o.g.,  suppose during the $j$-th interval, the MC is in state $i$ which has Poisson arrivals with rate $\theta_i,$ for $i=1,\cdots, m.$ Let $L_j(i)$ be the length of the corresponding interval.

Consider any two time $t_i$ and $t_i^\prime$ with $t_i^\prime>t_i$ during which the MC is in state $i$ and $M_i(t_i, t_i^\prime)=0$. As $\boldsymbol r\rightarrow\infty,$ we have $|t_i^\prime-t_i|\rightarrow0$ a.s.. 

Given the $m$-state MMPP with $\boldsymbol r\rightarrow\infty$  we have  
\begin{align}\label{eq:trans1}
\mathbb{P}^{\boldsymbol r}(t_i^\prime)= \mathbb{P}^{\boldsymbol r}(t_i) e^{Q(t_i^\prime-t_i)} =  \mathbb{P}^{\boldsymbol r}(t_i) e^{r^\prime Q_0(t_i^\prime-t_i)},
\end{align}
where $r^\prime\rightarrow \infty$ is a scalar and $Q=r^\prime Q_0$.

The stochastic process described by~(\ref{eq:trans1}) (call it {\bf Process $1$}) is equivalent to a stochastic process with a finite-valued $Q_0$ (call it {\bf Process $2$}) at stationary state (i.e., $t\rightarrow\infty$), since
\begin{align}\label{eq:trans2}
 \mathbb{P}^{\boldsymbol r}(t_i^\prime)\stackrel{(a)}{=} \mathbb{P}^{\boldsymbol r}(t_i) e^{r^\prime Q_0\cdot (t_i^\prime-t_i)} \stackrel{(b)}{=}\mathbb{P}^{\boldsymbol r}(t_i) e^{ Q_0 \cdot r^\prime(t_i^\prime-t_i)}, 
\end{align}
where (a) is from~(\ref{eq:trans1}) and (b) is through reordered the elements. We have $r^\prime(t_i^\prime-t_i)\rightarrow\infty$.

In other words, 
\begin{itemize}
\item {\bf Process $1$}: we consider a finite time interval $(a, b]$ with an infinite state transition rate $\boldsymbol r\rightarrow\infty$;
\item {\bf Process $2$}: we consider an infinite time interval $(ra, rb]$ with a finite state transition rate.
\end{itemize} 
From~(\ref{eq:trans2}), it is clear that {\bf Process $1$} and {\bf Process $2$} are equivalent. 

Denote $\boldsymbol \pi =(\pi_1,\cdots, \pi_m)$ as the stationary distribution for the $m$-state MMPP. Then we have 
\begin{align}\label{eq:time-length}
L_j(i)=\pi_i(b-a).
\end{align}

Since the stochastic process is in stationary at each interval for a particular state, and the arrivals in each state follows a Poisson process. Therefore, the interarrival times of the $m$-state MMPP are exponentially distributed since the mix of Poisson process is still Poisson.  From~(\ref{eq:time-length}), the equivalent arrival rate satisfies 
\begin{align}
\bar{\theta}=\sum_i^m \theta_i \pi_i.
\end{align}

\end{proof}

\bibliographystyle{ACM-Reference-Format}
\bibliography{refs} 


\begin{thebibliography}{36}

%%% the Web address.

\ifx \showCODEN    \undefined \def \showCODEN     #1{\unskip}     \fi
\ifx \showDOI      \undefined \def \showDOI       #1{#1}\fi
\ifx \showISBNx    \undefined \def \showISBNx     #1{\unskip}     \fi
\ifx \showISBNxiii \undefined \def \showISBNxiii  #1{\unskip}     \fi
\ifx \showISSN     \undefined \def \showISSN      #1{\unskip}     \fi
\ifx \showLCCN     \undefined \def \showLCCN      #1{\unskip}     \fi
\ifx \shownote     \undefined \def \shownote      #1{#1}          \fi
\ifx \showarticletitle \undefined \def \showarticletitle #1{#1}   \fi
\ifx \showURL      \undefined \def \showURL       {\relax}        \fi
% invisible to TeX
\providecommand\bibfield[2]{#2}
\providecommand\bibinfo[2]{#2}
\providecommand\natexlab[1]{#1}
\providecommand\showeprint[2][]{arXiv:#2}

\bibitem[\protect\citeauthoryear{Aven, Coffman, and Kogan}{Aven
  et~al\mbox{.}}{1987}]%
        {aven87}
\bibfield{author}{\bibinfo{person}{O.~I. Aven}, \bibinfo{person}{E.~G.
  Coffman}, {and} \bibinfo{person}{Y.~A. Kogan}.}
  \bibinfo{year}{1987}\natexlab{}.
\newblock \bibinfo{booktitle}{\emph{Stochastic {A}nalysis of {C}omputer
  {S}torage}}.
\newblock \bibinfo{publisher}{Springer Science \& Business Media}.
\newblock


\bibitem[\protect\citeauthoryear{Baccelli and Br{\'e}maud}{Baccelli and
  Br{\'e}maud}{2013}]%
        {baccelli13}
\bibfield{author}{\bibinfo{person}{F. Baccelli} {and} \bibinfo{person}{P.
  Br{\'e}maud}.} \bibinfo{year}{2013}\natexlab{}.
\newblock \bibinfo{booktitle}{\emph{Elements of {Q}ueueing {T}heory: {P}alm
  {M}artingale {C}alculus and {S}tochastic {R}ecurrences}}.
  Vol.~\bibinfo{volume}{26}.
\newblock \bibinfo{publisher}{Springer Science \& Business Media}.
\newblock


\bibitem[\protect\citeauthoryear{Berger, Gland, Singla, and Ciucu}{Berger
  et~al\mbox{.}}{2014}]%
        {berger14}
\bibfield{author}{\bibinfo{person}{D. Berger}, \bibinfo{person}{P. Gland},
  \bibinfo{person}{S. Singla}, {and} \bibinfo{person}{F. Ciucu}.}
  \bibinfo{year}{2014}\natexlab{}.
\newblock \showarticletitle{Exact {A}nalysis of {TTL} {C}ache {N}etworks}.
\newblock \bibinfo{journal}{\emph{Performance Evaluation}}
  \bibinfo{volume}{79} (\bibinfo{year}{2014}), \bibinfo{pages}{2--23}.
\newblock


\bibitem[\protect\citeauthoryear{Boyd and Vandenberghe}{Boyd and
  Vandenberghe}{2004}]%
        {boyd04}
\bibfield{author}{\bibinfo{person}{S Boyd} {and} \bibinfo{person}{L
  Vandenberghe}.} \bibinfo{year}{2004}\natexlab{}.
\newblock \bibinfo{booktitle}{\emph{Convex {O}ptimization}}.
\newblock \bibinfo{publisher}{Cambridge University Press}.
\newblock


\bibitem[\protect\citeauthoryear{Cha, Kwak, Rodriguez, Ahn, and Moon}{Cha
  et~al\mbox{.}}{2007}]%
        {cha07tube}
\bibfield{author}{\bibinfo{person}{Meeyoung Cha}, \bibinfo{person}{Haewoon
  Kwak}, \bibinfo{person}{Pablo Rodriguez}, \bibinfo{person}{Yong-Yeol Ahn},
  {and} \bibinfo{person}{Sue Moon}.} \bibinfo{year}{2007}\natexlab{}.
\newblock \showarticletitle{I {T}ube, {Y}ou {T}ube, {E}verybody {T}ubes:
  {A}nalyzing the {W}orld's {L}argest {U}ser {G}enerated {C}ontent {V}ideo
  {S}ystem}. In \bibinfo{booktitle}{\emph{ACM IMC}}.
\newblock


\bibitem[\protect\citeauthoryear{Cha, Kwak, Rodriguez, Ahn, and Moon}{Cha
  et~al\mbox{.}}{2009}]%
        {cha09}
\bibfield{author}{\bibinfo{person}{Meeyoung Cha}, \bibinfo{person}{Haewoon
  Kwak}, \bibinfo{person}{Pablo Rodriguez}, \bibinfo{person}{Yong-Yeol Ahn},
  {and} \bibinfo{person}{Sue Moon}.} \bibinfo{year}{2009}\natexlab{}.
\newblock \showarticletitle{Analyzing the {V}ideo {P}opularity
  {C}haracteristics of {L}arge-{S}cale {U}ser {G}enerated {C}ontent {S}ystems}.
\newblock \bibinfo{journal}{\emph{IEEE/ACM Transactions on Networking}}
  \bibinfo{volume}{17}, \bibinfo{number}{5} (\bibinfo{year}{2009}),
  \bibinfo{pages}{1357--1370}.
\newblock


\bibitem[\protect\citeauthoryear{Che, Tung, and Wang}{Che
  et~al\mbox{.}}{2002}]%
        {che02}
\bibfield{author}{\bibinfo{person}{H. Che}, \bibinfo{person}{Y. Tung}, {and}
  \bibinfo{person}{Z. Wang}.} \bibinfo{year}{2002}\natexlab{}.
\newblock \showarticletitle{Hierarchical {W}eb {C}aching {S}ystems: {M}odeling,
  {D}esign and {E}xperimental {R}esults}.
\newblock \bibinfo{journal}{\emph{IEEE Journal on Selected Areas in
  Communications}} \bibinfo{volume}{20}, \bibinfo{number}{7}
  (\bibinfo{year}{2002}), \bibinfo{pages}{1305--1314}.
\newblock


\bibitem[\protect\citeauthoryear{Dehghan, Massoulie, Towsley, Menasche, and
  Tay}{Dehghan et~al\mbox{.}}{2016}]%
        {dehghan16}
\bibfield{author}{\bibinfo{person}{M. Dehghan}, \bibinfo{person}{L. Massoulie},
  \bibinfo{person}{D. Towsley}, \bibinfo{person}{D. Menasche}, {and}
  \bibinfo{person}{YC Tay}.} \bibinfo{year}{2016}\natexlab{}.
\newblock \showarticletitle{A {U}tility {O}ptimization {A}pproach to {N}etwork
  {C}ache {D}esign}. In \bibinfo{booktitle}{\emph{IEEE INFOCOM}}.
\newblock


\bibitem[\protect\citeauthoryear{Fagin}{Fagin}{1977}]%
        {fagin77}
\bibfield{author}{\bibinfo{person}{Ronald Fagin}.}
  \bibinfo{year}{1977}\natexlab{}.
\newblock \showarticletitle{Asymptotic {M}iss {R}atios over {I}ndependent
  {R}eferences}.
\newblock \bibinfo{journal}{\emph{J. Comput. System Sci.}}
  \bibinfo{volume}{14}, \bibinfo{number}{2} (\bibinfo{year}{1977}),
  \bibinfo{pages}{222--250}.
\newblock


\bibitem[\protect\citeauthoryear{Feldmann and Whitt}{Feldmann and
  Whitt}{1997}]%
        {feldmann97}
\bibfield{author}{\bibinfo{person}{A. Feldmann} {and} \bibinfo{person}{W.
  Whitt}.} \bibinfo{year}{1997}\natexlab{}.
\newblock \showarticletitle{{F}itting {M}ixtures of {E}xponentials to
  {L}ong-Tail {D}istributions to {A}nalyze {N}etwork {P}erformance {M}odels}.
  In \bibinfo{booktitle}{\emph{IEEE INFOCOM}}.
\newblock


\bibitem[\protect\citeauthoryear{Ferragut, Rodr{\'\i}guez, and
  Paganini}{Ferragut et~al\mbox{.}}{2016}]%
        {ferragut16}
\bibfield{author}{\bibinfo{person}{Andr{\'e}s Ferragut},
  \bibinfo{person}{Ismael Rodr{\'\i}guez}, {and} \bibinfo{person}{Fernando
  Paganini}.} \bibinfo{year}{2016}\natexlab{}.
\newblock \showarticletitle{Optimizing {TTL} {C}aches under {H}eavy-tailed
  {D}emands}. In \bibinfo{booktitle}{\emph{ACM SIGMETRICS}}.
\newblock


\bibitem[\protect\citeauthoryear{Fofack, Dehghan, Towsley, Badov, and
  Goeckel}{Fofack et~al\mbox{.}}{2014a}]%
        {fofack14performance}
\bibfield{author}{\bibinfo{person}{N.~C. Fofack}, \bibinfo{person}{M. Dehghan},
  \bibinfo{person}{D. Towsley}, \bibinfo{person}{M. Badov}, {and}
  \bibinfo{person}{D.~L. Goeckel}.} \bibinfo{year}{2014}\natexlab{a}.
\newblock \showarticletitle{On the {P}erformance of {G}eneral {C}ache
  {N}etworks}. In \bibinfo{booktitle}{\emph{VALUETOOLS}}.
\newblock


\bibitem[\protect\citeauthoryear{Fofack, Nain, Neglia, and Towsley}{Fofack
  et~al\mbox{.}}{2012}]%
        {fofack12}
\bibfield{author}{\bibinfo{person}{N.~C. Fofack}, \bibinfo{person}{P. Nain},
  \bibinfo{person}{G. Neglia}, {and} \bibinfo{person}{D. Towsley}.}
  \bibinfo{year}{2012}\natexlab{}.
\newblock \showarticletitle{Analysis of {TTL}-based {C}ache {N}etworks}. In
  \bibinfo{booktitle}{\emph{VALUETOOLS}}.
\newblock


\bibitem[\protect\citeauthoryear{Fofack, Nain, Neglia, and Towsley}{Fofack
  et~al\mbox{.}}{2014b}]%
        {fofack14}
\bibfield{author}{\bibinfo{person}{N.~C. Fofack}, \bibinfo{person}{P. Nain},
  \bibinfo{person}{G. Neglia}, {and} \bibinfo{person}{D. Towsley}.}
  \bibinfo{year}{2014}\natexlab{b}.
\newblock \showarticletitle{Performance {E}valuation of {H}ierarchical
  {TTL}-based {C}ache {N}etworks}.
\newblock \bibinfo{journal}{\emph{Computer Networks}} (\bibinfo{year}{2014}).
\newblock


\bibitem[\protect\citeauthoryear{Fricker, Robert, Roberts, and Sbihi}{Fricker
  et~al\mbox{.}}{2012}]%
        {fricker12}
\bibfield{author}{\bibinfo{person}{C. Fricker}, \bibinfo{person}{P. Robert},
  \bibinfo{person}{J. Roberts}, {and} \bibinfo{person}{N. Sbihi}.}
  \bibinfo{year}{2012}\natexlab{}.
\newblock \showarticletitle{Impact of {T}raffic {M}ix on {C}aching
  {P}erformance in a {C}ontent-{C}entric {N}etwork}. In
  \bibinfo{booktitle}{\emph{INFOCOM WKSHPS}}.
\newblock


\bibitem[\protect\citeauthoryear{Garetto, Leonardi, and Martina}{Garetto
  et~al\mbox{.}}{2016}]%
        {garetto16}
\bibfield{author}{\bibinfo{person}{M. Garetto}, \bibinfo{person}{E. Leonardi},
  {and} \bibinfo{person}{v. Martina}.} \bibinfo{year}{2016}\natexlab{}.
\newblock \showarticletitle{A {U}nified {A}pproach to the {P}erformance
  {A}nalysis of {C}aching {S}ystems}.
\newblock \bibinfo{journal}{\emph{ACM Transactions on Modeling and Performance
  Evaluation of Computing Systems}} \bibinfo{volume}{1}, \bibinfo{number}{3}
  (\bibinfo{year}{2016}), \bibinfo{pages}{12}.
\newblock


\bibitem[\protect\citeauthoryear{Gast and Van~Houdt}{Gast and
  Van~Houdt}{2016}]%
        {gast16}
\bibfield{author}{\bibinfo{person}{N. Gast} {and} \bibinfo{person}{B.
  Van~Houdt}.} \bibinfo{year}{2016}\natexlab{}.
\newblock \showarticletitle{Asymptotically {E}xact {TTL}-{A}pproximations of
  the {C}ache {R}eplacement {A}lgorithms {LRU}(m) and h-{LRU}}. In
  \bibinfo{booktitle}{\emph{ITC 28}}.
\newblock


\bibitem[\protect\citeauthoryear{Hahn}{Hahn}{1958}]%
        {hahn58}
\bibfield{author}{\bibinfo{person}{W. Hahn}.} \bibinfo{year}{1958}\natexlab{}.
\newblock \showarticletitle{Uber die {A}nwendung der {M}ethode von {L}japunov
  auf {D}ifferenzen-gleichungen}.
\newblock \bibinfo{journal}{\emph{Math. Ann.}} \bibinfo{number}{136}
  (\bibinfo{year}{1958}), \bibinfo{pages}{430--441}.
\newblock


\bibitem[\protect\citeauthoryear{Jiang, Nain, and Towsley}{Jiang
  et~al\mbox{.}}{2017}]%
        {bo17}
\bibfield{author}{\bibinfo{person}{B. Jiang}, \bibinfo{person}{P. Nain}, {and}
  \bibinfo{person}{D. Towsley}.} \bibinfo{year}{2017}\natexlab{}.
\newblock \showarticletitle{On the {C}onvergence of the {TTL} {A}pproximation
  for an {LRU} {C}ache under {I}ndependent {S}tationary {R}request
  {P}rocesses}.
\newblock \bibinfo{journal}{\emph{Arxiv preprint arXiv:1707.06204}}
  (\bibinfo{year}{2017}).
\newblock


\bibitem[\protect\citeauthoryear{Jin and Gonigunta}{Jin and Gonigunta}{2010}]%
        {jin10}
\bibfield{author}{\bibinfo{person}{T Jin} {and} \bibinfo{person}{L.S.
  Gonigunta}.} \bibinfo{year}{2010}\natexlab{}.
\newblock \showarticletitle{Exponential {A}pproximation to {W}eibull {R}enewal
  with {D}ecreasing {F}ailure {R}ate}.
\newblock \bibinfo{journal}{\emph{J. Stat. Comput. Simul.}}
  \bibinfo{volume}{80}, \bibinfo{number}{3} (\bibinfo{year}{2010}),
  \bibinfo{pages}{273--285}.
\newblock


\bibitem[\protect\citeauthoryear{Jung, Berger, and Balakrishnan}{Jung
  et~al\mbox{.}}{2003}]%
        {jung03}
\bibfield{author}{\bibinfo{person}{J. Jung}, \bibinfo{person}{A. Berger}, {and}
  \bibinfo{person}{H. Balakrishnan}.} \bibinfo{year}{2003}\natexlab{}.
\newblock \showarticletitle{Analysis of {TTL}-based {C}ache {N}etworks}. In
  \bibinfo{booktitle}{\emph{IEEE INFOCOM}}.
\newblock


\bibitem[\protect\citeauthoryear{Kang and Sung}{Kang and Sung}{1995}]%
        {kang95}
\bibfield{author}{\bibinfo{person}{S. Kang} {and} \bibinfo{person}{D. Sung}.}
  \bibinfo{year}{1995}\natexlab{}.
\newblock \showarticletitle{{T}wo-state {MMPP} {M}odelling of {ATM}
  {S}uperposed {T}raffic {S}treams {B}ased on {T}he {C}haracterisation of
  {C}orrelated {I}nterarrival {T}imes}. In \bibinfo{booktitle}{\emph{IEEE
  GLOBECOM}}.
\newblock


\bibitem[\protect\citeauthoryear{Kelly}{Kelly}{1997}]%
        {kelly97}
\bibfield{author}{\bibinfo{person}{Frank Kelly}.}
  \bibinfo{year}{1997}\natexlab{}.
\newblock \showarticletitle{Charging and {R}ate {C}ontrol for {E}lastic
  {T}raffic}.
\newblock \bibinfo{journal}{\emph{Transactions on Emerging Telecommunications
  Technologies}} \bibinfo{volume}{8}, \bibinfo{number}{1}
  (\bibinfo{year}{1997}), \bibinfo{pages}{33--37}.
\newblock


\bibitem[\protect\citeauthoryear{Kelly, Maulloo, and Tan}{Kelly
  et~al\mbox{.}}{1998}]%
        {kelly98}
\bibfield{author}{\bibinfo{person}{F.~P. Kelly}, \bibinfo{person}{A.~K.
  Maulloo}, {and} \bibinfo{person}{D.~K.H. Tan}.}
  \bibinfo{year}{1998}\natexlab{}.
\newblock \showarticletitle{Rate {C}ontrol for {C}ommunication {N}etworks:
  {S}hadow {P}rices, {P}roportional {F}airness and {S}tability}.
\newblock \bibinfo{journal}{\emph{Journal of the Operational Research society}}
  \bibinfo{volume}{49}, \bibinfo{number}{3} (\bibinfo{year}{1998}),
  \bibinfo{pages}{237--252}.
\newblock


\bibitem[\protect\citeauthoryear{Li, Shakkottai, C.~S.~Lui, and Subramanian}{Li
  et~al\mbox{.}}{2018}]%
        {jslv17}
\bibfield{author}{\bibinfo{person}{J. Li}, \bibinfo{person}{S. Shakkottai},
  \bibinfo{person}{J. C.~S.~Lui}, {and} \bibinfo{person}{V. Subramanian}.}
  \bibinfo{year}{2018}\natexlab{}.
\newblock \showarticletitle{{Accurate {L}earning or {F}ast {M}ixing? {D}ynamic
  {A}daptability of {C}aching {A}lgorithms}}.
\newblock \bibinfo{journal}{\emph{IEEE Journal on Selected Areas in
  Communications}} (\bibinfo{year}{2018}).
\newblock


\bibitem[\protect\citeauthoryear{Ma and Towsley}{Ma and Towsley}{2015}]%
        {ma15}
\bibfield{author}{\bibinfo{person}{R.T. Ma} {and} \bibinfo{person}{D.
  Towsley}.} \bibinfo{year}{2015}\natexlab{}.
\newblock \showarticletitle{{C}ashing in on {C}aching: {O}n-demand {C}ontract
  {D}esign with {L}inear {P}ricing}. In \bibinfo{booktitle}{\emph{CoNext}}.
\newblock


\bibitem[\protect\citeauthoryear{Meyn and Tweedie}{Meyn and Tweedie}{2009}]%
        {MeyTwe09}
\bibfield{author}{\bibinfo{person}{Sean~P Meyn} {and}
  \bibinfo{person}{Richard~L Tweedie}.} \bibinfo{year}{2009}\natexlab{}.
\newblock \bibinfo{booktitle}{\emph{{Markov {C}hains and {S}tochastic
  {S}tability}}}.
\newblock \bibinfo{publisher}{Cambridge University Press}.
\newblock


\bibitem[\protect\citeauthoryear{Panigrahy, Li, and Towsley}{Panigrahy
  et~al\mbox{.}}{2017a}]%
        {nitishjian17}
\bibfield{author}{\bibinfo{person}{N.~K. Panigrahy}, \bibinfo{person}{J. Li},
  {and} \bibinfo{person}{D. Towsley}.} \bibinfo{year}{2017}\natexlab{a}.
\newblock \showarticletitle{Hit Rate vs. Hit Probability Based Cache Utility
  Maximization}. In \bibinfo{booktitle}{\emph{ACM MAMA}}.
\newblock


\bibitem[\protect\citeauthoryear{Panigrahy, Li, Zafari, Towsley, and
  Yu}{Panigrahy et~al\mbox{.}}{2017b}]%
        {nitishjianfaheem17}
\bibfield{author}{\bibinfo{person}{N.~K. Panigrahy}, \bibinfo{person}{J. Li},
  \bibinfo{person}{F. Zafari}, \bibinfo{person}{D. Towsley}, {and}
  \bibinfo{person}{P. Yu}.} \bibinfo{year}{2017}\natexlab{b}.
\newblock \showarticletitle{{Optimizing {T}imer-based {P}olicies for {G}eneral
  {C}ache {N}etworks}}.
\newblock \bibinfo{journal}{\emph{Arxiv preprint arXiv:1711.03941}}
  (\bibinfo{year}{2017}).
\newblock


\bibitem[\protect\citeauthoryear{Paxson and Floyd}{Paxson and Floyd}{1995}]%
        {paxson95}
\bibfield{author}{\bibinfo{person}{Vern Paxson} {and} \bibinfo{person}{Sally
  Floyd}.} \bibinfo{year}{1995}\natexlab{}.
\newblock \showarticletitle{{W}ide-{A}rea {T}raffic: {T}he {F}ailure of
  {P}oisson {M}odeling}.
\newblock \bibinfo{journal}{\emph{IEEE/ACM Transactions on Networking}}
  \bibinfo{volume}{3}, \bibinfo{number}{3} (\bibinfo{year}{1995}),
  \bibinfo{pages}{226--244}.
\newblock


\bibitem[\protect\citeauthoryear{Rodriguez, Spanner, and Biersack}{Rodriguez
  et~al\mbox{.}}{2001}]%
        {rodriguez01}
\bibfield{author}{\bibinfo{person}{P. Rodriguez}, \bibinfo{person}{C. Spanner},
  {and} \bibinfo{person}{E.~W. Biersack}.} \bibinfo{year}{2001}\natexlab{}.
\newblock \showarticletitle{{A}nalysis of {W}eb {C}aching {A}rchitectures:
  {H}ierarchical and {D}istributed {C}aching}.
\newblock \bibinfo{journal}{\emph{IEEE/ACM Transactions on Networking}}
  (\bibinfo{year}{2001}).
\newblock


\bibitem[\protect\citeauthoryear{Smith and Coit}{Smith and Coit}{1996}]%
        {smith96}
\bibfield{author}{\bibinfo{person}{Alice~E Smith} {and}
  \bibinfo{person}{David~W Coit}.} \bibinfo{year}{1996}\natexlab{}.
\newblock \bibinfo{booktitle}{\emph{Evolutionary {C}omputation}}.
\newblock \bibinfo{publisher}{Institute of Physics Publishing and Cambridge
  University Press}.
\newblock


\bibitem[\protect\citeauthoryear{Srikant and Ying}{Srikant and Ying}{2013}]%
        {srikant13}
\bibfield{author}{\bibinfo{person}{R. Srikant} {and} \bibinfo{person}{Lei
  Ying}.} \bibinfo{year}{2013}\natexlab{}.
\newblock \bibinfo{booktitle}{\emph{Communication {N}etworks: an
  {O}ptimization, {C}ontrol, and {S}tochastic {N}etworks {P}erspective}}.
\newblock \bibinfo{publisher}{Cambridge University Press}.
\newblock


\bibitem[\protect\citeauthoryear{Weinberg}{Weinberg}{2016}]%
        {weinberg16}
\bibfield{author}{\bibinfo{person}{G.V. Weinberg}.}
  \bibinfo{year}{2016}\natexlab{}.
\newblock \showarticletitle{Kullback {L}eibler {D}ivergence and the {P}areto
  {E}xponential {A}pproximation}.
\newblock \bibinfo{journal}{\emph{SpringerPlus 5}} (\bibinfo{year}{2016}).
\newblock


\bibitem[\protect\citeauthoryear{Zerfos, Srivatsa, Yu, Dennerline, Franke, and
  Agrawal}{Zerfos et~al\mbox{.}}{2013}]%
        {zerfos13}
\bibfield{author}{\bibinfo{person}{P. Zerfos}, \bibinfo{person}{M. Srivatsa},
  \bibinfo{person}{H. Yu}, \bibinfo{person}{D. Dennerline}, \bibinfo{person}{H.
  Franke}, {and} \bibinfo{person}{D. Agrawal}.}
  \bibinfo{year}{2013}\natexlab{}.
\newblock \showarticletitle{{P}latform and {A}pplications for {M}assive-scale
  {S}treaming {N}etwork {A}nalytics}.
\newblock \bibinfo{journal}{\emph{IBM Journal for Research and Development:
  Special Edition on Massive Scale Analytics}} \bibinfo{volume}{57},
  \bibinfo{number}{136} (\bibinfo{year}{2013}), \bibinfo{pages}{1--11}.
\newblock


\bibitem[\protect\citeauthoryear{Zink, Suh, Gu, and Kurose}{Zink
  et~al\mbox{.}}{2008}]%
        {zink08}
\bibfield{author}{\bibinfo{person}{M. Zink}, \bibinfo{person}{K. Suh},
  \bibinfo{person}{Y. Gu}, {and} \bibinfo{person}{J. Kurose}.}
  \bibinfo{year}{2008}\natexlab{}.
\newblock \showarticletitle{Watch {G}lobal, {C}ache {L}ocal: {Y}ou{T}ube
  {N}etwork {T}raffic at a {C}ampus {N}etwork: {M}easurements and
  {I}mplications}. In \bibinfo{booktitle}{\emph{Electronic Imaging}}.
\newblock


\end{thebibliography}

\end{document}